\documentclass[xcolor=svgnames,11pt]{amsart}

\newtheorem{theorem}{Theorem}[section]
\newtheorem{lemma}[theorem]{Lemma}

\newtheorem{proposition}[theorem]{Proposition}

\theoremstyle{definition}
\newtheorem{remark}[theorem]{Remark}
\newtheorem{definition}[theorem]{Definition}
\newtheorem{example}[theorem]{Example}

\theoremstyle{plain}
\newtheorem{assumption}{Assumption}

\numberwithin{equation}{section}

\renewcommand{\(}{\left(}
\renewcommand{\)}{\right)}
\renewcommand{\tilde}{\widetilde}

\DeclareMathOperator*{\argmax}{\arg\max}
\DeclareMathOperator*{\argmin}{\arg\min}

\usepackage{amsmath}
\usepackage{amsfonts}
\usepackage{amssymb,amsthm}
\usepackage{chngcntr}
\usepackage{nicefrac}
\usepackage{framed}
\usepackage{color}
\usepackage{pstricks}
\usepackage{ipa}
\usepackage{textcomp}
\usepackage{url}
\usepackage{stmaryrd}
\usepackage{mathabx}
\usepackage{xcolor}
\usepackage{graphicx}
\usepackage{subfig}
\usepackage[left=2.2cm,top=2.8cm,right=2.2cm,bottom=2.8cm,head=3cm,headsep=1.3cm, foot=3cm]{geometry}
\usepackage[bottom,flushmargin,hang,multiple]{footmisc}
\usepackage{comment}
\usepackage{setspace}
\usepackage{xr-hyper}
\usepackage{hyperref}

\usepackage{enumitem}
\usepackage{blkarray}

\definecolor{darkblue}{rgb}{0.1,0.1,0.7}
\definecolor{darkred}{rgb}{0.9,0.1,0.1}

\hypersetup{colorlinks, allcolors= darkblue}

\DeclareMathOperator*{\val}{val}
\DeclareMathOperator*{\WD}{WD}
\DeclareMathOperator*{\leqHR}{\leq_{hr}}
\DeclareMathOperator*{\leqLR}{\leq_{lr}}
\DeclareMathOperator*{\ce1}{CE_1}
\DeclareMathOperator*{\CE2}{CE_2}

\newlist{thmcases}{enumerate}{1}
\setlist[thmcases]{
  label=\textbf{\upshape Case~\arabic*.},
  leftmargin=*,
  ref={\arabic*}}

\newlist{proofcases}{enumerate}{1}
\setlist[proofcases]{
  label=\textbf{\upshape Case~\theproposition.\arabic*.},
  leftmargin=*,
  ref={\theproposition.\arabic*}}

\newlist{subcases}{enumerate}{1}
\setlist[subcases]{
  label=\textbf{\upshape Subcase~\arabic*.},
  leftmargin=*,
  ref={\arabic*}}

\newenvironment{rcases}
  {\left.\begin{aligned}}
  {\end{aligned}\right\rbrace}

\makeatletter

\newcommand{\Rmnum}[1]{\expandafter\@slowromancap\romannumeral #1@}
\makeatother

\begin{document}

\title[Optimal Insurance under Maxmin Expected Utility]{Optimal Insurance under\vspace{0.1cm}\\Maxmin Expected Utility\vspace{0.6cm}}

\author[Corina Birghila, Tim J.\ Boonen, and Mario Ghossoub]{Corina Birghila \vspace{0.1cm} \\ University of Waterloo \vspace{0.7cm}\\Tim J.\ Boonen \vspace{0.1cm} \\ University of Amsterdam \vspace{0.7cm} \\ Mario Ghossoub  \vspace{0.1cm} \\ University of Waterloo \vspace{1.4cm}  \\ This draft: \today \vspace{0.7cm}}

\address{{\bf Corina Birghila}: University of Waterloo -- Department of Statistics and Actuarial Science -- 200 University Ave.\ W.\ -- Waterloo, ON, N2L 3G1 -- Canada}
\email{\href{mailto:corina.birghila@uwaterloo.ca}{corina.birghila@uwaterloo.ca}\vspace{0.4cm}}

\address{{\bf Tim J.\ Boonen}:  Amsterdam School of Economics -- University of Amsterdam -- Roetersstraat 11, 1018 WB -- Amsterdam -- The Netherlands}
\email{\href{mailto:t.j.boonen@uva.nl}{t.j.boonen@uva.nl}\vspace{0.4cm}}

\address{{\bf Mario Ghossoub}: University of Waterloo -- Department of Statistics and Actuarial Science -- 200 University Ave.\ W.\ -- Waterloo, ON, N2L 3G1 -- Canada}
\email{\href{mailto:mario.ghossoub@uwaterloo.ca}{mario.ghossoub@uwaterloo.ca}\vspace{0.4cm}}

\thanks{ \vspace{0.15cm} }

\thanks{ \vspace{0.15cm}}

\keywords{Optimal Insurance, Ambiguity, Multiple Priors, Maxmin-Expected Utility, Heterogeneous Beliefs.\vspace{0.4cm}}

\thanks{Mario Ghossoub acknowledges financial support from the Natural Sciences and Engineering Research Council of Canada (NSERC Grant No.\ 2018-03961).}

%====================================================================================
%====================================================================================
%====================================================================================

\begin{abstract}
We examine a problem of demand for insurance indemnification, when the insured is sensitive to ambiguity and behaves according to the Maxmin-Expected Utility model of Gilboa and Schmeidler \cite{Gilboa-Schmeidler1989}, whereas the insurer is a (risk-averse or risk-neutral) Expected-Utility maximizer. We characterize optimal indemnity functions both with and without the customary \textit{ex ante no-sabotage} requirement on feasible indemnities, and for both concave and linear utility functions for the two agents. This allows us to provide a unifying framework in which we examine the effects of the \textit{no-sabotage} condition, marginal utility of wealth, belief heterogeneity, as well as ambiguity (multiplicity of priors) on the structure of optimal indemnity functions. In particular, we show how the singularity in beliefs leads to an optimal indemnity function that involves full insurance on an event to which the insurer assigns zero probability, while the decision maker assigns a positive probability. We examine several illustrative examples, and we provide numerical studies for the case of a Wasserstein and a R{\'e}nyi ambiguity set.
\end{abstract}

\maketitle

\newpage
%====================================================================================
%====================================================================================
%====================================================================================

\section{Introduction}
A foundational, and by now folkloric problem in economic theory and the theory of risk exchange is the problem of demand for insurance indemnification. Specifically, an insurance buyer, or decision maker (DM), faces a random insurable loss, against which she seeks coverage through the purchase of an insurance policy. The price of this coverage is termed the \textit{policy premium}, and the insurance pricing functional is called the \textit{premium principle}. The premium principle is assumed to be known by the DM, and to be given by the certainty equivalent of the insurer's utility. Although this is a classical problem, it has traditionally been confined to the accustomed framework of Expected-Utility Theory (EUT), going back to the pioneering work of Arrow \cite{Arrow1963b,Arrow1971} and Mossin \cite{Mossin1968a}. With the impetus of the von Neumann-Morgenstern \cite{VNM} then newly minted theory of choice under uncertainty (the EUT), Arrow \cite{Arrow1963b} shows the optimality of deductible insurance (a zero indemnification below a fixed threshold of loss, and a linear indemnification above) in an EUT framework, when the DM is risk-averse, the insurer is risk-neutral, and the two parties have the same beliefs about the underlying loss probability distribution. This foundational result has subsequently been extended in multiple directions. For instance, Raviv \cite{Raviv1979} proposes a bargaining approach where the insurer is risk-averse; Dana and Scarsini \cite{DanaScarsini2007} introduce background risk of the DM, and Cummins and Mahul \cite{CumminsMahul2004} study limited liability via an exogenous upper limit on the indemnity. Numerous other extensions and modifications of the classical insurance demand framework have been proposed, and we refer to Gollier \cite{Gollier2013} and Schlesinger \cite{Schlesinger2000} for surveys thereof.

\vspace{0.2cm}

\subsection*{Ambiguity in Insurance Demand}
The vast majority of this literature remains within the confines of the classical EUT, an entirely objective Bayesian approach to decision-making under uncertainty. Yet, ever since the major challenges to the foundations of EUT that the work of Allais \cite{allais53} and Ellsberg \cite{Ellsberg1961} has put forward, decision theory has been pulling away from parts of the axiomatic foundations of EUT, in favour of non-EU models that can rationalize behavior depicted by Allais \cite{allais53} and Ellsberg \cite{Ellsberg1961}, as well as other cognitive biases that are not captured by EUT. Arguably, one of the most important achievements of the modern theory of choice under uncertainty is the remarkable development spurred by the work of Ellsberg \cite{Ellsberg1961}, in the study of what came to be known as \textit{ambiguity}, or \textit{model uncertainty}. Two main approaches to the rationalization of attitudes toward ambiguity have been explored in the literature on axiomatic decision theory: the \textit{non-additive prior} approach, and the \textit{multiple additive priors} approach. These two approaches do intersect, but they are not equivalent. The first category is based on the seminal contributions of Yaari \cite{yaari} (Dual Theory, or DT), Quiggin \cite{quiggin82} (Rank-Dependent Expected-Utility, or RDEU), and Schmeidler \cite{schmeidler89} (Choquet-Expected Utility, or CEU), which encompasses the previous two models. The second category was initiated by Gilboa and Schmeidler \cite{Gilboa-Schmeidler1989} (Maxmin-Expected Utility, or MEU) and further refined by Ghirardato et al.\ \cite{Ghirardato} (the $\alpha$-maxmin model), Klibanoff et al.\ \cite{Klibanoffetal2005} (the KMM model), and Amarante \cite{Amarante2009b} who provides a unifying framework.

\vspace{0.2cm}

While the literature on non-EU preferences in risk-sharing or optimal insurance design problems is considerably thinner than the literature on risk-sharing with EU preferences, \textit{behavioral} preferences, and ambiguity in particular play an increasing role in this literature. Yet, Machina  \cite{MachinaINS2013} points out that the robustness of standard  optimal insurance results under situations of ambiguity is still very much an open question, despite a growing literature on the topic. For instance, Bernard et al.\ \cite{Bernardetal2015} and Xu et al.\  \cite{Xuetal2018} study RDEU preferences of the DM and risk-neutral EU preferences of the insurer, and they derive optimal insurance indemnities. Ghossoub \cite{Ghossoub2019b} extends the analysis to account for more general premium principles. Also within the first category of ambiguity representation as a non-additive prior, Jeleva \cite{Jeleva2000} considers the case of a DM who is a CEU-maximizer. In the second category of ambiguity representation as a collection of additive priors and an aggregation rule, Alary et al.\ \cite{AlaryGollierTreich203} and Gollier \cite{Gollier2014} consider the case of an ambiguity-averse DM, in the sense of KMM. However, they consider a finite state space and restrict the set of priors to have a given parametric form.

\vspace{0.2cm}

Despite its appeal, for its ability to provide a separation of the effect of ambiguity aversion from that of risk aversion, as well as for its capability to define the notion of ambiguity neutrality, the KMM model, as a model of ambiguity with multiple priors, is arguably not as intuitive or popular as the MEU model of Gilboa and Schmeidler \cite{Gilboa-Schmeidler1989}. The MEU model gives rise to decision-making problems that can be embedded into to a larger class of \textit{model uncertainty} problems, which lie at the core of the theory of \textit{distributionally robust optimization} (DRO). In this framework, a decision-making problem is often modelled via a maxmin formulation: the agent is uncertain about the underlying model (prior), and therefore formulates an objective function using a collection of (additive) priors, also referred to as the \textit{ambiguity set}. The agent then aims to maximize the objective under the worst-case model (e.g., Ben-Tal et al.\ \cite{BenTaiElGhaouiNemirovski2009}). However, the intuitiveness and wide popularity of the MEU model notwithstanding, there has surprisingly been no study of optimal insurance contracting when the DM is a MEU-maximizer, to the best of our knowledge. This paper fills this void. Specifically, we extend the classical setup and results in two ways: (i) the DM is endowed with MEU preferences; and (ii) the insurer is not necessarily risk-neutral (that is, the premium principle is not necessarily an expected-value premium principle). The main objective of this paper is to determine the shape of the optimal insurance indemnity when the DM is sensitive to ambiguity and behaves according to MEU.

\vspace{0.2cm}

\subsection*{This Paper's Contribution}
In the literature on optimal insurance contracting, a popular assumption is the \textit{no-sabotage condition}, typically imposed as an \textit{ex ante} condition of feasibility of insurance indemnities. This condition stipulates that the insured (ceded) risk and the retained risk are comonotonic (they are both nondecreasing functions of the underlying loss). Under the no-sabotage condition, the DM has no incentive to under-report the underlying loss, nor does the DM have an incentive to create incremental losses. This condition is also sometimes referred to as \textit{incentive compatibility}, or a condition that avoids \textit{ex post} moral hazard; and it is further studied by Huberman et al.\ \cite{Hubermanetal1983} and Carlier and Dana \cite{CarlierDana2003b}\footnote{We refer to Carlier and Dana \cite{CarlierDana2003b} for a discussion of various notions of \textit{ex ante} admissible contracts.}. In this paper, we characterize optimal insurance contracts under MEU, both with and without the no-sabotage condition. In doing so, this paper sheds light on the consequences of the no-sabotage assumption on the construction of optimal insurance indemnities, in the presence of belief heterogeneity as well as multiple priors for the DM.

\vspace{0.2cm}

Our main results are the following. First, we examine in Section~\ref{section3} the general case in which the insurer is a risk-averse EU-maximizer, and the DM is a MEU-maximizer with a concave utility function, displaying decreasing marginal utility of wealth. We provide an implicit characterization of optimal indemnity functions, both with and without the no-sabotage condition on feasible indemnities. Optimal indemnity functions can be formulated as a solution to an ordinary differential equation, which can then be easily solved numerically in practice.

\vspace{0.2cm}

Second, as a special case of the above setting, we examine in Section~\ref{section4} the situation in which the insurer is risk-neutral, and hence the premium principle is an expected-value premium principle, as is commonly assumed in the literature (e.g., Bernard et al.\ \cite{Bernardetal2015} and Xu et al.\ \cite{Xuetal2018}). In this case, we provide an explicit, closed-form characterization of optimal indemnity functions in the absence of the no-sabotage condition, and an implicit characterization in the presence of the no-sabotage condition. In particular, by doing so, we provide in both cases (with and without the no-sabotage condition) a crisp depiction of the effect of heterogeneity in beliefs between the two parties, showing how the singularity in beliefs leads to an optimal indemnity function that involves full insurance on an event to which the insurer assigns zero probability, but not the DM. This an important and intuitive feature of our optimal contracts in this case. We then illustrate these results with an example.

\vspace{0.2cm}

Third, as a further special case, we examine the situation in which both parties display constant marginal utility of wealth, that is, their utility functions are linear. We show that in this case, a particular type of layer insurance (also called tranching) is optimal if the no-sabotage condition is imposed, which we derive in closed form. Layer insurance contracts consist of a finite number of long and short positions on various stop-loss contracts of the underlying risk. If the no-sabotage condition is not imposed, the optimal indemnity makes use of a partition of the state space in three sets, providing no insurance for events in the first set, full insurance for events in the second set, and proportional insurance for events in the third set. Artzner et al.\ \cite{Artzneretal1999} (and Delbaen \cite{Delbaen2002}) show that the class of MEU preferences with linear utility is related to the class of coherent risk measures. Therefore, our analysis can be used to derive optimal insurance contracts when the DM is endowed with a general coherent risk measure. To the best of our knowledge, this has not been studied in the literature when the DM and seller have heterogeneous beliefs regarding the underlying probability distribution.\footnote{This is studied by Birghila and Pflug \cite{BirghilaPflug2019} in the context of homogeneous beliefs regarding the underlying probability distribution.} Layer indemnity contracts play an important role in insurance, as deductible or truncated deductible contracts are important examples thereof. Truncated deductibles or layer-type insurance indemnities are indeed commonly observed in reinsurance markets, where tranches of aggregate losses of the insurer are ceded to a reinsurer (e.g., Albrecher et al.\ \cite{AlbrecherBeirlantTeugels2017}).

\vspace{0.2cm}

Lastly, we examine in Section \ref{section5} some numerical examples to illustrate our results. By specifying the structure of the DM's ambiguity set $\mathcal{C}$, we are able to obtain explicitly the worst-case probability measure for the problems analyzed in Sections~\ref{section3} and \ref{section4}. In particular, we examine the special case in which the DM's set of priors forms a neighborhood around the insurer's probability measure. First, when the insured is risk-averse, $\mathcal{C}$ is a Wasserstein ambiguity set, and the insurer is risk-neutral, we are able to characterize the saddle point of the problem in Section~\ref{section4}. In this case, the optimal indemnity is a deductible contract, while the worst-case measure dominates the insurer's probability measure in the sense of first-order stochastic dominance. In the second example, we consider a general setting in which both participants are risk-averse, and $\mathcal{C}$ is the R\'enyi ambiguity set. In a discretized framework, we use a successive convex programming algorithm to solve the ordinary differential equation in Section~\ref{section3}. We then assess the influence of the ambiguity set on the optimal value. In particular, we show numerically that a larger ambiguity set yields a lower certainty equivalent of final wealth for the DM, but increases the willingness-to-pay for insurance. Moreover, for both examples, the impact of the no-sabotage condition on the feasible set of insurance indemnities is illustrated.

\vspace{0.2cm}

\subsection*{Other Related Literature}
This paper contributes to the literature on heterogeneity in beliefs between the DM and the insurer. Heterogeneity in beliefs has been studied recently in the context of optimal (re)insurance by Ghossoub \cite{Ghossoub2019c}, %\cite{Ghossoub2015b,Ghossoub2017,Ghossoub2019c}, Boonen \cite{Boonen2016},
Boonen and Ghossoub \cite{BoonenGhossoub2019b,BoonenGhossoub2020c}, %\cite{BoonenGhossoub2019b,BoonenGhossoub2020a,BoonenGhossoub2020c},
and Chi \cite{Chi2019}. All of these studies focus on unambiguous subjective preferences on the side of the DM (that is, a unique subjective prior on the state space), but they differ in the formulation of the objective function that is optimized.

\vspace{0.2cm}

Heterogeneous beliefs can arise for different reasons. For instance, in the Subjective Expected Utility (SEU) theory of De Finetti \cite{DeFinetti1937} and Savage \cite{Savage}, disagreements about (subjective) beliefs are a result of differences in preferences over alternatives. Moreover, in epistemic game theory, the celebrated Agreement Theorem of Aumann \cite{Aumann1976} states that if agents have common priors, then the Hars\'anyi Doctrine holds, that is, disagreements about probabilities result only from information asymmetry. Furthermore, disagreement about (posterior) beliefs can be a direct consequence of relaxing the controversial and heavily criticized common priors assumption (e.g., \cite{Aumann1998,Gul1998,Morris1995}).

\vspace{0.2cm}

On a practical level, an insurance buyer may have private information about the distribution of the insurable loss. The insurer may use another probability measure, based on an average historical distribution of losses in the insurer's portfolio over a relevant time frame. For instance, the price of insurance may seem low for insurance buyers facing relatively high risks. This observation forms the basis of adverse selection in insurance markets (e.g., \cite{ChiapporiSalanie2000,FinkelsteinPoterba2004}). Jeleva and Villeneuve \cite{JelevaVilleneuve2004} study  heterogeneous beliefs in an adverse selection model in an insurance market with two future states of the world. Moreover, heterogeneity in reference probabilities may also be driven by ambiguity on the side of the insurer, rather than the insurance buyer (e.g., \cite{AGP2015,HogarthKunreuther1989}). More precisely, the insurer may experience ambiguity about the underlying probability distribution, and hence use a pricing rule that would be deemed more prudent than one based on the probability distribution used by the insurance buyer.

\vspace{0.2cm}

By explicitly incorporating model uncertainty into the problem formulation via a set of priors $\mathcal{C}$, the present paper also falls within the DRO framework. In this perspective, insurance contracts can be seen as saddle points of a DRO problem. The benefit of this technique is twofold. First, the worst-case approach ensures that the optimal decision is not sensitive to possible model misspecification. Second, in many situations, there exist tractable reformulations or algorithms to solve these distributionally robust models, even when the corresponding non-ambiguous problem (that is, when there is a unique prior) cannot be efficiently solved. The idea of incorporating multiple models in the decision-making process dates back to the fundamental work of Scarf \cite{Scarf1958} in the inventory management applications. He considers a robust formulation of the newsvendor problem, where the optimal strategy is constructed over all possible demand functions with known mean and variance. This initial idea is further developed in the work of Ben-Tal et al.\ \cite{BenTaiElGhaouiNemirovski2009} and Bertsimas and Sim \cite{BertsimasSim2004}, among others. A key concept in DRO is the structure of the set of priors, known here as the ambiguity set. Clearly, the choice of the ambiguity set $\mathcal{C}$ influences the worst-case model, and thus the optimal decision, while it also facilitates a tractable reformulation and efficient algorithm implementation. The existing literature has focused so far on two types of ambiguity sets: those built using the moment-based approach (e.g., \cite{DelageYe2010,Scarf1958,ZymlerKuhnRustem2013}), and those built using the statistical distance-based approach (such as the Kullback-Leibler divergence in \cite{CalafioreElGhaoui2006}, the $L_1$-ball in \cite{Thiele2008}, or the Wasserstein distance in \cite{EsfahaniKuhn2018}). Each such choice comes with useful structural properties, but also with shortcomings that need to be dealt with. Ultimately, it is the available set of observations and the type of application that would dictate a suitable choice of $\mathcal{C}$.

\vspace{0.2cm}

The rest of the paper is organized as follows. Section~\ref{section2} presents the setup of our problem together with the necessary background. In Section~\ref{section3}, we consider the case in which both the insurer and DM have concave utility functions, and the DM is a MEU-maximizer. We characterize optimal indemnity functions both in the presence and absence of the no-sabotage condition. Section~\ref{section4} considers the particular case of a risk-neutral insurer, and provides some illustrating examples. Section~\ref{section5} reports numerical illustrations, and Section~\ref{section6} concludes the paper. Some definitions and technical proofs are provided in Appendices~\ref{appendixProofpropCase2} to~\ref{appendixAlgorithm}.

\vspace{0.4cm}

%====================================================================================
%====================================================================================
%====================================================================================

\section{Setup and preliminaries\label{section2}}
Let $S$ be a nonempty collection of states of the world, and equip $S$ with a $\sigma$-algebra $\mathcal{G}$ of events. A DM is facing an insurable state-contingent loss represented by a random variable $X$ on the measurable space $(S, \mathcal{G})$, with values in the interval $[0,M]$, for some $M \in \mathbb{R}^+$. We denote by $\Sigma$ the sub-$\sigma$-algebra $\sigma\{X\}$ of $\mathcal{G}$ on $S$ generated by the random variable $X$.

\vspace{0.2cm}

Let $B\(\Sigma\)$ denote the vector space of all bounded, $\mathbb{R}$-valued, and $\Sigma$-measurable functions on $\(S, \Sigma\)$, and let $B^{+}\(\Sigma\)$ be its positive cone. When endowed with the supnorm $\| . \|_{sup}$, %\footnote{Any $Y \in B\(\Sigma\)$ is bounded, and its supnorm is defined by $\|Y\|_{sup} := \sup \{ \abs{Y(s)}: s \in S \} < +\infty$.},
$B\(\Sigma\)$  is a Banach space (e.g., \cite[IV.5.1]{Dunford}). By Doob's measurability theorem \cite[Theorem 4.41]{AliprantisBorder}, for any $Y \in B\(\Sigma\)$ there exists a bounded, Borel-measurable map $I : \mathbb{R} \rightarrow \mathbb{R}$ such that $Y = I \circ X$. Moreover, $Y \in B^{+}\(\Sigma\)$ if and only if the function $I$ is nonnegative.

\vspace{0.2cm}

\begin{definition}
Two functions $Y_{1},Y_{2} \in B\(\Sigma\)$ are said to be comonotonic (resp., anti-comonotonic) if
\[
\Big [ Y_{1}\(s\) - Y_{1}\(s^{\prime}\) \Big ] \Big[ Y_{2}\(s\) - Y_{2}\(s^{\prime}\) \Big] \geq 0 \ \hbox{(resp., $\leq 0$)}, \hbox{ for all } s, s^{\prime} \in S.
\]
\label{Def:ComonotonicFunctions}
\end{definition}

\noindent For instance any $Y \in  B\(\Sigma\)$ is comonotonic and anti-comonotonic with any $c \in \mathbb{R}$. Moreover, if $Y_{1},Y_{2} \in B\(\Sigma\)$, and if $Y_{2}$ is of the form $Y_{2} = I \circ Y_{1}$, for some Borel-measurable function $I$, then $Y_{2}$ is comonotonic (resp., anti-comonotonic) with $Y_{1}$ if and only if the function $I$ is nondecreasing (resp., nonincreasing).

\vspace{0.3cm}

Let $ba\(\Sigma\)$ denote the linear space of all bounded finitely additive set functions on $\Sigma$, endowed with the usual mixing operations. When endowed with the total variation norm $\| . \|_{v}$, $ba\(\Sigma\)$ is a Banach space. By a classical result (e.g., \cite[IV.5.1]{Dunford}), $\(ba\(\Sigma\), \|.\|_{v}\)$ is isometrically isomorphic to the norm-dual of $B\(\Sigma\)$, via the duality $\textless \phi, \lambda \textgreater = \int \phi \ d \lambda, \ \forall \lambda \in ba\(\Sigma\), \ \forall \phi \in B\(\Sigma\)$. Consequently, we can endow $ba\(\Sigma\)$ with the weak$^{*}$ topology $\sigma\(ba\(\Sigma\), B\(\Sigma\)\)$.

\vspace{0.3cm}

Let $ca\(\Sigma\)$ denote the collection of all countably additive elements of $ba\(\Sigma\)$, and let $ca^{+}\(\Sigma\)$ denote its positive cone. Then $ca\(\Sigma\)$ is a $\|.\|_{v}$-closed linear subspace of $ba\(\Sigma\)$. Hence, $ca\(\Sigma\)$ is $\|.\|_{v}$-complete, i.e.\ $\(ca\(\Sigma\), \|.\|_{v}\)$ is a Banach space. Denote by
\[
ca^{+}_{1}\(\Sigma\) := \Big\{\mu \in ca^{+}\(\Sigma\): \mu(S) =1\Big\}
\]
the collection of probability measures on $\(S, \Sigma\)$. We shall endow $ca^{+}_{1}\(\Sigma\)$ with the weak$^{*}$ topology inherited from $ba\(\Sigma\)$.

\vspace{0.3cm}

For any $Y \in B(\Sigma)$ and $P \in ca^{+}_{1}\(\Sigma\)$, let $F_{Y,P}(t) := P\big(\lbrace s \in S : Y(s) \leq t\rbrace \big)$ denote the cumulative distribution function (cdf) of $Y$ with respect to the probability measure $P$, and let $F^{-1}_{Y,P}(t)$ denote the left-continuous inverse of the $F_{Y,P}$ (i.e., the quantile function of $Y$), defined by
\[
F^{-1}_{Y,P}(t) := \inf\big\lbrace z \in \mathbb{R} :\,  F_{Y,P}(z)\geq t\big\rbrace, \ \forall t \in \left[0, 1\right].
\]

\vspace{0.2cm}

\subsection{The DM's and the Insurer's Preferences}
The DM can purchase insurance against the random loss $X$ in a perfectly competitive insurance market, for a premium set by the insurer. In return for the premium payment, the DM is promised an indemnification against the realizations of $X$. An indemnity function is a random variable $Y = I(X)$ on $(S, \Sigma)$, for some bounded, Borel-measurable map $I: X(S) \rightarrow \mathbb{R}$, which pays off the amount $I(X(s)) \in \mathbb{R}$ in state of world $s \in S$, corresponding to a realization $X(s)$ of $X$. That is, we can identify the set of indemnity functions with a subset of $B\(\Sigma\)$.
%In order to mitigate the risk associated with $X$, the DM intends to purchase insurance for a premium $\Pi_0$, determined by insurer at the beginning of the contact. The insurance provides indemnification against the loss $X(s)$, for all $s\in S$. The indemnity function $Y$ is a random variable on $(S,\mathcal{F})$ given by $Y=I(X)$, for some bounded and Borel-measurable map $I:X(S)\rightarrow\mathbb{R}$.
%
%\vspace{0.2cm}
%
For each indemnity function $Y \in B\(\Sigma\)$, we define the corresponding retention function by $R := X -Y \in B\(\Sigma\)$. As the name suggests, $R$ is the retained loss after insurance indemnification.

\vspace{0.2cm}

The DM has a preference relation over insurance indemnification functions (or over wealth profiles) that admits a MEU representation $V^{\text{MEU}}: B(\Sigma) \rightarrow \mathbb{R}$ as in Gilboa and Schmeidler \cite{Gilboa-Schmeidler1989}, of the form
\begin{equation}
V^{\text{MEU}}(Z) := \min_{\mu \in \mathcal{C}} \int u(Z) \ d\mu, \ \ \forall \ Z \in B\(\Sigma\),
\label{MEUfunctional}
\end{equation}

\noindent where $u:\mathbb{R}\rightarrow\mathbb{R}$ is a concave utility function, and $\mathcal{C}$ is a (unique) weak$^*$-compact and convex subset of $ba^+_1(\Sigma)$. Moreover, we assume that the DM's preferences satisfy the Arrow-Villegas Monotone Continuity axiom as in Chateauneuf et al.\ \cite{Chateauneufetal2005a}, so that $\mathcal{C} \subset ca^+_1(\Sigma)$, i.e., all priors are countably additive. Additionally, the DM's utility function $u$ satisfies the following assumption.

%\vspace{0.2cm}
%
%We assume that the DM's preferences admit a representation in terms of a MEU-functional, with a concave utility function $u:\mathbb{R}\rightarrow\mathbb{R}$ and a convex and weak$^*$-compact set $\mathcal{C}\subset ca^+_1(\Sigma)$ of probability measures on $(S,\Sigma)$. The set $\mathcal{C}$, referred to as the ambiguity set in DRO, is given \textit{a priori}. Moreover, we will also assume that the insurer is a risk-averse EU maximizer, whose preferences induce a utility function $v:\mathbb{R}\rightarrow\mathbb{R}$ and a probability measure $Q$ on $(S,\Sigma)$. {\color{red}Both utility functions satisfy the following.

\vspace{0.2cm}

\begin{assumption} \label{assumptionsUtility}
The utility function $u:\mathbb{R} \longrightarrow \mathbb{R}$ is strictly increasing, concave and continuously differentiable.
\end{assumption}

%\vspace{0.2cm}
%
%\noindent Note that Assumption~\ref{assumptionsUtility} includes the case of a risk-neutral insurer, as in Arrow \cite{Arrow1971}. %A strict concavity assumption on utilities $u$ and $v$ reflects the risk-averse attitude of DM and insurer, respectively, and it will be imposed separately, whenever we analyze this scenario.

\vspace{0.2cm}

Let $W_0 \in \mathbb{R}^+$ be the DM's initial wealth. After purchasing insurance coverage for a premium $\Pi_0 >0$, the DM's terminal wealth is a random variable $W\in B(\Sigma)$ given by
\[
W(s):=W_0-X+Y-\Pi_0.
\]

\vspace{0.2cm}

The insurer's preference over $B\(\Sigma\)$ admits an EU representation $V^{\text{Ins}}: B(\Sigma) \rightarrow \mathbb{R}$ of the form
\begin{equation*}
V^{\text{Ins}}(Z) :=  \int v(Z) \ dQ, \ \ \forall \ Z \in B\(\Sigma\),
\end{equation*}

\noindent for a utility function $v:\mathbb{R}\rightarrow\mathbb{R}$ satisfying Assumption~\ref{assumptionsUtility} and a probability measure $Q \in ca^+_1(\Sigma)$.

\vspace{0.2cm}

The insurer has an initial wealth $W_0^{\text{Ins}}$, and faces an administration cost, often called an \textit{indemnification cost}, associated with the handling of an indemnity payment. As customary in the literature (e.g., Bernard et al.\ \cite{Bernardetal2015} and Xu et al.\ \cite{Xuetal2018}), we assume that for a given indemnity function $Y = I \circ X$, this indemnification cost is a proportional cost of the form $\rho Y$, for a given safety loading factor $\rho \geq 0$ specified exogenously and \textit{a priori}. Hence, the insurer's terminal wealth is the random variable $W^{\text{Ins}} \in B\(\Sigma\)$ given by
\[
W^{\text{Ins}}:= W_0^{\text{Ins}} - (1+\rho)Y + \Pi_0.
\]

\vspace{0.2cm}

\subsection{Admissible Indemnity Functions}

In Arrow's \cite{Arrow1963b} original formulation of the optimal insurance problem under EUT, an \textit{ex ante} condition of feasibility of indemnity schedules is the requirement that these be nonnegative and no larger than the realization of the loss in each state of the world. This is often referred to as the \textit{indemnity principle}, and it translates into the requirement that an admissible set of indemnities be restricted to those $Y\in B(\Sigma)$ that satisfy $0\leq Y\leq X$. We shall denote this set of indemnity functions by $\mathcal{I}$:
\begin{equation}\label{FeasibWithoutNSC}
\mathcal{I}:=\Big\{ Y=I\circ X\in B^+(\Sigma) :\, 0\leq I(x)\leq x, \, \forall x\in [0,M]\Big\}.
\end{equation}

\vspace{0.2cm}

A desirable property of optimal indemnities is that an indemnity function $Y = I \circ X$ and the corresponding retention function $R = X - Y$ be both nondecreasing functions of the loss $X$, that is both comonotonic with $X$ (and hence $Y$ and $R$ are comonotonic). Indeed, if $Y$ fails to be comonotonic with $X$, then the DM has an incentive to under-report the loss; whereas if $R$ fails to be comonotonic with $X$, then the DM has an incentive to create additional damage. These situations of \textit{ex post} moral hazard are not desirable, and one often seeks additional \textit{ex ante} conditions that would rule out such behavior from the DM. In the setting of Arrow \cite{Arrow1963b}, the optimal indemnity is a deductible contract of the form $Y = \max\(X-d,0\)$, for some $d \in \mathbb{R}^+$. For such contracts, both the indemnity and retention functions are comonotonic with the loss, and optimal indemnities are \textit{de facto} immune to the kind of \textit{ex post} moral hazard described above. However, outside of EUT, optimality of deductible contracts does not always hold, and optimal indemnities might suffer from the aforementioned type of moral hazard, as in Bernard et al.\ \cite{Bernardetal2015}.

\vspace{0.2cm}

In order to rule out \textit{ex post} moral hazard that might arise from a misreporting of the loss by the DM, an additional condition is often imposed \textit{ex ante} on the set of feasible indemnity schedules (as in Xu et al.\ \cite{Xuetal2018}). Such a condition is called the \textit{no-sabotage} condition, and it stipulates that admissible indemnity functions and the corresponding retention functions be comonotonic, hence resulting in the feasibility set $\hat{\mathcal{I}}$ given by:
\begin{equation*}
\hat{\mathcal{I}}:=\Big\{ \hat Y \in \mathcal{I} :\, \hat Y \ \hbox{and} \ \hat R = X- \hat Y \ \hbox{are comonotonic}\Big\}.
\end{equation*}

\noindent Since $Y \in \mathcal{I}$ is of the form $Y = I \circ X$, with $0\leq I(x)\leq x$ for all $x\in [0,M]$, we can write $\hat{\mathcal{I}}$ as
\begin{equation}\label{FeasibWithNSC}
\hat{\mathcal{I}}=\Big\{ \hat{Y}=\hat{I}\circ X \in B^+(\Sigma) :\, \hat{I}(0)=0, \, 0\leq \hat{I}(x_1)-\hat{I}(x_2)\leq x_1-x_2, \forall\, 0\leq x_2\leq x_1\Big\}.
\end{equation}

\vspace{0.2cm}

\noindent The no-sabotage condition is also sometimes referred to as \textit{incentive compatibility} by Xu et al.\ \cite{Xuetal2018}, and it is further studied by Huberman et al.\ \cite{Hubermanetal1983} and Carlier and Dana \cite{CarlierDana2003b}. The latter discuss various classes of \textit{ex ante} admissible contracts, as well as their implications of optimal indemnities.

\vspace{0.2cm}

\begin{remark}\label{CompactIhat}
Let $C[0,M]$ denote the set of all continuous functions on $[0,M]$ (and hence bounded), equipped with the supnorm $\Vert\cdot\Vert_{sup}$. Note that $\hat{\mathcal{I}}$ is a uniformly bounded subset of $C[0,M]$ consisting of Lipschitz-continuous functions $[0,M] \rightarrow [0,M]$, with common Lipschitz constant $K=1$. Therefore, $\hat{\mathcal{I}}$ is equicontinuous, and hence compact by the Arzel\`a-Ascoli theorem (e.g., \cite[Theorem IV.6.7]{Dunford}).
\end{remark}

\vspace{0.2cm}

In this paper, we will characterize optimal indemnity functions, both with and without the no-sabotage condition, in order to examine the impact of such an \textit{ex ante} requirement on feasible indemnity schedules. This will first be done in the general setting of a MEU-maximizing DM with a concave utility and an EU-maximizing insurer with concave utility (Section~\ref{section3}), and then in a setting where the insurer is risk-neutral (hence uses an expected-value premium principle).

%In this paper, we characterize optimal insurance contracts under MEU, both with and without the no-sabotage condition. In doing so, this paper sheds light on the consequences of the no-sabotage assumption on the construction of optimal insurance indemnities, in the presence of belief heterogeneity as well as multiple priors for the DM.

\vspace{0.4cm}

%====================================================================================
%====================================================================================
%====================================================================================

\section{Optimal Indemnity Functions}\label{section3}

In this section, we investigate the DM's problem of demand for insurance indemnification, when the DM is ambiguity-sensitive and has preferences %over insurance indemnification functions (or over wealth profiles)
admitting a MEU representation %$V^{\text{MEU}}: B(\Sigma) \rightarrow \mathbb{R}$
of the form given in eq.~\eqref{MEUfunctional}, whereas the insurer is a risk-averse EU-maximizer with a concave utility function $v$. We first examine in Section~\ref{section31} the class $\mathcal{I}$ of indemnities that are nonnegative and cannot exceed the loss $X$ (as defined in eq.~\eqref{FeasibWithoutNSC}), and we provide in Theorem \ref{propCase1V} a closed-form characterization of the optimal indemnity in this case.
We then consider in Section~\ref{section32} the class $\hat{\mathcal{I}}$ of indemnities that are such that both indemnity and retention functions are nondecreasing functions of the loss (as defined in eq.~\eqref{FeasibWithNSC}). In that case, Theorem \ref{propCase2V} provides an implicit characterization of the optimal indemnity function.

\vspace{0.2cm}

Let $\mathcal{F}$ denote the set of admissible indemnity functions, which could be either the set $\mathcal{I}$ defined in eq.~\eqref{FeasibWithoutNSC}, or the set $\hat{\mathcal{I}} \subset \mathcal{I}$ defined in eq.~\eqref{FeasibWithNSC}. For a given insurance budget $\Pi_0>0$ and a compact and convex set $\mathcal{C}$ of probability measures, the optimal indemnity function is obtained as the solution of the problem

\begin{equation}\label{optimV}
\left\{
\begin{aligned}
& \underset{Y\in \mathcal{F}}{\sup}\,\underset{P\in\mathcal{C}}{\inf}
& & \mathbb{E}_P [u( W_0-X+Y-\Pi_0) ]\\
& \text{s.t.}
& & \mathbb{E}_Q [v( W_0^{\text{Ins}}-(1+\rho)Y+\Pi_0) ] \geq v(W_0^{\text{Ins}}),
\end{aligned}\tag{$P$}
\right.
\end{equation}

\vspace{0.2cm}

\noindent where $\rho\geq0$ is a given safety loading factor. The constraint in~\eqref{optimV} is interpreted as the insurer's participation constraint.\footnote{All of this paper's results can be derived for any participation constraint of the form $\mathbb{E}_Q [v( W_0^{\text{Ins}}-(1+\rho)Y+\Pi_0) ] \geq k$, with $k\leq v(W_0^{\text{Ins}}+\Pi_0)$. To maintain a direct economic interpretation, we choose throughout the paper $k=v(W_0^{\text{Ins}})$, i.e., the insurer's reservation utility.} Observe that $\mathbb{E}_P[u(W_0-X+Y-\Pi_0)]\leq u(W_0-\Pi_0)$, for all $P\in\mathcal{C}$ and all $I\in\mathcal{F}$, and thus Problem~\eqref{optimV} is finite.

\vspace{0.2cm}

\begin{remark}\label{remDecomposition}
By the Lebesgue Decomposition Theorem, for any $P\in\mathcal{C}$ there are finite nonnegative countably additive measures $P_{ac}$ and $P_s$ on $(S,\Sigma)$ such that $P=P_{ac}+P_s$, where $P_{ac}\ll Q$ and $P_s\perp Q$. Hence, for each $P\in\mathcal{C}$, there exists some $A_P\in\Sigma$ and $h_P:S\rightarrow [0,\infty)$ such that $Q(S\setminus A_P)=P_s(A_P)=0$ and $h_P=dP_{ac}/dQ$. In particular, since $h_P$ is $\Sigma$-measurable, there exists a nonnegative Borel measurable function $\xi_P:\mathbb{R}_+\rightarrow\mathbb{R}_+$ such that $h_P=\xi_P\circ X$.
\end{remark}

\vspace{0.2cm}

\subsection{Without the No-Sabotage Condition}
\label{section31}

The solution $(Y^*,P^*)$ to Problem~\eqref{optimV}, when $\mathcal{F}=\mathcal{I}$ is given in the following result. 	

\begin{theorem}\label{propCase1V}
Suppose that the utility functions $u$ and $v$ satisfy Assumption~\ref{assumptionsUtility} and are, in addition, strictly concave and such that $\underset{x\rightarrow -\infty} \lim u'(x) = \underset{x\rightarrow -\infty}\lim v'(x) = +\infty$ and $\underset{x\rightarrow +\infty} \lim u'(x) = \underset{x\rightarrow +\infty} \lim v'(x)= 0$.\footnote{The limit conditions on $u$ and $v$ are the customary Inada \cite{Inada} conditions, often encountered in the literature.} Let $\mathcal{F}=\mathcal{I}$ as defined in eq.\ \eqref{FeasibWithoutNSC} be the set of admissible indemnity functions. Then there exists $P^*\in\mathcal{C}$ such that an optimal solution $Y^*\in\mathcal{I}$ of Problem~\eqref{optimV} is of the form:

\begin{equation}\label{YoptimCase1V}
Y^* = \tilde{Y}^*\mathrm{1}_{A\setminus A_{h^*}} + Y_{h^*}\mathrm{1}_{A_{h^*}} + X\mathrm{1}_{S\setminus A},
\end{equation}
where
\vspace{0.2cm}
\begin{enumerate}[label=(\alph*)]
\item $A\in\Sigma$ is such that $P^*=P^*_{ac}+P^*_s$, with $P^*_s(A)=Q(S\setminus A)=0$;
\vspace{0.2cm}
\item $h^*:S\rightarrow [0,\infty)$ is such that $h^*=dP^*_{ac}/dQ$;
\vspace{0.2cm}
\item $A_{h^*}:=\lbrace s\in A:\, h^*(s)=0\rbrace$;
\vspace{0.2cm}
\item $\tilde{Y}^*$ and $Y_{h^*}$ are of the form:
\vspace{0.2cm}
\begin{thmcases}
\item If $\lambda^*>0$, then $Y_{h^*}=0$ and $\tilde{Y}^*=\max\left[0, \min\left(X,Y^*_0\right)\right]$, where $Y^*_0$ solves
\[
u'(W_0-X(s)+Y(s)-\Pi_0)h^*(s)-\lambda^* (1+\rho)v'(W_0^{\text{Ins}}-(1+\rho)Y(s)+\Pi_0)=0, \, \forall s\in A\setminus A_{h^*};
\]
%\vspace{0.1cm}
\item If $\lambda^*=0$, then $\tilde{Y}^*=X$ and $Y_{h^*}$ solves
\[
\mathbb{E}_Q [v( W_0^{\text{Ins}}-(1+\rho)Y\mathrm{1}_{A_{h^*}} +\Pi_0)]= v(W_0^{\text{Ins}})-\mathbb{E}_Q [v( W_0^{\text{Ins}}-(1+\rho)X\mathrm{1}_{A\setminus {A_{h^*}} }+\Pi_0)];
\]
\end{thmcases}
\vspace{0.2cm}
\item $\lambda^*\in\mathbb{R}_+$ is such that $\mathbb{E}_Q [v( W_0^{\text{Ins}}-(1+\rho)Y^*+\Pi_0) ]= v(W_0^{\text{Ins}})$.
\end{enumerate}
\end{theorem}

\vspace{0.2cm}

\begin{proof}
Define the set $\mathcal{I}_0:=\big\lbrace Y\in\mathcal{I}:\, \mathbb{E}_Q [v( W_0^{\text{Ins}}-(1+\rho)Y+\Pi_0) ] \geq v(W_0^{\text{Ins}})\big\rbrace$. Observe that for $Y_1$, $Y_2\in\mathcal{I}$ and $\alpha\in (0,1)$, we have $\tilde{Y}:=\alpha Y_1+ (1-\alpha)Y_2\in \mathcal{I}$ and
\[
\mathbb{E}_Q[v(W_0^{\text{Ins}}-(1+\rho)\tilde{Y}+\Pi_0)] \geq \alpha \mathbb{E}_Q[v(W_0^{\text{Ins}}-(1+\rho)Y_1+\Pi_0)]  + (1-\alpha)\mathbb{E}_Q[v(W_0^{\text{Ins}}-(1+\rho)Y_2+\Pi_0)] \geq v(W_0^{\text{Ins}}),
\]
and thus the set $\mathcal{I}_0$ is convex. Moreover, the objective function in \eqref{optimV} is continuous and concave in $Y$ and continuous and linear in $P$, while $\mathcal{C}$ is weak$^*$-compact and convex set.  Hence the Sion's Minimax Theorem states that there exists a saddle point $(Y^*,P^*)\in\mathcal{I}_0\times\mathcal{C}$ such that
\begin{equation*}
\begin{split}
\sup_{Y\in\mathcal{I}_0}\,\inf_{P\in\mathcal{C}}\mathbb{E}_P [u(W_0-X+Y-\Pi_0)] & = \sup_{Y\in\mathcal{I}_0}\,\min_{P\in\mathcal{C}}\mathbb{E}_P [u(W_0-X+Y-\Pi_0)]  \\
& = \min_{P\in\mathcal{C}}\,\sup_{Y\in\mathcal{I}_0}\mathbb{E}_P [u(W_0-X+Y-\Pi_0)].
\end{split}
\end{equation*}
For $P^*\in\mathcal{C}$, to characterize the optimal indemnity $Y^*$, we focus on the following inner problem:
\begin{equation}\label{optimVc1}
\begin{aligned}
& \underset{Y\in\mathcal{I}}{\sup}
& & \mathbb{E}_{P^*}[u(W_0-X+Y-\Pi_0)]\\
& \text{s.t.}
& & \mathbb{E}_Q [v(W_0^{\text{Ins}}-(1+\rho)Y+\Pi_0)]\geq v(W_0^{\text{Ins}}).
\end{aligned}
\end{equation}

\vspace{0.2cm}

Problem~\eqref{optimVc1} is a convex optimization problem, since the constraint can be equivalently written as $\mathbb{E}_Q [v_1(W_0^{\text{Ins}}-(1+\rho)Y+\Pi_0)]\leq v_1(W_0^{\text{Ins}})$, where $v_1:=-v$ is a convex utility function. For $P^*\in\mathcal{C}$, let $A:=A_{P^*}$ and $h^*:=h_{P^*}$ be as in Remark~\ref{remDecomposition}, and consider the following two problems:
\begin{equation}\label{optimV1c1}
\begin{split}
\sup_{Y\in\mathcal{I}} \bigg\lbrace \int_{S\setminus A} u(W_0-X+Y-\Pi_0)dP^*_s : & \, 0\leq Y\mathrm{1}_{S\setminus A}\leq X\mathrm{1}_{S\setminus A}, \\
&  \int_{S\setminus A} v(W_0^{\text{Ins}}-(1+\rho)Y+\Pi_0)\,dQ =0 \bigg\rbrace.
\end{split}
\end{equation}

\begin{equation}\label{optimV2c1}
\sup_{Y\in\mathcal{I}} \bigg\lbrace \int_A u(W_0-X+Y-\Pi_0)h^*\,dQ :\, \int v(W_0^{\text{Ins}}-(1+\rho)Y+\Pi_0)\,dQ\geq v(W_0^{\text{Ins}})\bigg\rbrace.
\end{equation}

\vspace{0.2cm}

Observe that $\overline{Y}:=X$ is a feasible solution for Problem~\eqref{optimV1c1} and it holds that
\[
\int_{S\setminus A} u(W_0-X+\overline{Y}-\Pi_0)dP^*_s = u(W_0-\Pi_0)P^*_s(S\setminus A) \geq \int_{S\setminus A} u(W_0-X+Y-\Pi_0)dP^*_s,
\]
for any feasible solution $Y$ for Problem~\eqref{optimV1c1}. Hence $\overline{Y}=X$ is optimal for~\eqref{optimV1c1}.

\vspace{0.2cm}

Now, let $Y_1^*\in\mathcal{I}$ be an optimal solution for Problem~\eqref{optimV2c1}. We claim that $Y^*:=Y_1^*\mathrm{1}_A + X\mathrm{1}_{S\setminus A}$ is optimal for Problem~\eqref{optimVc1}. To see this, we remark that
\[
\int_S v(W_0^{\text{Ins}}-(1+\rho)Y^*+\Pi_0)\,dQ = \int_A v(W_0^{\text{Ins}}-(1+\rho)Y_1^*+\Pi_0)\,dQ \geq v(W_0^{\text{Ins}}),
\]
where the last inequality follows from the feasibility of $Y_1^*$ for~\eqref{optimV2c1}. Hence $Y^*$ is feasible for Problem~\eqref{optimVc1}. The optimality of $Y^*$ is then derived similar to \cite[Lemma C.6]{Ghossoub2019c}.

\vspace{0.2cm}

Next, we focus on the optimal indemnity $Y_1^*$ that solves Problem~\eqref{optimV2c1}. The associated Lagrange function is
\begin{equation*}
\mathcal{L}(Y_1,\lambda)  = \int_A \big[ u(W_0-X(s)+Y_1(s)-\Pi_0)h^*(s) + \lambda v(W_0^{\text{Ins}}-(1+\rho)Y_1(s)+\Pi_0) \big]\,dQ(s) - \lambda v(W_0^{\text{Ins}}),
\end{equation*}
where $\lambda\in\mathbb{R}_+$ is the Lagrange multiplier. As the domain $\mathcal{I}$ of $Y_1$ is convex and $\mathcal{L}(Y_1,\lambda)$ is continuous and concave in $Y_1$ and linear in $\lambda$, the strong duality holds, i.e.,
\begin{equation*}
\val(\mathcal{L}):=\sup_{Y_1\in\mathcal{I}}\inf_{\lambda\in\mathbb{R}_+}\mathcal{L}(Y_1,\lambda) =\inf_{\lambda\in\mathbb{R}_+}\sup_{Y_1\in\mathcal{I}}\mathcal{L}(Y_1,\lambda),
\end{equation*}
where the optimal value $\val(\mathcal{L})$ of Problem~\eqref{optimV2c1} is finite, since \eqref{optimV} is finite. For fixed $\lambda\in\mathbb{R}_+$, a necessary and sufficient condition for $Y_1^*\in\mathcal{I}$ to be the optimal solution of Problem~\eqref{optimV2c1} is
\begin{equation}\label{eqYoptimV1c1}
\lim_{\theta\rightarrow 0^+}\mathcal{L}'((1-\theta)Y_1^*+\theta Y_1)\leq 0,\,  \forall Y_1\in\mathcal{I}.
\end{equation}
By direct computation, \eqref{eqYoptimV1c1} becomes
\begin{equation}\label{eqY1optimV1c1}
\int_A [u'(W_0-X+Y_1^*-\Pi_0)h^*-\lambda (1+\rho)v'(W_0^{\text{Ins}}-(1+\rho)Y^*+\Pi_0)](Y_1-Y_1^*)\, dQ\leq 0,\, \forall Y_1\in\mathcal{I}.
\end{equation}

\noindent Define the following sets, depending on Lagrange multiplier $\lambda$:
\vspace{0.2cm}
\[
\left
  \{
  \begin{aligned}
	  A^+_\lambda & :=\lbrace s\in A : u'(W_0-X(s)+Y_1^*(s)-\Pi_0)h^*(s)-\lambda (1+\rho)v'(W_0^{\text{Ins}}-(1+\rho)Y_1^*(s)+\Pi_0)>0\rbrace, \\
    A^0_\lambda & :=\lbrace s\in A: u'(W_0-X(s)+Y_1^*(s)-\Pi_0)h^*(s)-\lambda (1+\rho)v'(W_0^{\text{Ins}}-(1+\rho)Y_1^*(s)+\Pi_0)=0\rbrace, \\
    A^-_\lambda & :=\lbrace s\in A: u'(W_0-X(s)+Y_1^*(s)-\Pi_0)h^*(s)-\lambda (1+\rho)v'(W_0^{\text{Ins}}-(1+\rho)Y_1^*(s)+\Pi_0)<0\rbrace.
  \end{aligned}\right.
\]

\vspace{0.2cm}

\noindent First, observe that on $A^+_\lambda$ and $A^-_\lambda$, condition~\eqref{eqY1optimV1c1} holds for all $Y_1\in\mathcal{I}$ only if
\begin{equation}\label{eqYAminusAplus}
Y_1^*\mathrm{1}_{A^+_\lambda}=X\mathrm{1}_{A^+_\lambda} \text{ and } Y_1^*\mathrm{1}_{A^-_\lambda}=0.
\end{equation}

\vspace{0.2cm}

\noindent Next, define the set $A_{h^*}:=\lbrace s\in A:\, h^*(s)=0\rbrace$. To obtain the structure of $Y_1^*$ in \eqref{optimV2c1}, we distinguish the following cases, depending on $\lambda$.

\vspace{0.2cm}

\begin{proofcases}[noitemsep,wide=0pt, leftmargin=\dimexpr\labelwidth + 2\labelsep\relax]
\item If $\lambda>0$, then $A_{h^*}\subseteq A^-_\lambda$ and thus $Y_1^*\mathrm{1}_{A_{h^*}}=0$. On $A^0_\lambda\setminus A_{h^*}$, $Y_1^*$ satisfies the following condition:
\begin{equation}\label{eqY0}
u'(W_0-X(s)+Y_1^*(s)-\Pi_0)h^*(s)-\lambda (1+\rho)v'(W_0^{\text{Ins}}-(1+\rho)Y_1^*(s)+\Pi_0)=0.
\end{equation}
\noindent Let $Y^*_0$ be the solution of~\eqref{eqY0}. In the view of~\eqref{eqYAminusAplus}, $Y_1^*$ is thus $Y_1^*\mathrm{1}_{A\setminus A_{h^*}} =\max\left[0, \min\left(X,Y^*_0\right)\right]\mathrm{1}_{A\setminus A_{h^*}}$, which depends on the state of the world only through $h^*$ and $X$.
\vspace{0.2cm}
\item If $\lambda=0$, then $A_{h^*}= A^0_\lambda$ and $u'(W_0-X(s)+Y_1^*(s)-\Pi_0)h^*(s)>0$, for all $s\in A\setminus A_{h^*}$. Thus $Y_1^*\mathrm{1}_A=Y\mathrm{1}_{A_{h^*}}+ X\mathrm{1}_{A\setminus A_{h^*}}$, for any feasible $Y\in\mathcal{I}$.
\end{proofcases}

The indemnity $Y_{1,\lambda}^*:=Y_1^*$, depending on $\lambda$, is the optimal solution of Problem~\eqref{optimV2c1} if there exists some $\lambda^*\geq 0$ such that $\mathbb{E}_Q[v(W_0^{\text{Ins}}-(1+\rho)Y_{1,\lambda^*}^*+\Pi_0)]=v(W_0^{\text{Ins}})$. To see this, define the constant $\overline{\lambda}:=\dfrac{\val(\mathcal{L})+\varepsilon}{v(W_0^{\text{Ins}}+\Pi_0)-v(W_0^{\text{Ins}})}\in\mathbb{R}_+$, for some large $\varepsilon>\big\vert u(W_0-\Pi_0)\big\vert P^*_{ac}(A)$. Then for any $\lambda>\overline{\lambda}$, we obtain
\begin{equation*}
\begin{split}
\sup_{Y_{1,\lambda}\in\mathcal{I}}\mathcal{L}(Y_{1,\lambda},\lambda) & \geq \mathcal{L}(0,\lambda) = \int_A u(W_0-X-\Pi_0)h^*dQ+ \lambda\bigg(\int_A v(W_0^{\text{Ins}}+\Pi_0)dQ- v(W_0^{\text{Ins}})\bigg) \\
& \geq  u(W_0-\Pi_0)P^*_{ac}(A)+\overline{\lambda}\big(v(W_0^{\text{Ins}}+\Pi_0)-v(W_0^{\text{Ins}})\big) =  u(W_0-\Pi_0)P^*_{ac}(A)+\val(\mathcal{L})+\varepsilon.
\end{split}
\end{equation*}

\vspace{0.2cm}

It follows that $\displaystyle{\inf_{\lambda\geq\overline{\lambda}}\,\sup_{Y_{1,\lambda}\in\mathcal{I}}}\mathcal{L}(Y_{1,\lambda},\lambda)>\val(\mathcal{L})$, a contradiction; hence the feasible set of $\lambda$ reduces to the compact interval $[0,\overline{\lambda}]$ and the strong duality of Problem~\eqref{optimV2c1} yields
\[
\sup_{Y_{1,\lambda}\in\mathcal{I}}\min_{\lambda\in [0,\overline{\lambda}]}\mathcal{L}(Y_{1,\lambda},\lambda) = \sup_{Y_{1,\lambda}\in\mathcal{I}}\mathcal{L}(Y_{1,\lambda^*},\lambda^*)=\mathcal{L}(Y_{1,\lambda^*}^*,\lambda^*),
\]
where $\mathbb{E}_Q[v(W_0^{\text{Ins}}-(1+\rho)Y_{1,\lambda^*}^*+\Pi_0)]=v(W_0^{\text{Ins}})$.
\end{proof}

\vspace{0.2cm}

Note that the optimal indemnity $Y^*$ in eq.~\eqref{YoptimCase1V} provides full insurance over the event $S\setminus A$, whenever the DM's worst-case measure $P^*$ is such that $P^*(S\setminus A)\neq\emptyset$. Note also that on the event $A$, $Y^*$ satisfies an ordinary differential equation for which an analytical expression is difficult to provide in general. However, for particular choices of $\mathcal{C}$, we can obtain numerically the structure of $Y^*$, as well as $P^*$ (see Example~\ref{exRenyi} in Section~\ref{section5}).

\vspace{0.4cm}

%======================================================================================
\subsection{With the No-Sabotage Condition}\label{section32}
Next we analyze the case when the set $\mathcal{F}$ of Problem~\eqref{optimV} is restricted to the set of indemnities satisfying the no-sabotage condition. In this case the feasibility set becomes:
%$$\hat{\mathcal{I}}:=\lbrace \hat{Y}=\hat{I}\circ X \in B^+(\Sigma) :\, \hat{I}(0)=0, \, 0\leq \hat{I}(x_1)-\hat{I}(x_2)\leq x_1-x_2, \forall\, 0\leq x_2\leq x_1 \rbrace$$
\begin{equation}\label{FeasibHatI}
\hat{\mathcal{I}}_0 := \Big\{ \hat{Y} \in \hat{\mathcal{I}} :\,\mathbb{E}_Q [v( W_0^{\text{Ins}}-(1+\rho)\hat{Y}+\Pi_0) ] \geq v(W_0^{\text{Ins}}) \Big\}.
\end{equation}

\vspace{0.2cm}

\begin{remark}\label{CompactIhat0}
Since $\hat{\mathcal{I}}$ is a compact subset of the space $\(C[0,M], \Vert\cdot\Vert_{sup}\)$ (see Remark~\ref{CompactIhat}), and $\hat{\mathcal{I}}_0$ is a closed subset of $\hat{\mathcal{I}}$, it follows that $\hat{\mathcal{I}}_0$ is compact.
\end{remark}

\vspace{0.2cm}

\begin{theorem}\label{propCase2V}
Suppose that the utility functions $u$ and $v$ satisfy Assumption~\ref{assumptionsUtility}. Let $\mathcal{F}=\hat{\mathcal{I}}$ as defined in eq.\ \eqref{FeasibWithNSC} be the set of admissible indemnity functions. Then there exists $P^*\in\mathcal{C}$ such that the optimal solution of Problem~\eqref{optimV}  is $\hat{Y}^*\in\hat{\mathcal{I}}$, $Q$-a.s. and is of the form
\begin{equation*}
\hat{Y}^*=  \hat{Y}_1^*\mathrm{1}_A + X\mathrm{1}_{S\setminus A},
\end{equation*}
where
\vspace{0.2cm}
\begin{enumerate}[label=(\alph*)]
\item $A\in\Sigma$ is such that $P^*=P^*_{ac}+P^*_s$, with $P^*_s(A)=Q(S\setminus A)=0$;
\vspace{0.2cm}
\item $h^*:S\rightarrow [0,\infty)$ is such that $h^*=dP^*_{ac}/dQ$;
\vspace{0.2cm}
\item $\xi^*:\mathbb{R}_+\rightarrow\mathbb{R}_+$ is a Borel measurable function such that $h^*=\xi^*\circ X$;
\vspace{0.2cm}
\item $\hat{Y}_1^*=\hat{I}^*\circ X$, where $\hat{I}^*(x)=\displaystyle{\int_0^x} (\hat{I}^*)'(t)\,dt$, $\forall x\in [0,M]$ and
\end{enumerate}
{\footnotesize
\[
(\hat{I}^*)'(t)=\begin{cases}
0, &\text{if } \int_{[t,M]\cap X(A)} (u'(W_0-x+\hat{I}^*(x)-\Pi_0)\xi^*(x)-\lambda^*(1+\rho)v'(W_0^{\text{Ins}}-(1+\rho)\hat{I}^*(x)+\Pi_0))\,dF_{X,Q}(x)
<0,\\
\kappa(t),  &\text{if }  \int_{[t,M]\cap X(A)} (u'(W_0-x+\hat{I}^*(x)-\Pi_0)\xi^*(x)-\lambda^*(1+\rho)v'(W_0^{\text{Ins}}-(1+\rho)\hat{I}^*(x)+\Pi_0))\,dF_{X,Q}(x)=0, \\
1,  &\text{if } \int_{[t,M]\cap X(A)} (u'(W_0-x+\hat{I}^*(x)-\Pi_0)\xi^*(x)-\lambda^*(1+\rho)v'(W_0^{\text{Ins}}-(1+\rho)\hat{I}^*(x)+\Pi_0))\,dF_{X,Q}(x)>0,
\end{cases}
\]}
for some Lebesgue measurable and $[0,1]$-valued function $\kappa$;
\begin{enumerate}[label=(\alph*), resume]
\vspace{0.2cm}
\item $\lambda^*\in\mathbb{R}_+$ is such that $\mathbb{E}_Q [v( W_0^{\text{Ins}}-(1+\rho)\hat{Y}^*+\Pi_0) ]= v(W_0^{\text{Ins}})$.
\end{enumerate}
\end{theorem}

\vspace{0.2cm}

\begin{proof}
Similar to Theorem~\ref{propCase1V}, there exists a saddle point $(\hat{Y}^*,P^*)\in\hat{\mathcal{I}}_0\times\mathcal{C}$ such that
\begin{equation}\label{optimVc2}
\sup_{\hat{Y}\in\hat{\mathcal{I}}_0}\,\inf_{P\in\mathcal{C}}\mathbb{E}_P [u(W_0-X+\hat{Y}-\Pi_0)] = \min_{P\in\mathcal{C}}\,\max_{\hat{Y}\in\hat{\mathcal{I}}_0}\mathbb{E}_P [u(W_0-X+\hat{Y}-\Pi_0)],
\end{equation}
\noindent where $\hat{\mathcal{I}}_0$ given in eq.\ \eqref{FeasibHatI} is compact (see Remark~\ref{CompactIhat0}). For $P^*\in\mathcal{C}$, the inner optimization problem in~\eqref{optimVc2} becomes:
\begin{equation*}
\begin{aligned}
& \underset{\hat{Y}\in\hat{\mathcal{I}}}{\sup}
& & \int_A u(W_0-X+\hat{Y}-\Pi_0)h^*\,dQ+ \int_{S\setminus A} u(W_0-X+\hat{Y}-\Pi_0)\,dP^*_s \\
& \text{s.t.}
& & \int_A v(W_0^{\text{Ins}}-(1+\rho)\hat{Y}+\Pi_0)\,dQ\geq v(W_0^{\text{Ins}}),
\end{aligned}
\end{equation*}
where $A:=A_{P^*}$ and $h^*:=h_{P^*}$, depending on $P^*$, are defined in Remark~\ref{remDecomposition}. Similar to Theorem~\ref{propCase1V}, the optimal indemnity function $\hat{Y}^*$ can be obtained as $\hat{Y}^*=\hat{Y}_1^*\mathrm{1}_A + X\mathrm{1}_{S\setminus A}$, where $\hat{Y}_1^*$ solves Problem~\eqref{optimV1c2} below.

\begin{equation}\label{optimV1c2}
\sup_{\hat{Y}_1\in\hat{\mathcal{I}}} \bigg\lbrace \int u(W_0-X+\hat{Y}_1-\Pi_0)h^*\,dQ\, : \, \int v(W_0^{\text{Ins}}-(1+\rho)\hat{Y}_1+\Pi_0)\,dQ\geq v(W_0^{\text{Ins}})\bigg\rbrace.
\end{equation}

\vspace{0.2cm}

The Lagrange function of Problem~\eqref{optimV1c2} is
\begin{equation*}
\mathcal{L}(\hat{Y}_1,\lambda)  = \int_A \big[ u(W_0-X+\hat{Y}_1-\Pi_0)h^* + \lambda v(W_0^{\text{Ins}}-(1+\rho)\hat{Y}_1+\Pi_0) \big]\,dQ - \lambda v(W_0^{\text{Ins}}),
\end{equation*}
where $\lambda\in\mathbb{R}_+$ is the Lagrange multiplier. By strong duality, for fixed $\lambda\in\mathbb{R}_+$, a necessary and sufficient condition for $\hat{Y}_1^*\in\mathcal{I}$ to be the optimal solution of Problem~\eqref{optimV1c2} is
\begin{equation}\label{eqYoptimV1c2}
\lim_{\theta\rightarrow 0^+}\mathcal{L}'((1-\theta)\hat{Y}_1^*+\theta \hat{Y}_1)\leq 0,\,  \forall\,\hat{Y}_1\in\hat{\mathcal{I}}.
\end{equation}

\vspace{0.2cm}

Since $Q(S\setminus A)=0$, we can extend the domain of the integral above in~\eqref{eqYoptimV1c2} over $[0,M]$. By direct computation,~\eqref{eqYoptimV1c2} becomes:  for all $\hat{Y}_1=\hat{I}(X)\in\hat{\mathcal{I}}$,
\begin{equation} \label{eqY1optimV1c2}
\int_0^M [u'(W_0-t+\hat{I}^*(t)-\Pi_0)\xi^*(t)-\lambda (1+\rho)v'(W_0^{\text{Ins}}-(1+\rho)\hat{I}^*(t)+\Pi_0)](\hat{I}(t)-\hat{I}^*(t))\, dQ\circ X^{-1}(t)\leq 0.
\end{equation}
As any $\hat{Y}_1=\hat{I}(X)\in\hat{\mathcal{I}}$ is absolutely continuous, it is almost everywhere differentiable on $[0,M]$, and hence~\eqref{eqY1optimV1c2} becomes
\begin{equation*}
\begin{split}
0 & \geq \int_0^M\int_0^t (u'(W_0-t-\hat{I}^*(t)-\Pi_0)\xi^*(t)-\lambda (1+\rho)v'(W_0^{\text{Ins}}-(1+\rho)\hat{I}^*(t)+\Pi_0))\times\\
& \hspace*{10cm} \times (\hat{I}'(x)-(\hat{I}^*)'(x))\,dx\,dQ\circ X^{-1}(t) \\
& = \int_0^M\int_x^M (u'(W_0-t-\hat{I}^*(t)-\Pi_0)\xi^*(t)-\lambda (1+\rho)v'(W_0^{\text{Ins}}-(1+\rho)\hat{I}^*(t)+\Pi_0))\,dQ\circ X^{-1}(t)\times \\
& \hspace*{12cm} \times (\hat{I}'(x)-(\hat{I}^*)'(x))\,dx,
\end{split}
\end{equation*}
for all $\hat{Y}_1=\hat{I}(X) \in\hat{\mathcal{I}}$; hence $\hat{Y}^*_1=\hat{I}^*(X)$ is of the form:
\[
(\hat{I}^*)'(x) =\begin{cases}
0, &\mbox{if } \int_{[t,M]\cap X(A)} \tau(x)\,dQ\circ X^{-1}(x)<0, \\
\kappa(t),  &\mbox{if }  \int_{[t,M]\cap X(A)} \tau(x)\,dQ\circ X^{-1}(x) =0, \\
1,  &\mbox{if } \int_{[t,M]\cap X(A)} \tau(x)\,dQ\circ X^{-1}(x)>0,
\end{cases}
\]
\noindent where $\tau(x):= u'(W_0-x+\hat{I}^*(x)-\Pi_0)\xi^*(x)-\lambda^*(1+\rho)v'(W_0^{\text{Ins}}-(1+\rho)\hat{I}^*(x)+\Pi_0)$ and
$\kappa$ is some Lebesgue measurable and $[0,1]$-valued function. The existence of the Lagrange multiplier $\lambda^*\in\mathbb{R}_+$ that guarantees the existence of the solution $\hat{Y}_1^*$ follows similar to Theorem~\ref{propCase1V}.
\end{proof}

\vspace{0.2cm}

Theorem~\ref{propCase2V} above provides a general characterization of the optimal solution of Problem~\eqref{optimV}, when the set of admissible indemnity functions is given by $\hat{\mathcal{I}}$. The exact structure of the optimal indemnity $\hat{Y}^*$ may be difficult to interpret, due to its implicit form. However, under the (more common) assumption of risk-neutrality of the insurer, closed-form solutions for $\hat{Y}^*$ can be obtained, as seen in the next section.

\vspace{0.4cm}

%====================================================================================
%====================================================================================
%====================================================================================

\section{The case of a Risk-Neutral Insurer}\label{section4}

In this section, we examine the case of a risk-neutral insurer, that is, when the utility function $v$ is linear. We first consider in Section \ref{subsection1} the subcase in which the DM's utility function $u$ is concave, and characterize the optimal indemnity both without and with the no-sabotage condition. In the former case, we obtain a closed-form characterization of the optimal indemnity (Proposition~\ref{propCase1}), whereas in the latter case, the optimal indemnity is determined implicitly (Proposition~\ref{propCase2}). The results are illustrated in Example~\ref{exGeneralU} for a specific ambiguity set $\mathcal{C}$, and closed-form solutions are obtained. We then examine in Section~\ref{subsection2} the subcase in which the DM's utility function $u$ is linear. In that case, we also characterize the optimal indemnity both with and without the no-sabotage condition, and then provide two illustrative examples.

\vspace{0.2cm}

\subsection{The Case of a Concave Utility for the DM}\label{subsection1}
In this section, we study Problem~\eqref{optimV} under the assumption of risk neutrality of insurer:

\begin{equation*}\label{optim}
\left\{
\begin{aligned}
& \underset{Y\in \mathcal{F}}{\sup}\,\underset{P\in\mathcal{C}}{\inf}
& & \mathbb{E}_P [u( W_0-X+Y-\Pi_0) ]\\
& \text{s.t.}
& & (1+\rho)\mathbb{E}_Q [Y]\leq \Pi_0.
\end{aligned}\tag{$P_1$}
\right.
\end{equation*}

\vspace{0.2cm}

If $(1+\rho)\mathbb{E}_Q[X]\leq \Pi_0$, then we can eliminate the constraint in Problem~\eqref{optim}, as the DM's budget is large enough. In this case, the optimal indemnity is $Y^*=X,\, Q$-a.s. In the following, we assume that $(1+\rho)\mathbb{E}_Q[X]>\Pi_0$.

\vspace{0.2cm}

\subsubsection{Without the No-Sabotage Condition}

\begin{proposition}\label{propCase1}
Suppose that the utility function $u$ satisfies Assumption~\ref{assumptionsUtility} and is, in addition, strictly concave and such that $\underset{x\rightarrow -\infty} \lim u'(x) = +\infty$ and $\underset{x\rightarrow +\infty} \lim u'(x) = 0$. Let $\mathcal{F}=\mathcal{I}$ as defined in eq.\ \eqref{FeasibWithoutNSC} be the set of admissible indemnity functions. Then there exists $P^*\in\mathcal{C}$ such that an optimal solution $Y^*\in\mathcal{I}$ of Problem~\eqref{optim} is of the form:
\begin{equation}\label{optimY}
Y^*=  (X-R^*)\mathrm{1}_{A\setminus A_{h^*}} + (X-R_{h^*})\mathrm{1}_{A_{h^*}} + X\mathrm{1}_{S\setminus A},
\end{equation}
where
\vspace{0.2cm}
\begin{enumerate}[label=(\alph*)]
\item \label{itm:1c1} $A\in\Sigma$ is such that $P^*=P^*_{ac}+P^*_s$, with $P^*_s(A)=Q(S\setminus A)=0$;
\vspace{0.2cm}
\item \label{itm:2c1} $h^*:S\rightarrow [0,\infty)$ is such that $h^*=dP^*_{ac}/dQ$;
\vspace{0.2cm}
\item \label{itm3c1} $A_{h^*}:=\lbrace s\in A:\, h^*(s)=0\rbrace$;
\vspace{0.2cm}
\item \label{itm4c1} $R^*$ and $R_{h^*}$ are given by:
\vspace{0.2cm}
\end{enumerate}
\begin{thmcases}
\item \label{case1c1} If $(1+\rho)\mathbb{E}_Q[X\mathrm{1}_{S\setminus A_{h^*}}]> \Pi_0$, then $R^*=\max\left[0, \min\left(X,W_0-\Pi_0-(u^\prime)^{-1}\left(\dfrac{\lambda^*}{h^*}\right)\right)\right]\mathrm{1}_{A\setminus A_{h^*}}$ and $R_{h^*}=X\mathrm{1}_{A_{h^*}}$, $Q$-a.s., where $\lambda^*\in\mathbb{R}_+$ is such that $(1+\rho)\mathbb{E}_Q[Y^*]=\Pi_0$;
\vspace{0.6cm}
\item \label{case2c1} If $(1+\rho)\mathbb{E}_Q[X\mathrm{1}_{S\setminus A_{h^*}}]\leq \Pi_0$, then $R^*=0$, $Q$-a.s. and $R_{h^*}=c\,X \mathrm{1}_{A_{h^*}}$, where $c\in (0,1]$ is defined as $c:=\dfrac{\mathbb{E}_Q[X]-(1+\rho)^{-1}\Pi_0}{\mathbb{E}_Q[X\mathrm{1}_{A_{h^*}}]}$.
\end{thmcases}
\end{proposition}

\vspace{0.2cm}

\begin{proof}
Letting  $R:=X-Y$ be the retention random variables, Problem~\eqref{optim} becomes
\begin{equation}\label{optimc1}
\begin{aligned}
& \underset{R\in\mathcal{I}}{\sup}\,\underset{P\in\mathcal{C}}{\inf}
& & \mathbb{E}_P [u(W_0-R-\Pi_0)]\\
& \text{s.t.}
& & \mathbb{E}_Q [R]\geq \tilde{\Pi}_0,
\end{aligned}
\end{equation}
where $\tilde{\Pi}_0:=\mathbb{E}_Q[X]-(1+\rho)^{-1}\Pi_0$. For fixed $\Pi_0$, define the convex set $\mathcal{I}_0:=\lbrace R\in\mathcal{I}:\, \mathbb{E}_Q [R]\geq \tilde{\Pi}_0\rbrace$. Problem~\eqref{optimc1} fulfills the conditions of Sion's Minimax Theorem, and thus there exists some $R^*\in\mathcal{I}_0$ and $P^*\in \mathcal{C}$ such that
\begin{equation*}
\sup_{R\in\mathcal{I}_0}\,\inf_{P\in\mathcal{C}}\mathbb{E}_P [u(W_0-R-\Pi_0)] = \min_{P\in\mathcal{C}}\,\sup_{R\in\mathcal{I}_0}\mathbb{E}_P [u(W_0-R-\Pi_0)].
\end{equation*}
For $P^*\in\mathcal{C}$, the inner problem becomes:
\begin{equation}\label{optim1c1}
\begin{aligned}
& \underset{R\in\mathcal{I}}{\sup}
& & \int_A u(W_0-R-\Pi_0)h^*\,dQ+ \int_{S\setminus A} u(W_0-R-\Pi_0)\,dP^*_s \\
& \text{s.t.}
& & \int_A R(s)\,dQ(s)\geq \tilde{\Pi}_0,
\end{aligned}
\end{equation}
where $A:=A_{P^*}$ and $h^*:=h_{P^*}$ are defined in Remark~\ref{remDecomposition}. Based on the splitting technique in \cite[Lemma C.6]{Ghossoub2019c}, the optimal retention function $R^*$ can be obtained as $R^*=R_1^*\mathrm{1}_A$, where $R_1^*$ solves Problem~\eqref{optim2c1} below.

\begin{equation}\label{optim2c1}
\sup_{R_1\in\mathcal{I}} \bigg\lbrace \int u(W_0-R_1-\Pi_0)h^*\,dQ\, : \, \int R_1 dQ \geq \tilde{\Pi}_0\bigg\rbrace.
\end{equation}

\vspace{0.2cm}

Following the proof of Theorem~\ref{propCase1V}, for a fixed $\lambda\in\mathbb{R}_+$,  a necessary and sufficient condition for $R_1^*\in\mathcal{I}$ to be the optimal solution of Problem~\eqref{optim2c1} is
\begin{equation}\label{eqRoptim}
\begin{aligned}
& \int_A \big[u'(W_0-R_1^*(s)-\Pi_0)h^*(s)-\lambda\big](R_1^*(s)-R_1(s))dQ(s)\leq 0,\, \forall R_1\in\mathcal{I}.
\end{aligned}
\end{equation}

\noindent Next, we define the following sets, depending on Lagrange multiplier $\lambda\in\mathbb{R}_+$:
\[
\left \{\begin{aligned}
	  A^+_\lambda & :=\lbrace s\in A : u'(W_0-R_1^*(s)-\Pi_0)h^*(s)-\lambda>0\rbrace, \\
    A^0_\lambda & :=\lbrace s\in A: u'(W_0-R_1^*(s)-\Pi_0)h^*(s)-\lambda=0\rbrace, \\
    A^-_\lambda & :=\lbrace s\in A: u'(W_0-R_1^*(s)-\Pi_0)h^*(s)-\lambda<0\rbrace.
    \end{aligned}\right.
\]

\vspace{0.2cm}

\noindent Clearly, condition~\eqref{eqRoptim} holds for all $R_1\in\mathcal{I}$ only if
\begin{equation}\label{eqAminusAplus}
R_1^*\mathrm{1}_{A^+_\lambda}=0 \text{ and } R_1^*\mathrm{1}_{A^-_\lambda}=X\mathrm{1}_{A^-_\lambda}.
\end{equation}
Next, define the set $A_{h^*}:=\lbrace s\in A:\, h^*(s)=0\rbrace$. To characterize $R_1^*$ in \eqref{optim2c1}, we consider the following cases:
\vspace{0.2cm}
\begin{proofcases}[noitemsep,wide=0pt, leftmargin=\dimexpr\labelwidth + 2\labelsep\relax]
\item \label{case1lambda} If $\lambda>0$, then $A_{h^*}\subseteq A^-_\lambda$ and thus $R_1^*\mathrm{1}_{A_{h^*}}=X\mathrm{1}_{A_{h^*}}$. On $A^0_\lambda$, $R_1^*$ satisfies the following condition:
\[
u'(W_0-R_1^*(s)-\Pi_0)h^*(s)-\lambda=0.
\]
By equation~\eqref{eqAminusAplus} and strict monotonicity of $u$, the first order condition yields:
\begin{equation}\label{eqRcase1propCase1}
R_1^*\mathrm{1}_A =\max\left[0, \min\left(X,W_0-\Pi_0-(u^\prime)^{-1}\left(\dfrac{\lambda}{h^*}\right)\right)\right]\mathrm{1}_{A\setminus A_{h^*}} + X\mathrm{1}_{A_{h^*}},
\end{equation}
which depends on the state of the world only through $h^*$ and $X$.
\vspace{0.2cm}
\item \label{case2lambda} If $\lambda=0$, then $A_{h^*} = A^0_\lambda$ and $u'(W_0-R_1^*(s)-\Pi_0)h^*(s)>0$, for all $s\in A\setminus A_{h^*}$. Moreover, in this case, $Q(A_{h^*})>0$; otherwise, $R_1^*\mathrm{1}_A=0$, $Q$-a.s., leading to a contradiction, since $\mathbb{E}_Q[R_1^*]=0<\tilde{\Pi}_0$. Thus, according to~\eqref{eqAminusAplus}, $R_1^*\mathrm{1}_A=R\mathrm{1}_{A_{h^*}}$, for any feasible $R\in\mathcal{I}$.
\end{proofcases}

\vspace{0.2cm}

The retention $R_1^*$ is the optimal solution of Problem~\eqref{optim2c1} if there exists some $\lambda^*\in\mathbb{R}$ such that $\mathbb{E}_Q[R_1^*]=\tilde{\Pi}_0$. We will denote this as $R_{1,\lambda}^*$, to emphasize the dependence on $\lambda$. Now define the function $\varphi:\mathbb{R}_+\rightarrow\mathbb{R}_+$, $\varphi(\lambda):=\mathbb{E}_Q[R_{1,\lambda}^*]$. Since $R_{1,\lambda}^*$ is continuous in $\lambda$, and $R_{1,\lambda}^*\leq X$, by Lebesgue Dominated Convergence Theorem, it follows that $\varphi(\lambda)$ is continuous in $\lambda$. Moreover, $\varphi$ is a nondecreasing function of $\lambda$ and satisfies $\lim_{\lambda\rightarrow 0}\varphi(\lambda)=\mathbb{E}_Q[R\mathrm{1}_{A_{h^*}}]$, for some $R\in\mathcal{I}$ and $\lim_{\lambda\rightarrow \infty}\varphi(\lambda)=\mathbb{E}_Q[X]$. If $\tilde{\Pi}_0 > \mathbb{E}_Q[X\mathrm{1}_{A_{h^*}}]$, there exists some $\lambda^*>0$ such that $\varphi(\lambda^*)=\tilde{\Pi}_0$. Otherwise, $\lambda^*=0$.

\vspace{0.2cm}

If $\mathbb{E}_Q[X\mathrm{1}_{A_{h^*}}]< \tilde{\Pi}_0$, this implies that $\lambda^* > 0$, and therefore $R_1^*$ is as in~\eqref{eqRcase1propCase1}. Now, if $\mathbb{E}_Q[X\mathrm{1}_{A_{h^*}}]\geq\tilde{\Pi}_0$, then according to Cases~\ref{case1lambda} and~\ref{case2lambda} and the discussion above, $\lambda^*=0$ and $R_{1,\lambda}^*\mathrm{1}_{A_{h^*}}=R\mathrm{1}_{A_{h^*}}$, for any $R\in\mathcal{I}$ that is feasible for Problem~\eqref{optim}. For instance, $R\mathrm{1}_{A_{h^*}}=c\,X\mathrm{1}_{A_{h^*}}$ with $c:=\dfrac{\tilde{\Pi}_0}{\mathbb{E}_Q[X\mathrm{1}_{A_{h^*}}]}\in (0,1]$ satisfies $R\mathrm{1}_{A_{h^*}}\leq X\mathrm{1}_{A_{h^*}}$ and the constraint $\mathbb{E}_Q[R]=\tilde{\Pi}_0$, and hence it is optimal for~\eqref{optim}.
\end{proof}

% \vspace{0.2cm}

%Note that the conditions imposed on the utility function $u$ in Proposition~\ref{propCase1} are weaker than the Inada-type conditions, often encountered in the literature.

\vspace{0.4cm}

%======================================================================================

\begin{remark}\normalfont
In the setting of Proposition~\ref{propCase1}, Case~\eqref{case2c1}, an important special case is when $h^*\equiv 0$, i.e., $P^*\perp Q$, where $P^*$ is the worst-case measure that attains the infimum in~\eqref{optim}. Thus the optimal indemnity function is $Y^*=(X-R_{h^*})\mathrm{1}_A + X\mathrm{1}_{S\setminus A}$, where $R_{h^*}\in\mathcal{I}$ satisfies $\int_A R_{h^*}(s)\,dQ(s)=\tilde{\Pi}_0$. A possible choice for $R_h^*$ is shown in Proposition~\ref{propCase1}.
\end{remark}

\vspace{0.4cm}

%======================================================================================

\subsubsection{With the No-Sabotage Condition}

The following result characterizes the optimal indemnity $Y^*$, when the no-sabotage condition is enforced.

\vspace{0.2cm}

\begin{proposition}\label{propCase2}

Suppose that the utility function $u$ satisfies Assumption~\ref{assumptionsUtility} and let $\mathcal{F}=\hat{\mathcal{I}}$ as defined in eq.\ \eqref{FeasibWithNSC} be the set of admissible indemnity functions. Then there exists $P^*\in\mathcal{C}$ such that the optimal solution of Problem~\eqref{optim} is $\hat{Y}^*\in\hat{\mathcal{I}}$, $Q$-a.s. and is of the form
\begin{equation*}
\hat{Y}^*=  (X-\hat{R}^*)\mathrm{1}_A + X\mathrm{1}_{S\setminus A},
\end{equation*}
where
\vspace{0.2cm}
\begin{enumerate}[label=(\alph*)]
\item $A\in\Sigma$ is such that $P^*=P^*_{ac}+P^*_s$, with $P^*_s(A)=Q(S\setminus A)=0$;
\vspace{0.2cm}
\item $h^*:S\rightarrow [0,\infty)$ is such that $h^*=dP^*_{ac}/dQ$;
\vspace{0.2cm}
\item $\xi^*:\mathbb{R}_+\rightarrow\mathbb{R}_+$ is a Borel measurable function such that $h^*=\xi^*\circ X$;
\vspace{0.2cm}
\item $\hat{R}^*=\hat{r}^*\circ X$, where $\hat{r}^*(x)=\displaystyle{\int_0^x} (\hat{r}^*)'(t)\,dt$, $\forall x\in [0,M]$ and
\[
(\hat{r}^*)'(t)=\begin{cases}
0, &\mbox{if } \int_{[t,M]\cap X(A)} (u'(W_0-\hat{r}^*(x)-\Pi_0)\xi^*(x)-\lambda^*)\,dQ\circ X^{-1}(x)>0,\\
\kappa(t),  &\mbox{if }  \int_{[t,M]\cap X(A)} (u'(W_0-\hat{r}^*(x)-\Pi_0)\xi^*(x)-\lambda^*)\,dQ\circ X^{-1}(x)=0,\\
1,  &\mbox{if } \int_{[t,M]\cap X(A)} (u'(W_0-\hat{r}^*(x)-\Pi_0)\xi^*(x)-\lambda^*)\,dQ\circ X^{-1}(x)<0,
\end{cases}
\]
for some Lebesgue measurable and $[0,1]$-valued function $\kappa$;
\vspace{0.2cm}
\item $\lambda^*\in\mathbb{R}_+$ is such that $(1+\rho)\mathbb{E}_Q[\hat{Y}^*]=\Pi_0$.
\end{enumerate}
\end{proposition}

Proposition~\ref{propCase2} is a particular case of Theorem~\ref{propCase2V} when the insurer's utility $v$ is linear. For completeness, the proof is provided in Appendix~\ref{appendixProofpropCase2}.

\vspace{0.4cm}

%======================================================================================
\subsubsection{Example}
The following example analyzes the structure of the optimal $Y^*$ and $\hat{Y}^*$ in Propositions~\eqref{propCase1} and~\eqref{propCase2}, respectively, when all probability measures $P$ in $\mathcal{C}$ are absolutely continuous with respect to $Q$, with a particular structure of the Radon-Nikod\'ym derivatives. This is done both with and without the no-sabotage condition. Specifically, we assume that each $P \in \mathcal{C}$ is such that $P \ll  Q$ with
\begin{equation}\label{WeightedTransform}
\dfrac{dP}{dQ}=\dfrac{w(X)}{\int w(X)dQ},
\end{equation}
for some nonnegative and increasing weight function $w$ satisfying $\int w(X)dQ > 0$. Such measure transformations have a long tradition in insurance pricing, dating back to the Esscher transform (e.g., B\"{u}hlmann \cite{Buhlmann1980}), in which the function $w$ takes the form $w(x) = e^{bx}$, for a given $b \in \(0,+\infty\)$. More generally, Furman and Zitikis \cite{FurmanZitikis2008a,FurmanZitikis2008b,FurmanZitikis2009} discuss the general class of weighted premium principles where pricing is done via measure transformations as in eq.~\eqref{WeightedTransform}.

\vspace{0.2cm}

Suppose that the utility function $u$ satisfies Assumption~\ref{assumptionsUtility}, and assume that insurer's probability measure $Q$ has a continuous cdf over $[0,M]$.

\vspace{0.2cm}

\begin{example}\label{exGeneralU}
Let the DM's ambiguity set $\mathcal{C}$ be defined as follows:
\begin{equation*}
\mathcal{C}_{\mathcal{W}}:=\bigg\lbrace P\in ca^{+}_{1}(\Sigma):\, \dfrac{dP}{dQ}=\dfrac{w(X)}{\int w(X)dQ},\, w\in \mathcal{W}\bigg\rbrace,
\end{equation*}
\noindent where $\mathcal{W} \subset L^1\left(\mathbb{R}, \mathcal{B}(\mathbb{R}), Q \circ X^{-1}\right)$ is a collection of nonnegative increasing weight functions, such that $\int w(X)dQ > 0$ for all $w \in \mathcal{W}$. Appendix~\ref{appendixCompactConvexCw} provides conditions under which the set $\mathcal{C}_{\mathcal{W}}$ is convex and weak$^*$-compact.

\vspace{0.2cm}

First we analyze the case when the feasible set of indemnities is $\mathcal{F}=\mathcal{I}$, as defined in eq.~\eqref{FeasibWithoutNSC}. By definition of $\mathcal{C}_{\mathcal{W}}$, any optimal $P^*$ is absolutely continuous with respect to $Q$. Moreover, by monotonicity of $h^*=dP^*/dQ = \xi^*(X)$, there exists some $a\geq 0$ such that $\xi^*(x)=0$, for $x\in [0, a]$ and $\xi^*(x)>0$, for $x>a$, i.e., the set $A_{h^*}$ in Proposition~\ref{propCase1} is precisely $A_{h^*}=X^{-1}([0,a])$.

\vspace{0.2cm}

If $(1+\rho)\mathbb{E}_Q[X\mathrm{1}_{S\setminus A_{h^*}}] = (1+\rho)\int_a^M x \, dQ\circ X^{-1}(x)< \Pi_0$, the optimal indemnity $Y^*=I^*(X)$ in~\eqref{optimY} is such that 
\[
  I^*(x) = \max\left[0, \min\left(x,x-W_0+\Pi_0+(u^\prime)^{-1}\left(\dfrac{\lambda^*}{\xi^*}\right)\right)\right]\mathrm{1}_{[a,M]},
\]
where $\lambda^*\in\mathbb{R}_+$ is such that $(1+\rho)\mathbb{E}_Q[Y^*]=\Pi_0$.

\vspace{0.2cm}

If $(1+\rho)\mathbb{E}_Q[X\mathrm{1}_{S\setminus A_{h^*}}]= (1+\rho)\int_a^M x \, dQ\circ X^{-1}(x)\geq \Pi_0$, then \eqref{optimY} becomes
\begin{equation*}
I^*(x)=\begin{cases}
(1-c)x, &\mbox{if }  x\leq a,\\
x,  &\mbox{if } x>a,
\end{cases}
\end{equation*}
where the constant $c\in (0,1)$ is chosen as in Proposition~\ref{propCase1}.

\vspace{0.4cm}

Next, let $\mathcal{F}=\hat{\mathcal{I}}$, as defined in eq.~\eqref{FeasibWithNSC}. Following the setting of Proposition~\ref{propCase2}, the utility $u$ need not to be strictly concave, but only concave. According to Proposition~\ref{propCase2}, the optimal retention can be equivalently written as
\begin{equation*}
(\hat{r}^*)'(t)=\begin{cases}
0, &\mbox{if } \int_t^M (\lambda^*-u'(W_0-\hat{r}^*(x)-\Pi_0)\xi^*(x))\,dF_{X,Q}(x)<0,\\
\kappa(t),  &\mbox{if } \int_t^M (\lambda^*-u'(W_0-\hat{r}^*(x)-\Pi_0)\xi^*(x))\,dF_{X,Q}(x)=0,\\
1,  &\mbox{if } \int_t^M (\lambda^*-u'(W_0-\hat{r}^*(x)-\Pi_0)\xi^*(x))\,dF_{X,Q}(x)>0.
\end{cases}
\end{equation*}

\vspace{0.2cm}

Observe that the function $\varphi:[0,M]\rightarrow\mathbb{R}$, $\varphi(x):=-u'(W_0-\hat{r}^*(x)-\Pi_0)\xi^*(x)$ is a continuous, decreasing function, as it is the product of two decreasing functions. We distinguish the following cases:

\vspace{0.2cm}

\begin{thmcases}[noitemsep,wide=0pt, leftmargin=\dimexpr\labelwidth + 2\labelsep\relax]
\item If $-\lambda^*<\varphi(M)$, then $\int_t^M (\varphi(x)+\lambda^*)\,dF_{X,Q}(x)>0$, for all $t\in [0,M]$, and thus $(\hat{r}^*)'\equiv 1$.
\vspace{0.2cm}
\item If $-\lambda^*>\varphi(0)$, then $\int_t^M (\varphi(x)+\lambda^*)\,dF_{X,Q}(x)<0$, for all $t\in [0,M]$, and thus $(\hat{r}^*)'\equiv 0$.
\vspace{0.2cm}
\item If $-\lambda^*\in [\varphi(M),\varphi(0)]$, then there exists some $d^*\in (0,M)$ such that $\varphi(x)+\lambda^*>0$, for all $x<d$ and $\varphi(x)+\lambda^*<0$, for all $x>d $. This implies that for all $t\geq d$,  $\int_t^M (\varphi(x)+\lambda^*)\,dF_{X,Q}(x)<0$, and thus $(\hat{r}^*)'(t)= 0$, for $t\in [d,M]$. Moreover, for $t_1<t_2<d$, it holds that
\begin{equation*}
\int_{t_1}^d (\varphi(x)+\lambda^*)\,dF_{X,Q}(x) + \int_d^M (\varphi(x)+\lambda^*)\,dF_{X,Q}(x) \geq \int_{t_2}^d (\varphi(x)+\lambda^*)\,dF_{X,Q}(x) + \int_d^M (\varphi(x)+\lambda^*)\,dF_{X,Q}(x).
\end{equation*}
Therefore, there exists some $d^*\geq 0$ such that for all $t\leq d^*$,
$$
\int_t^{d^*} (\varphi(x)+\lambda^*)\,dF_{X,Q}(x) + \int_{d^*}^M (\varphi(x)+\lambda^*)\,dF_{X,Q}(x)<0.
$$
Therefore, $(\hat{r}^*)'(t)=1$ for all $t< d^*$, and $(\hat{r}^*)'(t)=0$, for all $t> d^*$. In this case, $\hat{r}^*(t)=\min(t,d^*)$ and thus $\hat{I}^*(t)=\max(t-d^*,0)$.
\end{thmcases}
\end{example}

\vspace{0.4cm}
%======================================================================================

\subsection{The Case of a Linear Utility for the DM}\label{subsection2}

In this section we assume that the utility functions of both the DM and the insurer are linear. Proposition~\ref{propUVlinearCase1} characterizes the optimal indemnity in the absence of the no-sabotage condition. If the admissible set of indemnities is $\hat{\mathcal{I}}$, then using the results of Proposition~\ref{propCase2}, we obtain a closed-form expression for the optimal indemnity in the presence of the no-sabotage condition, namely, a layer-type insurance indemnity. The section ends with two examples of the ambiguity set $\mathcal{C}$, for which we derive each time the optimal insurance contract.

\vspace{0.2cm}

When both DM and insurer are risk-neutral, Problem~\eqref{optim} can be further simplifies as follows:

\begin{equation*}\label{optimUVlinear}
\left\{
\begin{aligned}
& \underset{R\in \mathcal{F}}{\inf}\,\underset{P\in\mathcal{C}}{\sup}
& & \mathbb{E}_P [R ]\\
& \text{s.t.}
& & \mathbb{E}_Q [R]\geq \tilde{\Pi}_0,
\end{aligned}\tag{$P_2$}
\right.
\end{equation*}
where $R=X-Y$ is the retention function and $\tilde{\Pi}_0:=\mathbb{E}_Q[X]-(1+\rho)^{-1}\Pi_0$.  Clearly, if $\tilde{\Pi}_0\leq 0$, then $R^*=0$, $Q$-a.s., and so we focus on the case $\tilde{\Pi}_0>0$.

\vspace{0.2cm}
\subsubsection{Without the No-Sabotage Condition}

\begin{proposition}\label{propUVlinearCase1}
Let $\mathcal{F}=\mathcal{I}$ as defined in eq.~\eqref{FeasibWithoutNSC} be the set of admissible indemnity functions. Then there exists $P^*\in\mathcal{C}$ such that an optimal solution $Y^*\in\mathcal{I}$ of Problem~\eqref{optimUVlinear} is of the form:
\begin{equation*}
Y^*=  (X-R_{h^*})\mathrm{1}_{A}+ X\mathrm{1}_{S\setminus A},
\end{equation*}
where
\vspace{0.2cm}
\begin{enumerate}[label=(\alph*)]
\item $A\in\Sigma$ is such that $P^*=P^*_{ac}+P^*_s$, with $P^*_s(A)=Q(S\setminus A)=0$;
\vspace{0.2cm}
\item $h^*:S\rightarrow [0,\infty)$ is such that $h^*=dP^*_{ac}/dQ$;
\vspace{0.2cm}
\item $R_{h^*}$ is given by:
\vspace{0.2cm}
\begin{thmcases}[noitemsep,wide=0pt, leftmargin=\dimexpr\labelwidth + 2\labelsep\relax]
\item If $\mathbb{E}_Q[X\mathrm{1}_{A_{h^*}}]< \tilde{\Pi}_0$, then $R_{h^*}=X\mathrm{1}_{A^-_{\lambda^*}} + c\,X\mathrm{1}_{A^0_{\lambda^*}}$, i.e., $Y^*$ can be written as
\[
Y^* = (1-c)X\mathrm{1}_{A^0_{\lambda^*}} + X\mathrm{1}_{S\setminus (A^0_{\lambda^*}\cup A^-_{\lambda^*})},
\]
where \begin{enumerate}[label=(\roman*)]
\vspace{0.2cm}
\item $A_{h^*}:=\lbrace s\in A:\, h^*(s)=0\rbrace$;
\vspace{0.2cm}
\item $A^-_{\lambda^*}:=\lbrace s\in A :\, h^*(s)<\lambda^*\rbrace $;
\vspace{0.2cm}
\item $A^0_{\lambda^*}:=\lbrace s\in A :\, h^*(s)=\lambda^*\rbrace $;
\vspace{0.2cm}
\item $\lambda^*\in\mathbb{R}_+$ is such that $\mathbb{E}_Q[R_{h^*}]=\tilde{\Pi}_0$;
\vspace{0.2cm}
\item $c=\begin{cases} 0, &\mbox{if } \mathbb{E}_Q[X\mathrm{1}_{A^-_{\lambda^*}}]=\tilde{\Pi}_0, \\
\dfrac{\tilde{\Pi}_0-\mathbb{E}_Q[X\mathrm{1}_{A^-_{\lambda^*}}]}{\mathbb{E}_Q[X\mathrm{1}_{A^0_{\lambda^*}}]}, & \mbox{otherwise; } \end{cases}$
\end{enumerate}
\vspace{0.4cm}
\item  If $\mathbb{E}_Q[X\mathrm{1}_{A_{h^*}}]\geq \tilde{\Pi}_0$, then $R_{h^*}=c_0 X\mathrm{1}_{A_{h^*}}$,  i.e., $Y^*$ can be written as
\[
Y^* = (1-c_0)X\mathrm{1}_{A_{h^*}} + X\mathrm{1}_{S\setminus A_{h^*}},
\]
where $c_0:=\dfrac{\tilde{\Pi}_0}{\mathbb{E}_Q[X\mathrm{1}_{A_{h^*}}]}\in (0,1]$.
\end{thmcases}
\end{enumerate}
\end{proposition}

\vspace{0.2cm}

\begin{proof}
The proof is similar to the proof of Proposition~\ref{propCase1}, in the case $u(x)=x$. According to~\eqref{optim1c1}, there exists some $P^*\in\mathcal{C}$ such that for $A:=A_{P^*}$ and $h^*:=h_{P^*}$ as in Remark~\ref{remDecomposition}, $R^*$ is the optimal solution of the following problem:
\begin{equation*}
\inf_{R\in\mathcal{I}} \bigg\lbrace \int_A R h^*\,dQ+ \int_{S\setminus A} R\,dP^*_s : \, \int_A R dQ \geq \tilde{\Pi}_0 \bigg\rbrace.
\end{equation*}
According to \cite[Lemma C.6]{Ghossoub2019c}, the optimal retention is of the form $R^*=R_1^*\mathrm{1}_A$, where $R_1^*$ solves the problem below.
\begin{equation*}
\inf_{R_1\in\mathcal{I}} \bigg\lbrace \int_A R_1 h^*\,dQ : \, \int_A R_1 dQ \geq \tilde{\Pi}_0 \bigg\rbrace.
\end{equation*}
As both the objective function and the constraint are linear in $R_1$, a sufficient condition for $R^*_1$ to be optimal is that it minimizes pointwisely the integrand of the associated Lagrange function:
\begin{equation*}
\inf_{R_1\in\mathcal{I}}\,\sup_{\lambda\geq 0}\mathcal{L}(R_1,\lambda)  = \inf_{R_1\in\mathcal{I}}\,\sup_{\lambda\geq 0}\int_A (h^*(s)-\lambda)R_1(s) - \lambda\tilde{\Pi}_0,
\end{equation*}
where $\lambda\in\mathbb{R}_+$ is the Lagrange multiplier. To obtain the structure of optimal $R^*_{1,\lambda}:=R^*_1$, depending on $\lambda$, we consider again the following sets:
\[
\left \{\begin{aligned}
	A^+_\lambda & :=\lbrace s\in A : h^*(s)-\lambda>0\rbrace, \\
    A^0_\lambda & :=\lbrace s\in A: h^*(s)-\lambda=0\rbrace, \\
    A^-_\lambda & :=\lbrace s\in A: h^*(s)-\lambda<0\rbrace.
    \end{aligned}\right.
\]
It is clear that $R_{1,\lambda}\mathrm{1}_{A^+_\lambda}=0$ and $R_{1,\lambda}\mathrm{1}_{A^-_\lambda}=X\mathrm{1}_{A^-_\lambda}$. Next, let $A_{h^*}:=\lbrace s\in A: \,  h^*(s)=0\rbrace$. We distinguish the following cases:
\vspace{0.2cm}
\begin{proofcases}[noitemsep,wide=0pt, leftmargin=\dimexpr\labelwidth + 2\labelsep\relax]
\vspace{0.2cm}
\item If $\lambda>0$, then $A_{h^*}\subseteq A^-_\lambda$ and the corresponding retention is $R_{1,\lambda}=X\mathrm{1}_{ A^-_\lambda} + \tilde{R}\mathrm{1}_{A^0_\lambda}$, for some arbitrary $\tilde{R}\in\mathcal{I}$.
\vspace{0.2cm}
\item If $\lambda=0$, then $A_{h^*}=A^0_\lambda$ and $R_{1,\lambda}=\overline{R}\mathrm{1}_{A_{h^*}}$, for some arbitrary $\overline{R}\in\mathcal{I}$. Observe that in this case $Q(A_{h^*})>0$, otherwise $\mathbb{E}_Q[R_{1,\lambda}]=0<\tilde{\Pi}_0$, a contradiction.
\end{proofcases}

\vspace{0.2cm}

Similar to Proposition~\ref{propCase1}, there exists $\lambda^*\geq 0$ such that the slackness condition below holds:
\begin{equation}\label{eqSlackness}
\lambda^*(\mathbb{E}_Q[R_{1,\lambda^*}]-\tilde{\Pi}_0)=0.
\end{equation}

Now, based on the optimal value of $\lambda^*$, one can specify a choice of $\tilde{R}$ and $\overline{R}$, respectively, and thus characterize an optimal $R^*_{1,\lambda^*}$. To see this, let $\lambda_1:=\inf_{\, x\in X(A\setminus A_{h^*})}\xi^*(x)\geq 0$, for a Borel measurable function $\xi^*:\mathbb{R}_+\rightarrow\mathbb{R}_+$ such that $h^*=\xi^*(X)$. If $\lambda_1>0$, then $A_{h^*}=A^-_{\lambda_1}$ and the corresponding retention becomes $R_{1,\lambda_1}=X\mathrm{1}_{A_{h^*}} + \tilde{R}\mathrm{1}_{A^0_{\lambda_1}}$, for some $\tilde{R}\in\mathcal{I}$. Otherwise, $A_{h^*}=A^0_0$ and $R_{1,0}=\overline{R}\mathrm{1}_{A_{h^*}}$, for some $\overline{R}\in \mathcal{I}$.

If $\mathbb{E}_Q[X\mathrm{1}_{A_{h^*}}]<\tilde{\Pi}_0$, then $\lambda^*>\lambda_1$ and an optimal retention in this case is $R^*_{1,\lambda^*}=X\mathrm{1}_{A^-_{\lambda^*}} + c X\mathrm{1}_{A^0_{\lambda^*}}$, where the constant $c$ is defined as
\[
c:=\begin{cases} 0, &\mbox{if } \mathbb{E}_Q[X\mathrm{1}_{A^-_{\lambda^*}}]=\tilde{\Pi}_0, \\
\dfrac{\tilde{\Pi}_0-\mathbb{E}_Q[X\mathrm{1}_{A^-_{\lambda^*}}]}{\mathbb{E}_Q[X\mathrm{1}_{A^0_{\lambda^*}}]}, & \mbox{otherwise. }
\end{cases}
\]
Note that $c\in [0,1]$, since $\mathbb{E}_Q[X\mathrm{1}_{A^-_{\lambda^*}}]\leq \tilde{\Pi}_0$, according to eq.~\eqref{eqSlackness}. If $\mathbb{E}_Q[X\mathrm{1}_{A_{h^*}}]\geq \tilde{\Pi}_0$, $\lambda^*=0$ and one can choose $R^*_{1,0}=c_0 X\mathrm{1}_{A_{h^*}}$, where $c_0:=\tilde{\Pi}_0/\mathbb{E}_Q[X\mathrm{1}_{A_{h^*}}]$.
\end{proof}

\vspace{0.2cm}

Observe that the optimal $Y^*$ given in Proposition \ref{propUVlinearCase1} is not unique, since on the sets $A_{\lambda^*}^0$ and $A_{h^*}$, corresponding to the cases $\lambda^*>0$ and $\lambda^*=0$, the indemnity is undetermined. However, Proposition~\ref{propUVlinearCase1} above provides one possible choice of the optimal indemnity $Y^*$ on these sets, namely, a proportional contract (coinsurance).

\vspace{0.2cm}

%=======================================================================================

\subsubsection{With the No-Sabotage Condition}

In this section, we assume that $Q\circ X^{-1}$ is nonatomic, i.e., $X$ is a continuous random variable for $Q$.

\vspace{0.2cm}

If the set of admissible indemnities is such that the indemnity and the retention function are increasing function of the loss, then according to  Proposition~\ref{propCase2}, the optimal retention $\hat{R}^*=\hat{r}^*(X)\in\hat{\mathcal{I}}$ is such that:
\begin{equation}\label{eqRuLinear}
(\hat{r}^*)'(t)=\begin{cases}
0, &\mbox{if }  P^*_{ac}\circ X^{-1}(A_t)>\lambda^*Q\circ X^{-1}(A_t),\\
\kappa(t),  &\mbox{if }  P^*_{ac}\circ X^{-1}(A_t) = \lambda^*Q\circ X^{-1}(A_t),\\
1,  &\mbox{if } P^*_{ac}\circ X^{-1}(A_t)<\lambda^*Q\circ X^{-1}(A_t),
\end{cases}
\end{equation}
where $A_t:=[t,M]\cap X(A)$, for any $t\in X(A)$, $\lambda^*\geq 0$ and $\kappa$ is a Lebesgue measurable and $[0,1]$-valued function. Let $\underline{t} :=\inf\lbrace t:\, t\in X(A)\rbrace$ and $\overline{t}:=\sup\lbrace t:\, t\in X(A)\rbrace$ be the smallest and the largest value in $X(A)$, respectively. Moreover, let $t_Q:=\sup\lbrace t\geq \underline{t}:\, Q\circ X^{-1}(A_t)>0\rbrace$ and $t_*=\sup\lbrace t\geq \underline{t}:\, P^*_{ac}\circ X^{-1}(A_t)>0\rbrace$. As $P^*_{ac}\ll Q$, then $P^*_{ac}\circ X^{-1}(A_t)=0$, for all $t\geq t_Q$, i.e., $t_*\leq t_Q$.

\vspace{0.2cm}

If $X(A)$ is not connected, then using the notations above, we can extend the definition of the set $A_t$ over $[\underline{t},\overline{t}]$ as follows: if $\tilde{t}\in [\underline{t},\overline{t}]$ such that $\tilde{t}\notin X(A)$, then $A_{\tilde{t}}:=A_{\sup\lbrace t\in X(A):\, t<\tilde{t}\rbrace}$. Observe that $A_{t_n}\supseteq A_{t_{n+1}}$, for $t_n\in [\underline{t},\overline{t}]$, with $Q\circ X^{-1}(A_{\underline{t}})=Q\circ X^{-1}(A)=1$.

\vspace{0.2cm}

Next, we proceed by considering the following cases:
\vspace{0.2cm}
\begin{thmcases}[noitemsep,wide=0pt, leftmargin=\dimexpr\labelwidth + 2\labelsep\relax]
\item If $\lambda^*=0$, then $P^*_{ac}\circ X^{-1}(A_t)>0$, for all $t\in [\underline{t}, t_*]$, and $P^*_{ac}\circ X^{-1}(A_t)=0$, for all $t>t_*$. Hence,
\begin{equation*}
(\hat{r}^*)'(t)=\begin{cases}
0, &\mbox{if }  \underline{t}\leq t\leq t_*,\\
\kappa(t),  &\mbox{if }  t_*<t\leq \overline{t},
\end{cases}
\end{equation*}
for some $\kappa:[\underline{t},\overline{t}]\rightarrow [0,1]$ Lebesgue measurable function. In particular, if $\kappa$ is chosen as $\kappa\equiv 0$, then $\hat{Y}^*=X$, $P^*$-a.s.
\vspace{0.2cm}
\item If $\lambda^*\geq\dfrac{1}{Q\circ X^{-1}(A_{t_Q})}$, then for all $t\in [\underline{t}, t_Q]$, it holds that
\[
P^*_{ac}\circ X^{-1}(A_t)-\lambda^*Q\circ X^{-1}(A_t)\leq P^*_{ac}\circ X^{-1}(A_t)-\lambda^*Q\circ X^{-1}(A_{t_Q})\leq P^*_{ac}\circ X^{-1}(A_t)-1\leq 0.
\]
In this case, let $t_1:=\sup\lbrace t\geq \underline{t}:\, P^*_{ac}\circ X^{-1}(A_t) =1 \rbrace$. If no such $t_1$ exists, then $(\hat{r}^*)'(t)=1$, for all $t\in [\underline{t}, \overline{t}]$, i.e., $\hat{Y}^*=0$, $Q$-a.s. Otherwise,
\begin{equation*}
(\hat{r}^*)'(t)=\begin{cases}
\kappa(t), &\mbox{if }  \underline{t} \leq t\leq t_1,\\
1,  &\mbox{if }  t_1<t\leq \overline{t},
\end{cases}
\end{equation*}
where $\kappa:[\underline{t},\overline{t}]\rightarrow [0,1]$ is a Lebesgue measurable function. Similar to the previous case, one can choose $\kappa\equiv 1$ and thus $\hat{Y}^*=0$, $Q$-a.s.
\vspace{0.2cm}
\item If $\lambda^*\in \bigg(0,\dfrac{1}{Q\circ X^{-1}(A_{t_Q})}\bigg)$, define the function $\varphi:[\underline{t},\overline{t}]\rightarrow \bigg[0, \dfrac{1}{Q\circ X^{-1}(A_{t_Q})}\bigg]$,
$$
\varphi(t):=\dfrac{P^*_{ac}\circ X^{-1}(A_t)}{Q\circ X^{-1}(A_t)}= \dfrac{1-F_{X,P^*_{ac}}(t)}{1-F_{X,Q}(t)}.
$$

Let $\chi_{P_{ac}^*}:=\big\lbrace t \in [0,M] \cap X(A) : \, F_{X,P_{ac}^*} \text{ is continuous in } t \big\rbrace$. Since the set of discontinuities of $F_{X,P_{ac}^*}$ is at most countable, it follows that
$$
[\underline{t}, \overline{t}]\cap \chi_{P_{ac}^*} = \bigcup_{n\in\mathbb{N}} (\underline{t}_n, \overline{t}_n)
$$
is the set of disjoints intervals such that $\varphi_n:=\varphi\big\vert_{(\underline{t}_n, \overline{t}_n)}$ is continuous. The functions $(\varphi_n)_{n\in\mathbb{N}}$ can be extended over a compact interval by taking $\varphi_n(\underline{t}_n):=\lim_{t\searrow \underline{t}_n} \varphi(t)$ and $\varphi_n(\overline{t}_n):=\lim_{t\nearrow \overline{t}_n} \varphi(t)$. Then for each $\varphi_n$, $n\in\mathbb{N}$, the procedure is as follows: we construct a piecewise linear underestimator and a piecewise linear overestimator of $\varphi_n$ over $[\underline{t}_n, \overline{t}_n]$, which allows us to determine the sign of $\varphi_n-\lambda^*$; this sign, in turn, gives the marginal retention function corresponding to $\hat{r}^*$, according to eq.~\eqref{eqRuLinear}. The optimal $\hat{r}^*$ is thus of the form
$$
\hat{r}^*(t) =\sum_{n\in\mathbb{N}}\int_{\underline{t}_n}^{\overline{t}_n} (\hat{r}^*)'(x)\mathbf{1}_{\lbrace x\leq t\rbrace }\,dx, \, \forall t\in X(A).
$$

%\vspace{0.2cm}

Fix $n\in\mathbb{N}$ and consider $\varphi_n$. For a small enough $\Delta>0$, let $\psi_-, \psi_+:[\underline{t}_n, \overline{t}_n]\rightarrow\mathbb{R}_+$ be a $\Delta$-underestimator and $\Delta$-overestimator of $\varphi_n$, respectively and let $G_-(\varphi_n):=\bigcup _{i=1}^m([a_i,b_i], [x^{\min}_i,x^{\max}_i])$, and $G_+(\varphi_n):=\bigcup _{i=1}^m([a_i,c_i], [x^{\min}_i,x^{\max}_i])$ be their associated systems, as defined in Appendix~\ref{appendixDeltaApprox}, for $a_i$ $b_i$, $c_i\in\mathbb{R}$ and a system of breakpoints $\lbrace x^{\min}_i, x^{\max}_i\rbrace_{i=1}^m$ for $\varphi_n$. Figure~\ref{fig:estimator} illustrates the $\Delta$-underestimator and $\Delta$-overestimator for some function $\varphi_n$. According to eq.~\eqref{eqRuLinear}, the marginal retention function switches from $0$ to $1$, depending on the sign of $\varphi_n(x)-\lambda^*$, for all $x\in [\underline{t}_n, \overline{t}_n]$. We observe that if there exists some $i=1,\ldots, m$ such that
\begin{equation}\label{eqLinearIntersection}
\left\{\begin{aligned}
\psi_-(x^{\min}_i) & \leq \lambda^* \leq \psi_-(x^{\max}_i),\\
\psi_+(x^{\min}_i) & \leq \lambda^* \leq \psi_+(x^{\max}_i),
\end{aligned}\right.
\end{equation}
then there exists some $d_i\in [x^{\min}_i, x^{\max}_i]$ such that $\varphi_n(d_i)=\lambda^*$. The structure of $\hat{r}^*$ is determined by a case-by-case analysis, depending on the slope $a_i$ of both $\Delta$-underestimator and $\Delta$-overestimator:
\vspace{0.5cm}
\begin{subcases}[noitemsep,wide=0pt, leftmargin=\dimexpr\labelwidth + 2\labelsep\relax]
\item If $a_i>0$, then $\varphi_n(x)<\lambda^*$, for $x< d_i$ and $\varphi_n(x)>\lambda^*$, for $x>d_i$. In Figure~\ref{fig:estimator} the situation is depicted for instance over the segment $[x^{\min}_1,x^{\max}_1]$,  when $\psi_-(x)=a_1 x+ b_1$ and $\psi_+(x)=a_1 x + c_1$, for $a_1>0$ and $b_1$, $c_1\in\mathbb{R}$.
\vspace{0.5cm}
\item If $a_i<0$, then $\varphi_n(x)>\lambda^*$, for $x< d_i$ and $\varphi_n(x)<\lambda^*$, for $x>d_i$. This is the case over $[x^{\min}_2,x^{\max}_2]$ in Figure~\ref{fig:estimator} , when $\psi_-(x)=a_2 x+ b_2$ and $\psi_+(x)=a_2 x + c_2$, for $a_2<0$ and $b_2$, $c_2\in\mathbb{R}$.
\vspace{0.5cm}
\item If $a_i=0$, $\varphi_n$ is parallel with $x$-axis over $[x^{\min}_i,x^{\max}_i]$. In this case, we only need to check the relation between $\varphi_n(x^{\min}_i)$ and $\lambda^*$: if $\varphi_n(x^{\min}_i)>\lambda^*$, then $(\hat{r}^*)'(x)=0$, while if $\varphi_n(x^{\min}_i)<\lambda^*$, then $(\hat{r}^*)'(x)=1$, for all $x\in [x^{\min}_i,x^{\max}_i]$. Otherwise, if $\varphi_n(x^{\min}_i)=\lambda^*$, then ($\hat{r}^*)'(x)=\kappa(x)$. In particular, $\kappa$ can be chosen to be $0$ over this interval.
\end{subcases}

\begin{figure}[h!]
\center
\includegraphics[scale=2]{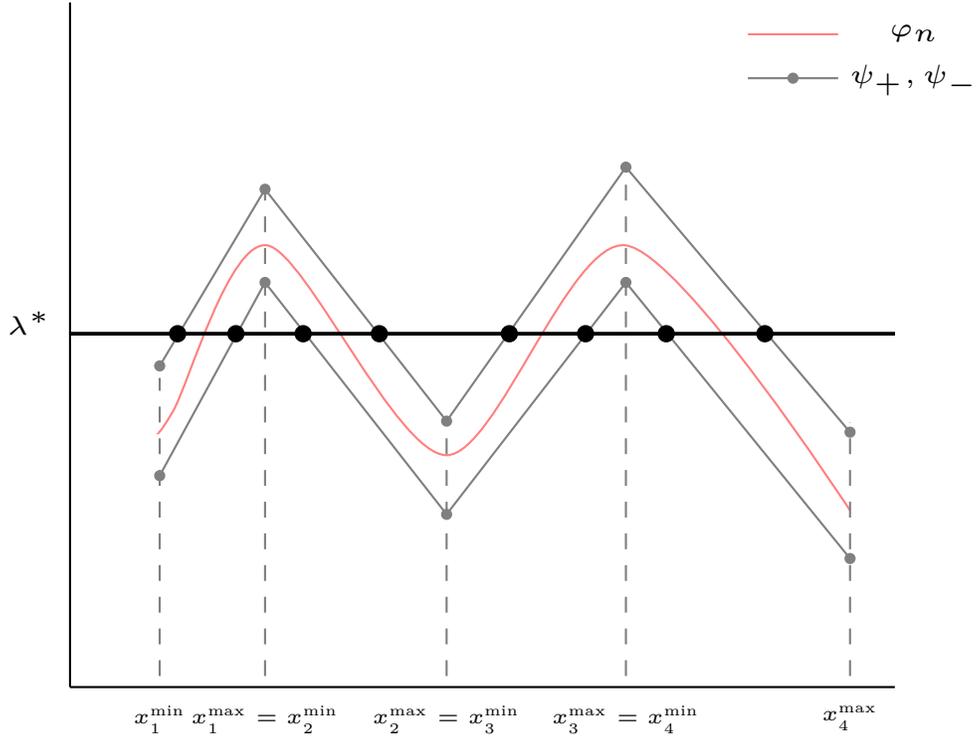}
\caption{$\Delta$-underestimator $\psi_-$ and $\Delta$- overestimator $\psi_+$ of $\varphi_n$.}\label{fig:estimator}
\end{figure}

\vspace{0.4cm}

It follows that the optimal retention function $r^*$ on $[\underline{t}_n, \overline{t}_n]$ is of the form: for each $i=1,\ldots,m$ and $x\in [x^{\min}_i, x^{\max}_i]$,
\begin{equation*}
(\hat{r}^*)'(x) =
\begin{cases}
\begin{rcases}
	\begin{cases}
		1, &\text{if  } x^{\min}_i\leq x<d_i \\
  		0, &\text{if  } d_i\leq x\leq x^{\max}_i \\
	\end{cases}, \text{ for } x^{\min}_i\leq d_i\leq x^{\max}_i
\end{rcases},  &\text{if \eqref{eqLinearIntersection} holds and }  a_i> 0, \\
\begin{rcases}
	\begin{cases}
		0, &\text{if  } x^{\min}_i\leq x<d_i \\
  		1, &\text{if  } d_i\leq x\leq x^{\max}_i \\
	\end{cases}, \text{ for } x^{\min}_i\leq d_i\leq x^{\max}_i
\end{rcases},  &\text{if \eqref{eqLinearIntersection} holds and }  a_i< 0, \\
\begin{cases}
		1, &\text{if  } \varphi_n(x^{\min}_i)<\lambda^* \\
  		0, &\text{if  } \varphi_n(x^{\min}_i)\geq \lambda^* \\
\end{cases}, \text{ if }a_i\, a_{i-1}=0, \\
\;\;\; 0, \;\;\;\text{if \eqref{eqLinearIntersection} doesn't hold, but } a_{i-1}>0, \\
\;\;\; 1, \;\;\; \text{if \eqref{eqLinearIntersection} doesn't hold, but } a_{i-1}<0.
\end{cases}
\end{equation*}
\end{thmcases}

\vspace{0.2cm}

%======================================================================================

\begin{remark}[Coherent Risk Measures]\label{RemCoherent}
Since $R=X-Y$ and since the location of the minimum is not affected by adding and subtracting constants to the objective function in Problem~\eqref{optimUVlinear}, Problem~\eqref{optimUVlinear} can readily be reformulated as follows:
\begin{equation*}\label{optimUVlinear3}
\left\{
\begin{aligned}
& \underset{Y\in \mathcal{F}}{\inf}\,\underset{P\in\mathcal{C}}{\sup}
& & \mathbb{E}_P [-(W_0-X+Y-\Pi_0)]\\
& \text{s.t.}
& & \mathbb{E}_Q [Y]\leq \mathbb{E}_Q [X]-\tilde{\Pi}_0.
\end{aligned}\tag{$P_3$}
\right.
\end{equation*}
Now, by Artzner et al.\ \cite{Artzneretal1999} and Delbaen \cite{Delbaen2002}, the expression $\sup_{P\in\mathcal{C}} \mathbb{E}_P [-(W_0-X+Y-\Pi_0)]$ is equivalent to a coherent risk measure of the random gain $W_0-X+Y-\Pi_0$. Hence, Problem~\eqref{optimUVlinear3} can be seen as the objective to minimize a coherent risk measure of final wealth under a premium budget constraint.
\end{remark}

\vspace{0.4cm}

%======================================================================================

\subsubsection{Examples}

For the following examples, we assume that $X$ is a continuous random variable for $Q$.

\begin{example}\label{exLinearUalternative1}
In this example, we consider again the ambiguity set $\mathcal{C}_{\mathcal{W}}$ in Example~\ref{exGeneralU}, i.e.,
\begin{equation*}
\mathcal{C}_{\mathcal{W}}=\bigg\lbrace P\in ca^{+}_{1}(\Sigma):\, \dfrac{dP}{dQ}=\dfrac{w(X)}{\int w(X)dQ},\, w\in \mathcal{W}\bigg\rbrace,
\end{equation*}
\noindent where $\mathcal{W}$ is defined in Example~\ref{exGeneralU}.

\vspace{0.2cm}

First we analyze the case when $\mathcal{F}=\mathcal{I}$, as defined in eq.~\eqref{FeasibWithoutNSC}. Observe that the worst-case measure $P^*\in\mathcal{C}_{\mathcal{W}}$ satisfies $P^*\ll Q$. Based on the same argument as in Example~\ref{exGeneralU}, the set $A_{h^*}$ in Proposition~\ref{propUVlinearCase1} is $A_{h^*}=X^{-1}([0,a])$, for some $a\geq 0$. Now, according to Proposition~\ref{propUVlinearCase1}, the structure of $R^*=r^*(X)\in\mathcal{I}$ depends on the sign of $\xi^*(x)-\lambda^*$, where $\lambda^*\geq 0$ is the optimal Lagrange multiplier of Problem~\eqref{optimUVlinear}. If $\lambda^*>0$, then there exists some $d_2\geq d_1\geq a$ such that $\xi^*(x)<\lambda^*$, for $x< d_1$ and $\xi^*(x)>\lambda^*$, for $x>d_2$. Otherwise,
$\lambda^*=0$, and in this case, $a>0$. With these observations, the optimal retention $R^*$ of Problem~\eqref{optimUVlinear} is given by:
\begin{equation*}
R^*=\begin{cases}
\begin{rcases}
	\begin{cases}
		X\mathrm{1}_{X^{-1}([0,d_1])} + c\,X\mathrm{1}_{X^{-1}([d_1,d_2])}  , &\text{if  } \mathbb{E}_Q[X\mathrm{1}_{X^{-1}([0,d_1])}]<\tilde{\Pi}_0 \\
  		X\mathrm{1}_{X^{-1}([0,d_1])} , &\text{if  } \mathbb{E}_Q[X\mathrm{1}_{X^{-1}([0,d_1])}]=\tilde{\Pi}_0\\
	\end{cases}
\end{rcases}, & \text{if }\mathbb{E}_Q[X\mathrm{1}_{X^{-1}([0,a])}]<\tilde{\Pi}_0, \\
\;\;\; c_0X\mathrm{1}_{X^{-1}([0,a])}, \;\;\; \text{if } \mathbb{E}_Q[X\mathrm{1}_{X^{-1}([0,a])}]\geq\tilde{\Pi}_0,
\end{cases}
\end{equation*}
where $c:=\dfrac{\tilde{\Pi}_0-\mathbb{E}_Q[X\mathrm{1}_{X^{-1}([0,d_1])}]}{\mathbb{E}_Q[X\mathrm{1}_{X^{-1}([d_1,d_2])}]}\in (0,1]$, $c_0:=\dfrac{\tilde{\Pi}_0}{\mathbb{E}_Q[X\mathrm{1}_{X^{-1}([0,a])}]}$ and $\tilde{\Pi}_0:=\mathbb{E}_Q[X]-(1+\rho)^{-1}\Pi_0$.

\vspace{0.2cm}

Now, suppose that $\mathcal{F}=\hat{\mathcal{I}}$, as defined in eq.~\eqref{FeasibWithNSC}. Then according to Proposition~\ref{propCase2}, the optimal retention $\hat{R}^*=\hat{r}^*(X)$ is determined by
\begin{equation*}
(\hat{r}^*)'(t)=\begin{cases}
0, &\mbox{if }  P^*\circ X^{-1}([t,M])>\lambda^*Q\circ X^{-1}([t,M]),\\
\kappa(t),  &\mbox{if }  P^*\circ X^{-1}([t,M]) = \lambda^*Q\circ X^{-1}([t,M]),\\
1,  &\mbox{if } P^*\circ X^{-1}([t,M])<\lambda^*Q\circ X^{-1}([t,M]),
\end{cases}
\end{equation*}
for all $t\in [0,M]$.
\vspace{0.2cm}

Define the function $\varphi:[0,M]\rightarrow \mathbb{R}_+$, $\varphi(t):=\dfrac{P^*_{ac}\circ X^{-1}([t,M])}{Q\circ X^{-1}([t,M])}=\dfrac{1-F_{X,P^*}(t)}{1-F_{X,Q}(t)}$. Let $Z_1$ and $Z_2$ be two random variables, whose probability distribution are $F_{X,P^*}$ and $F_{X,Q}$, respectively, and let $f_{X,P^*}$ and $f_{X,Q}$ be the corresponding probability densities. Since $f_{X,P^*}(t)/ f_{X,Q}$ is increasing in $t$, then $Z_1$ is smaller than $Z_2$ in the likelihood ratio order, i.e., $Z_1\leqLR Z_2$, and hence, by \cite[Theorem 1.C.1.]{ShakedShanthikumar2007}, $Z_1$ is smaller than $Z_2$ in the hazard ratio order, i.e., $Z_1\leqHR Z_2$. Thus, the function $\varphi$ is increasing in $t$. For $\lambda^*\in\mathbb{R}_+$, there exists some $d_1\leq d_2\in [0,M]$ such that $\varphi(t)=\lambda^*$, for $t\in [d_1,d_2]$. In this case, the optimal retention is of the following form:
\begin{equation*}
(\hat{r}^*)'(t)=\begin{cases}
1, &  0\leq t<d_1,\\
\kappa(t),  & d_1\leq t\leq d_2,\\
0,  &\mbox{if } d_2<t,
\end{cases}
\end{equation*}
for some Lebesgue measurable and $[0,1]$-valued function $\kappa$. In particular, for $\kappa(t)\equiv 0$, the optimal indemnity is $\hat{Y}^*=\max(X-d,0)$, for some $d\in [0,M]$.
\end{example}

\vspace{0.2cm}

\begin{example}\label{exLinearUalternative2}
Assume that the state-space $S$ is a Polish space. Let $T_1, \ldots, T_n$ be $n$ distortion functions such that $T_i:[0,1]\rightarrow [0,1]$ is a strictly increasing, strictly concave function satisfying $T_i(0)=0$ and $T_i(1)=1$, for $i=1,\ldots,n$. As before, let $F_{X,Q}$ be the continuous cdf corresponding to the insurer's belief $Q$. Moreover, let $P_1,\ldots, P_n$ be $n$ probability measures, whose corresponding cdf's are defined as $F_{X,P_i}:=1-T_i(1-F_{X,Q})$, $i=1,\ldots, n$. In the presence of model ambiguity, DM considers the ambiguity set $\mathcal{C}$ to be the convex hull of these probabilities, i.e., $\overline{\mathcal{C}}:=conv(P_1, \ldots, P_n)$. Observe that, by construction, $P\sim Q$, for any $P\in\overline{\mathcal{C}}$. As $S$ is a Polish space, according to Prohorov's Theorem, $\overline{\mathcal{C}}$ is weak$^*$-compact if it is closed and tight. Since the generating set is finite, it follows that $\overline{\mathcal{C}}$ is a closed set. To prove the tightness property, observe that $\mathbb{E}_{P_i}[W_0-X+Y-\Pi_0]\leq W_0$, for $i=1,\ldots,n$ and for all $Y\in\mathcal{F}$.

\vspace{0.2cm}

Next, observe that for any $P\in\overline{\mathcal{C}}$, there exists some $\alpha_1,\ldots,\alpha_n\in [0,1]$ with $\sum_{i=1}^n \alpha_i=1$ such that $\mathbb{E}_P[W_0-X+Y-\Pi_0]=\sum_{i=1}^n\alpha_i\mathbb{E}_{P_i}[W_0-X+Y-\Pi_0]\leq K$. For any $\delta>0$ and any $P\in\overline{\mathcal{C}}$,
\[
P\bigg( W_0-X+Y-\Pi_0\leq \dfrac{K}{\delta}\bigg)\geq 1-\dfrac{\mathbb{E}_P[W_0-X+Y-\Pi_0]}{K/\delta}\geq 1-\delta,
\]
where the first inequality holds by Markov's inequality. Hence, $\overline{\mathcal{C}}$ is tight and thus weak$^*$-compact.

\vspace{0.2cm}

Next, consider Problem~\eqref{optim} with the ambiguity set $\overline{\mathcal{C}}$ defined above. It is a problem of minimization of a linear function in $P$, and thus minimization of a concave function in $P$ over a convex and compact set. As the ambiguity set is not empty, the optimal solution is an extreme point of $\overline{\mathcal{C}}$, i.e., $P^*\in\lbrace P_1,\ldots, P_n\rbrace$.

\vspace{0.2cm}

In this setting, let $\mathcal{F}=\mathcal{I}$ be the admissible set of indemnities in Proposition~\ref{propUVlinearCase1}. Since $h^*>0$ on $S$, the optimal retention in Example~\ref{exLinearUalternative1} simplifies further to:
\begin{equation*}
R^*=\begin{cases}
		X\mathrm{1}_{X^{-1}([0,d_1])} + c\,X\mathrm{1}_{X^{-1}([d_1,d_2])}  , &\text{if  } \mathbb{E}_Q[X\mathrm{1}_{X^{-1}([0,d_1])}]<\tilde{\Pi}_0, \\
  		X\mathrm{1}_{X^{-1}([0,d_1])} , &\text{if  } \mathbb{E}_Q[X\mathrm{1}_{X^{-1}([0,d_1])}]=\tilde{\Pi}_0,
	\end{cases}
\end{equation*}
where $c\in (0,1]$ is given in Proposition~\ref{propUVlinearCase1}.

\vspace{0.2cm}

For $\mathcal{F}=\hat{\mathcal{I}}$ as defined in eq.~\eqref{FeasibWithNSC} and $u(x)\equiv x$, Proposition~\ref{propCase2} states that the optimal retention $\hat{R}^*=\hat{r}^*\circ X$ is such that
\begin{equation*}
(\hat{r}^*)'(t)=\begin{cases}
0, &\mbox{if }  1-F_{X,P^*}(t)>\lambda^*(1-F_{X,Q}(t)),\\
\kappa(t),  &\mbox{if }  1-F_{X,P^*}(t)=\lambda^*(1-F_{X,Q}(t)),\\
1,  &\mbox{if } 1-F_{X,P^*}(t)<\lambda^*(1-F_{X,Q}(t)),
\end{cases}
\end{equation*}
for all $t\in [0,M]$.
\vspace{0.2cm}

Define the continuous function $\varphi:[0,M]\rightarrow \mathbb{R}_+$, $\varphi(t):=\dfrac{1-F_{X,P^*}(t)}{1-F_{X,Q}(t)}$. Let $f_{X,P^*}$ and $f_{X,Q}$ be the probability densities of $F_{X,P^*}$ and $F_{X,Q}$, respectively. Note that for any distortion $T\in \lbrace T_1,\ldots, T_n\rbrace$, $\dfrac{f_{X,P^*}}{f_{X,Q}}= T'(1-F_{X,Q}(t))$ is increasing in $t$. Now let $Z_1$ and $Z_2$ be two continuous random variables such that $Z_1\sim F_{X,P^*}$ and $Z_2\sim F_{X,Q}$. Similar to Example \ref{exLinearUalternative1}, $Z_1\leqHR Z_2$, and thus, the function $\varphi$ is increasing in $t$. For $\lambda^*\in\mathbb{R}_+$, there exists some $d_1\leq d_2\in [0,M]$ such that $\varphi(t)=\lambda^*$, for $t\in [d_1,d_2]$. In this case, the optimal retention is of the following form:
\begin{equation*}
(\hat{r}^*)'(t)=\begin{cases}
1, &  0\leq t<d_1,\\
\kappa(t),  & d_1\leq t\leq d_2,\\
0,  &\mbox{if } d_2<t,
\end{cases}
\end{equation*}
for some Lebesgue measurable and $[0,1]$-valued function $\kappa$. If $F_{X,Q}$ is strictly increasing, then $d_1=d_2$.
\end{example}

\vspace{0.4cm}

%====================================================================================
%====================================================================================
%====================================================================================

\section{Numerical Examples}\label{section5}

This section presents numerical examples that illustrate the structure of the optimal indemnity $\hat{Y}^*$ obtained in Sections~\ref{section3} and~\ref{section4}, when the ambiguity set $\mathcal{C}$ is constructed as a specific neighbourhood around a \textit{reference/ baseline distribution}. Throughout this analysis, we assume that the underlying space $S$ is a Polish space, equipped with its Borel sigma-algebra.

\vspace{0.2cm}

As above, $X$ is a nonnegative random variable representing the insurable loss, whose true distribution may be unknown. The insurer's belief $Q\in ca^+_1(\Sigma)$ regarding the loss $X$ can be the empirical distribution, derived from experts' opinion or estimated using standard statistical tools. The DM's ambiguity regarding the realizations of $X$ is described by a $\delta$-neighbourhood around $Q$ defined as:
\begin{equation}\label{defCparametric}
\mathcal{C}_\delta:=\lbrace P\in ca^+_1(\Sigma) :\, \mathrm{d}(P, Q)\leq\delta\rbrace,
\end{equation}
where $\mathrm{d}:ca^+_1(\Sigma)\times ca^+_1(\Sigma)\rightarrow\mathbb{R}_+$ is some \textit{discrepancy measure} between probability measures $P$ and $Q$, and $\delta>0$ is a \textit{tolerance level/ambiguity radius}. The mapping $\mathrm{d}$ satisfies $\mathrm{d}(P,Q)=0$ if and only if $P=Q$. It is worth mentioning that the worst-case distribution $P^*$ depends not only on the choice of $\mathrm{d}$, but also on the ambiguity radius $\delta$. In general, the size of $\mathcal{C}_\delta$ is connected to the amount of observations available: if $\delta$ is close to zero, the impact of ambiguity is negligible; while large values of $\delta$ indicate high levels of model uncertainty. The question of how to optimally choose the ambiguity radius is an ongoing stream of research in robust optimization. The data-driven approach estimates $\delta$ either by evaluating the discrepancy between the empirical model $\widehat{P}_n$ and the calibrated model, or using measure concentration inequalities to target a certain confidence level $\beta\in (0,1)$, i.e., $\mathbb{P}(\mathrm{d}(P^*,\widehat{P}_n)\leq\delta)\geq 1-\beta$ (see \cite[Theorem 3.4 and the discussion afterwards]{EsfahaniKuhn2018}, \cite[Section 5.1]{BlanchetHeMurthy2020}). We investigate this method in Example~\ref{exWD}, when the ambiguity set is constructed using the Wasserstein metric. Another approach interprets $\delta$ as the degree of ambiguity about the reference model and thus argues that this choice depends on the risk preferences of market participants (e.g., \cite{BreuerCsiszar2016,Wozabal2012}). In Example~\ref{exRenyi}, we follow the latter approach and solve Problem~\eqref{optimV} for different levels of ambiguity. This allows us to analyze the impact of ambiguity on the optimal indemnity $\hat{Y}^*$ and the worst-case distribution $P^*$.

\vspace{0.2cm}

The following observation characterizes the change in the DM's expected utility, as a function of the ambiguity radius $\delta$. This dynamic is later illustrated in Figure~\ref{fig:RenyiCE} in Example~\ref{exRenyi}.

\vspace{0.2cm}

%======================================================================================

\begin{remark}\label{remC}
For a fixed premium budget $\Pi_0>0$, let $\hat{\mathcal{I}}_0$ (defined in eq.~\eqref{FeasibHatI}) be the feasible set of indemnities in Theorem~\ref{propCase2V}. Moreover, for a discrepancy measure $\mathrm{d}$ and some ambiguity radii $\delta_1\leq\delta_2$, let $\mathcal{C}_{\delta_1}$ and $\mathcal{C}_{\delta_2}$ be the corresponding ambiguity sets, as defined in eq.~\eqref{defCparametric}. Let $(\hat{Y}_1^*,P_1^*)$ and $(\hat{Y}_2^*,P_2^*)$ be the saddle points of Problems~\eqref{optimV}, for $\mathcal{C}_{\delta_1}$ and $\mathcal{C}_{\delta_2}$, respectively. It holds that
\[
\mathbb{E}_{P_2^*}[u(W_0-X+\hat{Y}_2^*+\Pi_0)] \leq  \mathbb{E}_{P_1^*}[u(W_0-X+\hat{Y}_2^*+\Pi_0)] \leq \mathbb{E}_{P_1^*}[u(W_0-X+\hat{Y}_1^*+\Pi_0)],
\]
where the first inequality follows from $\mathcal{C}_{\delta_1}\subseteq\mathcal{C}_{\delta_2}$, as $\delta_1\leq \delta_2$. Hence, for increasing values of $\delta$, the optimal DM's expected utility decreases.
\end{remark}

%======================================================================================

\vspace{0.2cm}

\begin{example}[Wasserstein ambiguity set]\label{exWD}
We examine the structure of the saddle point $(\hat{Y}^*,P^*)$ in the setting of Problem~\eqref{optim}, when the insurer is risk-neutral and the admissible set of indemnities is $\mathcal{F}=\hat{\mathcal{I}}$ as defined in eq.~\eqref{FeasibWithNSC}. The DM's ambiguity about the realizations of $X$ is characterized by the ambiguity set $\mathcal{C}_\delta$ given by
\[
\mathcal{C}_\delta^{\WD} :=\lbrace P\in ca^+_1(\Sigma) :\, \WD(P,Q)\leq \delta\rbrace,
\]
where $\WD$ is the Wasserstein distance on $\mathbb{R}$, with the $L_1$-norm being the underlying metric (e.g., \cite{Vallender1974}):
\[
\WD(P,Q):=\int_{\mathbb{R}} \big\vert F_{X,P}(x)-F_{X,Q}(x)\big\vert\,dx = \int_0^1\big\vert F_{X,P}^{-1}(t)-F_{X,Q}^{-1}(t)\big\vert\,dt.
\]
The Wasserstein distance is a metric on $ca^+_1(\Sigma)$, satisfying $\WD(P,Q)=0$ if and only if $P=Q$. With this metric, $\mathcal{C}_\delta^{\WD}$ is convex and weak$^*$-compact. See Villani \cite{Villani2008} for further properties of the Wasserstein distance.

\vspace{0.2cm}

Let $Q$ be the insurer's belief regarding the loss $X$.
%the random variable $X$ represent the loss faced by DM, whose underlying distribution is unknown; however, an estimated distribution $\widehat{Q}$ is available.
In this example, we assume that $F_{X,Q}$ is a truncated Generalized Pareto distribution with an upper bound $M$, with shape and scale parameters $0.3$ and $5$, respectively. Next, we simulate from the distribution $F_{X,Q}$, and obtain the empirical distribution. We then construct a piecewise linear approximation $F_{X,\widehat{Q}}$ of this empirical distribution, with given knots $\lbrace x_1, \ldots, x_n\rbrace $, where the partition $0=x_1<\cdots<x_n=M$ is chosen arbitrarily, but kept fixed all throughout. That is, $F_{X,\widehat{Q}}$ is given by a system $\big\lbrace [x_i,x_{i+1}], [F_{X,\widehat{Q}}(x_i), F_{X,\widehat{Q}}(x_{i+1})]\big\rbrace_{i=1}^{n-1}$. Note that by construction, $F_{X,\widehat{Q}}(x_n)=1$. The corresponding density $\widehat{\mathbf{q}}=[\widehat{q}_1,\ldots,\widehat{q}_{n-1}]^\top$ is piecewise constant on each interval $[x_i,x_{i+1}]$, for $i=1,\ldots, n-1$. More precisely, $\widehat{q}_i$ is the slope of the line passing through the points $(x_i,F_{X,\widehat{Q}}(x_i))$ and $(x_{i+1},F_{X,\widehat{Q}}(x_{i+1}))$, i.e., $\widehat{q}_i=\dfrac{F_{X,\widehat{Q}}(x_{i+1})-F_{X,\widehat{Q}}(x_i)}{x_{i+1}-x_i}$, for $i=1,\ldots, n-1$.

\vspace{0.2cm}

This representation of $\widehat{Q}$ allows us to compute the Wasserstein distance between $F_{X,P}$ and $F_{X,\widehat{Q}}$, and thus to characterize the alternative distributions via the system $\big\lbrace F_{X,P}(x_1),\ldots, F_{X,P}(x_n)\big\rbrace$, for the same segments $[x_i,x_{i+1}]$, $i=1,\ldots, n-1$. For two such distributions, the Wasserstein distance is the sum of the areas of the trapezoids with corners $\big\lbrace F_{X,\widehat{Q}}(x_i), F_{X,\widehat{Q}}(x_{i+1}), F_{X,P}(x_i), F_{X,P}(x_{i+1})\big\rbrace$ formed by $F_{X,P}$ and $F_{X,\widehat{Q}}$ (see the shaded area in Figure~\ref{fig:WDcomputation}), i.e.,
\begin{equation}\label{eqWD}
\WD(P,\widehat{Q})=\dfrac{1}{2}\sum_{i=1}^{n-1} (x_{i+1}-x_i)\phi(F_{X,P}(x_i)-F_{X,\widehat{Q}}(x_i), F_{X,P}(x_{i+1})-F_{X,\widehat{Q}}(x_{i+1})),
\end{equation}
where the function $\phi:[-1,1]^2\rightarrow\mathbb{R}_+$ defined below is convex in each component (e.g., \cite{PflugTimoninaHochrainer2017}):
\[
\phi(a,b)=\begin{cases}
\vert a\vert +\vert b\vert, &\mbox{if }  ab\geq 0,\\
\dfrac{a^2+b^2}{\vert a\vert +\vert b\vert},  &\mbox{otherwise}.
\end{cases}
\]

\begin{figure}[h]
\center
\includegraphics[scale=1]{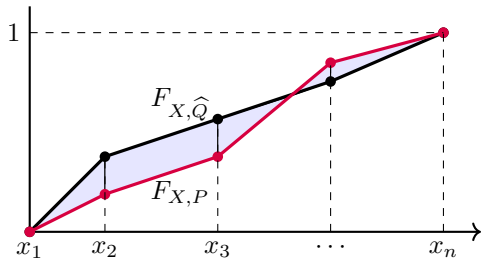}
\caption{Computation of Wasserstein distance between piecewise linear probability distributions $F_{X,\widehat{Q}}$ and $F_{X,P}$.} \label{fig:WDcomputation}
\end{figure}

\noindent The alternative measure $P$ is represented by an $(n-1)$-dimensional vector $\mathbf{p}=[p_1,\ldots, p_{n-1}]^\top$, where $p_i\in [0,1]$ is the constant forming the piecewise constant density of $F_{X,P}$. More precisely, the alternative cdf $F_{X,P}$ will be linear on each interval $[x_i,x_{i+1}]$, and will differ from $F_{X,\widehat{Q}}$ only in the cumulative probabilities $F_{X,P}(x_i)$. Thus, $p_i$ will be the slope of the line passing through the points $(x_i,F_{X,P}(x_i))$ and $(x_{i+1},F_{X,P}(x_{i+1}))$. The representation of $F_{X,P}$ is shown in Figure~\ref{fig:WDcomputation}. Therefore, the variable $\mathbf{p}$ must satisfy $\mathbf{p}^\top\mathbf{A}_{n-1}\mathbf{x}=1$, where for $i=1,\ldots, n-1$, the matrix $\mathbf{A}_i\in\mathbb{R}^{(n-1)\times n}$ is defined as follows:
\begin{equation}\label{defMatrixAi}
\mathbf{A}_i:=
\begin{blockarray}{ccccccccccc}
\begin{block}{(cccccccccc)c}
-1 & 1 & 0 & 0 & \ldots &  & \ldots &  & \ldots & 0 &  &\\
0 & -1 & 1 & 0 & \ldots &  & \ldots &  & \ldots & 0 &  & \\
 & \ddots &  & \ddots &  & \ddots &  & \ddots &  & \ddots &  &\\
\ldots &  & \ldots &  & \ldots & 0 & -1 & 1 & 0 & \ldots & \leftarrow i\text{-th row}. & \\
\ldots &  & \ldots &  & \ldots & 0 & 0 & 0 & 0 & \ldots &  & \\
 & \ddots &  & \ddots &  & \ddots &  & \ddots &  & \ddots &  & \\
\end{block}
\end{blockarray}
\end{equation}

\vspace{0.2cm}

\noindent Using the matrix above, $F_{X,P}$ can also be represented via $F_{X,P}(x_i)=\mathbf{p}^\top\mathbf{A}_i\mathbf{x}$, for $i=1,\ldots,n-1$.

\vspace{0.2cm}

Next, to identify the optimal $\hat{Y}^*$, we follow the equivalent formulation of Problem~\eqref{optim} and describe the decision variable in terms of retention function $\hat{R}$. We assume that $\hat{R}=\hat{r}\circ X$ is piecewise linear between the segments $[x_i,x_{i+1}]$, $i=1,\ldots, n-1$, and describe the corresponding values $\hat{r}(x_i)$ as $\mathbf{\hat{r}}=[\hat{r}_1,\ldots,\hat{r}_n]^\top\in\mathbb{R}^n_+$. Moreover, the variable $\mathbf{\hat{r}}$ is feasible for Problem~\eqref{optim} if $\mathbf{0}\leq \mathbf{\hat{r}}\leq\mathbf{x}$ and $\mathbf{0}\leq \mathbf{A}_{n-1}\mathbf{\hat{r}}\leq \mathbf{A}_{n-1}\mathbf{x}$, where the matrix $\mathbf{A}_{n-1}$ is defined in eq.~\eqref{defMatrixAi}.

\vspace{0.2cm}

With these representations for $\mathbf{p}$ and $\mathbf{\hat{r}}$, the objective function in~\eqref{optim} is approximated using the trapezoid rule on intervals $[x_i,x_{i+1}]$, $i=1,\ldots, n-1$, as follows:
\begin{equation*}
\begin{split}
\int_0^M u(W_0-\hat{r}(x)-\Pi_0)dF_{X,P}(x) & = \sum_{i=1}^{n-1} \int_{x_i}^{x_{i+1}} u(W_0-\hat{r}(x)-\Pi_0)p_i\, dx \\
& \approx \dfrac{1}{2}\sum_{i=1}^{n-1} \big( u(W_0-\hat{r}_i-\Pi_0) + u(W_0-\hat{r}_{i+1}-\Pi_0)\big)(x_{i+1}-x_i)p_i.
\end{split}
\end{equation*}

\vspace{0.2cm}

The trapezoid rule allows us to determine the number of segments $n-1$ used in the linear approximation of $F_{X,Q}$, by specifying \textit{a priori} the target error bound. The error bound for this method is given by
\begin{equation}\label{eqErrorBound}
\vert\epsilon\vert \leq \dfrac{K(x_n-x_1)^3}{12n^2},
\end{equation}
where the constant $K$ satisfies $\vert u''(W_0-x-\Pi_0)\vert\leq K$. As example, let DM's utility be $u(x)=(1-\exp(-\gamma x))/\gamma$, for $\gamma=0.03$, the initial wealth $W_0=250$ and the premium budget $\Pi_0=4$. By direct computation, we get
\[
\vert u''(W_0-\hat{r}-\Pi_0)\vert=\gamma e^{-\gamma(W_0-\Pi_0)}e^{-\gamma \hat{r}}\leq \gamma e^{-\gamma(W_0-\Pi_0)}:=K.
\]
As before, we select $x_1=0$ and $x_n=M$, where the upper bound of the random variable $X$ is equal to $M=W_0-\Pi_0$. If the largest possible error between the objective function in~\eqref{optim} and its approximation is $\epsilon=10^{-4}$, then according to eq.~\eqref{eqErrorBound}, $n\approx 200$. Therefore, to solve~\eqref{optim}, we choose $n=200$ points of approximation $\lbrace x_1,\ldots, x_n\rbrace$.

\vspace{0.2cm}

The error introduced by solving \eqref{optim} in terms of $\widehat{Q}$ instead of $Q$ can be used to estimate the ambiguity radius $\delta$. The estimator $\delta=\delta_n$ depends on the number of piecewise linear segments and thus, it is informed by the data. In particular, we propose to approximate $\delta_n$ as
\[
\delta_n:=\big\vert\mathbb{E}_{\widehat{Q}}[X]-\mathbb{E}_{Q}[X]\big\vert=\WD(\widehat{Q},Q),
\]
as $F_{X,\widehat{Q}}$ dominates $F_{X,Q}$ in the first stochastic order. In our setting, $\delta_n$ is estimated to be around $0.3$.

\vspace{0.2cm}

In sum, Problem~\eqref{optim} is approximated by the following problem:
\begin{equation}\label{optimWasserstein}
\left\{
\begin{aligned}[3]
    & \underset{\quad\ \mathbf{\hat{r}}\in\mathbb{R}^n_+ \ \mathbf{p}\in [0,1]^{n-1}{}}{\max\quad \min}   &  & \dfrac{1}{2}\displaystyle{\sum_{i=1}^{n-1}} \big( u(W_0-\hat{r}_i-\Pi_0) + u(W_0-\hat{r}_{i+1}-\Pi_0)\big)(x_{i+1}-x_i)p_i {} \\
    & \makebox(80,0){s.t.}		 	&       & \mathbf{0}\leq \mathbf{A}_{n-1}\mathbf{\hat{r}}\leq  \mathbf{A}_{n-1}\mathbf{x} {}, \\
    &                                             &       & \mathbf{0}\leq \mathbf{\hat{r}} \leq \mathbf{x} {},\\
    &                                             &       & -\widehat{\mathbf{q}}^\top\mathbf{A}_{n-1}\mathbf{\hat{r}}\leq -2\tilde{\Pi}_0{}, \\
    &												&       & \WD(F_{X,P},F_{X,\widehat{Q}})\leq\delta_n,  \:  F_{X,P}(x_i)=\mathbf{p}^\top\mathbf{A}_i\mathbf{x}, \, i=1,\ldots,n{}, \\
    &                                             &       & \mathbf{p}^\top\mathbf{A}_{n-1}\mathbf{x}=1 {},
\end{aligned}\tag{P$_1^{n}$}
\right.
\end{equation}
where $\tilde{\Pi}_0=\dfrac{1}{2}\widehat{\mathbf{q}}^\top \mathbf{A}_{n-1}\mathbf{x} - (1+\rho)^{-1}\Pi_0$ and $\WD(F_{X,P},F_{X,\widehat{Q}})$ is computed as in eq.~\eqref{eqWD}. The objective function in~\eqref{optimWasserstein} is concave in $\mathbf{\hat{r}}$ and linear in $\mathbf{p}$, while the constraints are convex in $\mathbf{p}$ and linear in $\mathbf{\hat{r}}$. Problem~\eqref{optimWasserstein} is solved via successive convex programming (SCP -- see \cite{PflugPicher2014}). The idea is to approximate the infinite dimensional ambiguity set $\mathcal{C}_\delta^{\WD}$ by a finitely generated set $\mathcal{P}^{(m)}:=\lbrace \widehat{\mathbf{q}},\mathbf{p}^{(1)}, \ldots, \mathbf{p}^{(m)}\rbrace$, obtained iteratively from solving the inner problem in~\eqref{optimWasserstein}. The algorithm starts with $m=0$, $\mathcal{P}^{(m)}:=\mathcal{P}^{(0)} = \lbrace \widehat{\mathbf{q}}\rbrace$, and solves the outer problem:

\begin{equation}\label{optimWassersteinOuter}
\left\{
\begin{aligned}[3]
    & \underset{\; \mathbf{\hat{r}}\in\mathbb{R}^n_+ \; \mathbf{p}\in\mathcal{P}^{(m)} {}}{\max\quad \min}   &  & \dfrac{1}{2}\displaystyle{\sum_{i=1}^{n-1}} \big( u(W_0-\hat{r}_i-\Pi_0) + u(W_0-\hat{r}_{i+1}-\Pi_0)\big)(x_{i+1}-x_i)p_i {} \\
    & \makebox(70,5){s.t.}		 	&       & \mathbf{0}\leq \mathbf{A}_{n-1}\mathbf{\hat{r}}\leq  \mathbf{A}_{n-1}\mathbf{x} {}, \\
    &                                             &       & \mathbf{0}\leq \mathbf{\hat{r}} \leq \mathbf{x} {},\\
    &                                             &       & -\widehat{\mathbf{q}}^\top\mathbf{A}_{n-1}\mathbf{\hat{r}}\leq -2\tilde{\Pi}_0{}.
\end{aligned}\tag{P$_{1,\text{outer}}^{n}$}
\right.
\end{equation}
The solution $\mathbf{\hat{r}}^{(m)}:=\mathbf{\hat{r}}^{(1)}$ acts as input for the inner problem:
\begin{equation}\label{optimWassersteinInner}
\left\{
\begin{aligned}[2]
    & \underset{\mathbf{p}\in [0,1]^{n-1}{}}{\min}   &  & \dfrac{1}{2}\displaystyle{\sum_{i=1}^{n-1}} \big( u(W_0-\hat{r}_i^{(m)}-\Pi_0) + u(W_0-\hat{r}_{i+1}^{(m)}-\Pi_0)\big)(x_{i+1}-x_i)p_i {} \\
    & \makebox(40,10){s.t.}		 	&       & \WD(F_{X,P},F_{X,\widehat{Q}})\leq\delta_n,  \,  F_{X,P}(x_i)=\mathbf{p}^\top\mathbf{A}_i\mathbf{x}, \, i=1,\ldots,n{}, \\
    &                                             &       & \mathbf{p}^\top\mathbf{A}_{n-1}\mathbf{x}=1 {}.
\end{aligned}\tag{P$_{1,\text{inner}}^{n}$}
\right.
\end{equation}
The new $\mathbf{p}^{(m+1)}:=\mathbf{p}^{(1)}$ is added to the discrete set, i.e., $\mathcal{P}^{(m+1)}=\mathcal{P}^{(m)}\cup \lbrace\mathbf{p}^{(m+1)}\rbrace$, $m=m+1$, and the outer Problem~\eqref{optimWassersteinOuter} is solved using the updated $\mathcal{P}^{(m)}$. The algorithm stops when no new model is found. %The robustness of the solution is guaranteed, as the minimum in~\eqref{optimWassersteinOuter} is taken with respect to an increasing set $\mathcal{P}^{(m)}$.
The convergence of the algorithm is proven in Pflug and Pichler \cite{PflugPicher2014}. For completeness, a sketch of the proof of this results in our setting is presented in Appendix~\ref{appendixAlgorithm}.

\vspace{0.2cm}

For the implementation, we resume the input for~\eqref{optimWasserstein}: the DM's initial wealth is $W_0=250$ and the utility is $u(x)=(1-\exp(-\gamma x))/\gamma$, for $\gamma=0.03$, while the premium budget $\Pi_0=4$ and the safety loading is $\rho=0.2$. For $n=200$, we simulate from the distribution $F_{X,Q}$, and construct the piecewise linear approximation $F_{X,\widehat{Q}}$ of the empirical distribution  on the partition $0=x_1<x_2<\cdots<x_n=M=W_0-\Pi_0=246$. The cdf $F_{X,\widehat{Q}}$ will play the role of the baseline distribution in Problem~\eqref{optimWasserstein}. Finally, the ambiguity radius $\delta_n$ is estimated to be around $0.3$. Figure~\ref{fig:WDYP} shows one of the saddle points of Problem~\eqref{optimWasserstein}: the optimal $\mathbf{\hat{y}}^*$ is a deductible indemnity with a deductible $d^*\approx 7.6$, while the corresponding $F_{X,P^*}$ dominates $F_{X,\widehat{Q}}$ in the first stochastic order.

\begin{figure}[ht]
\center
  \subfloat{
    \includegraphics[width=.5\linewidth]{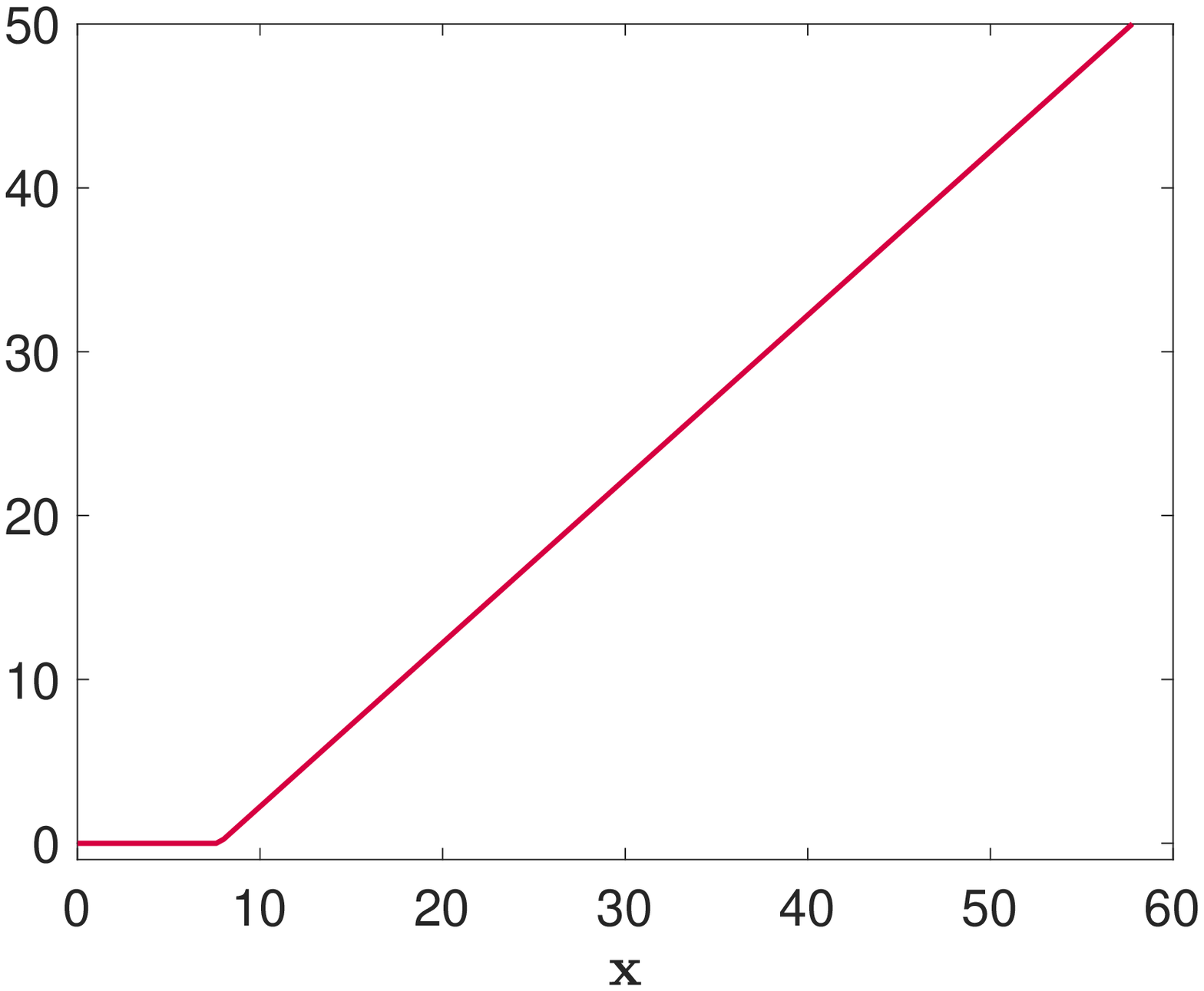}
  }
  \subfloat{
    \includegraphics[width=.5\linewidth]{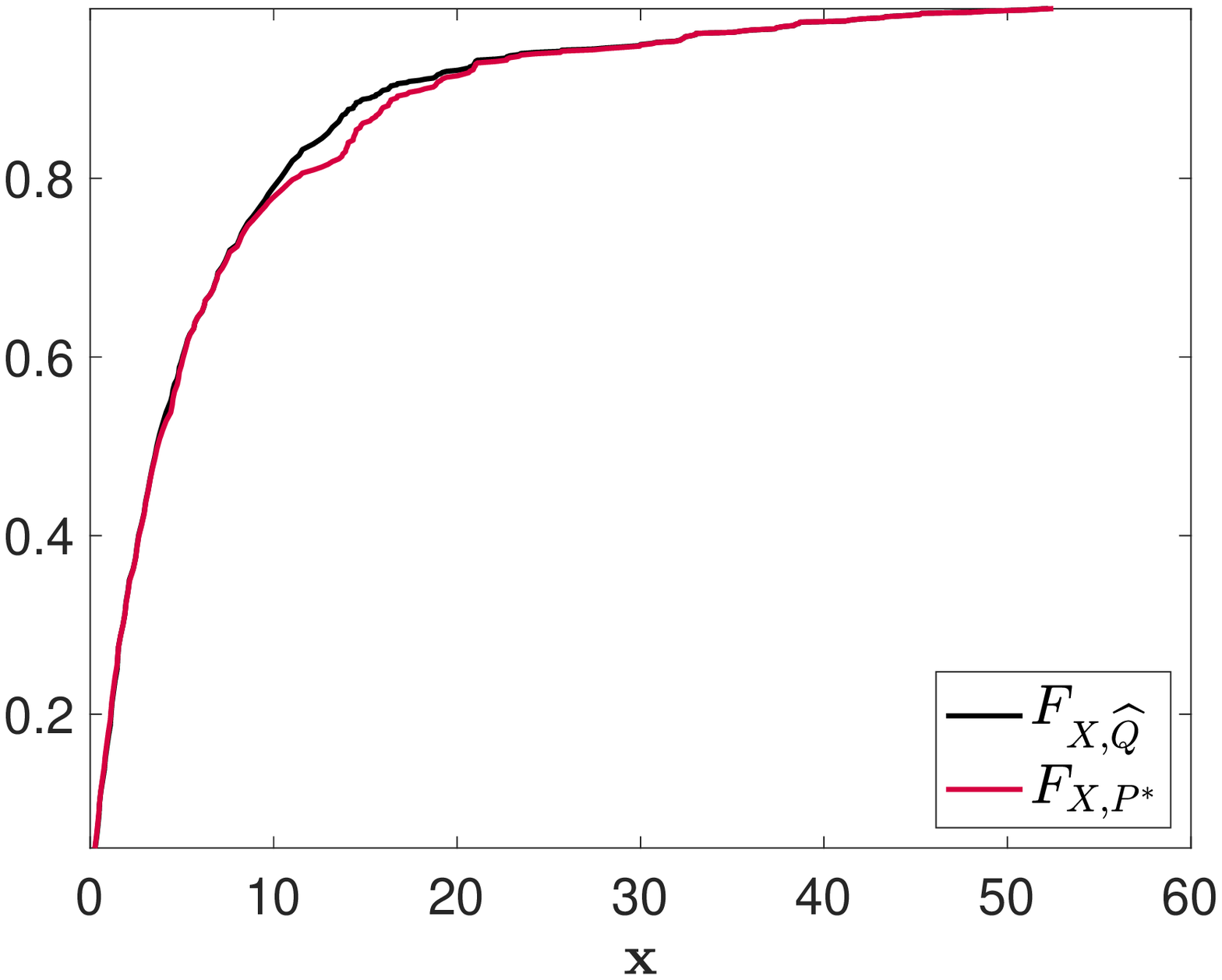}
  }
  \caption{Left: the optimal indemnity $\mathbf{\hat{y}}^*=\mathbf{x}-\mathbf{\hat{r}}^*$ as function of $\mathbf{x}$. Right: the DM's optimal distribution $F_{X,P^*}$ (red) compared to insurers' belief $F_{X,\widehat{Q}}$ (black).}\label{fig:WDYP}
\end{figure}

\vspace{0.2cm}

Next, we study a problem related to Problem~\eqref{optimWasserstein}, in which retention function $\mathbf{r}$ is required only to be bounded by $\mathbf{x}$, i.e. $\mathbf{r}\in\mathcal{I}$, where $\mathcal{F}=\mathcal{I}$ as defined in eq.~\eqref{FeasibWithoutNSC}. In Figure~\ref{fig:WDCase12} (left) we display the difference between the optimal retention functions of Problem~\eqref{optimWasserstein}, for the sets $\mathcal{I}$ and $\hat{\mathcal{I}}$, respectively. Figure~\ref{fig:WDCase12} (right) provides a closer look at the corresponding indemnities $\mathbf{y}^*$ and $\mathbf{\hat{y}}^*$, respectively. While the deductible contract is optimal when the no-sabotage condition is present (the red line in Figure~\ref{fig:WDCase12}), the indemnity $\mathbf{y}^*$ can be decreasing with respect to the loss $\mathbf{x}$ in the absence of this condition (see the blue line in Figure~\ref{fig:WDCase12}).

\begin{figure}[ht]
\center
  \subfloat{
    \includegraphics[width=.5\linewidth]{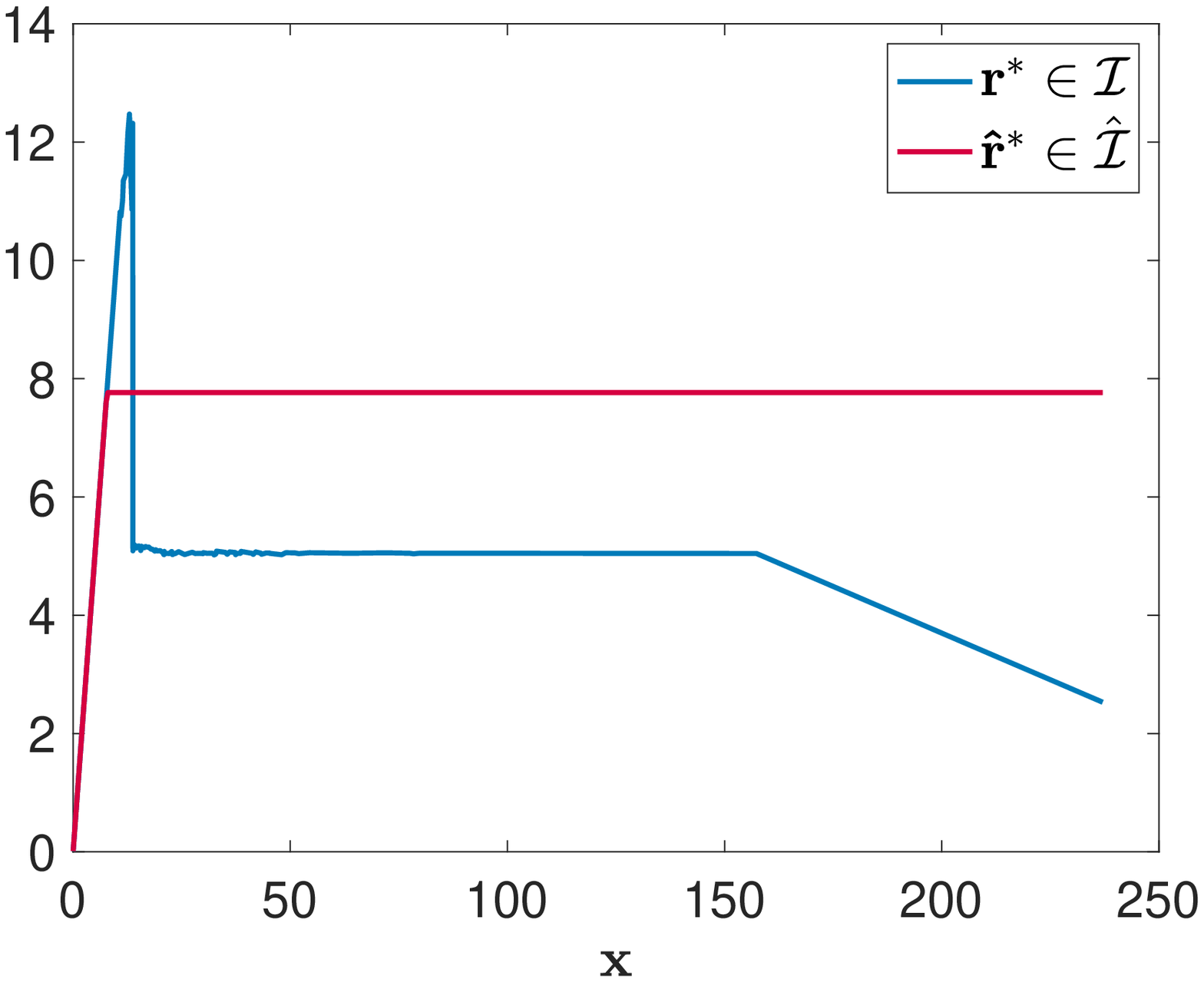}
  }
  \subfloat{
    \includegraphics[width=.5\linewidth]{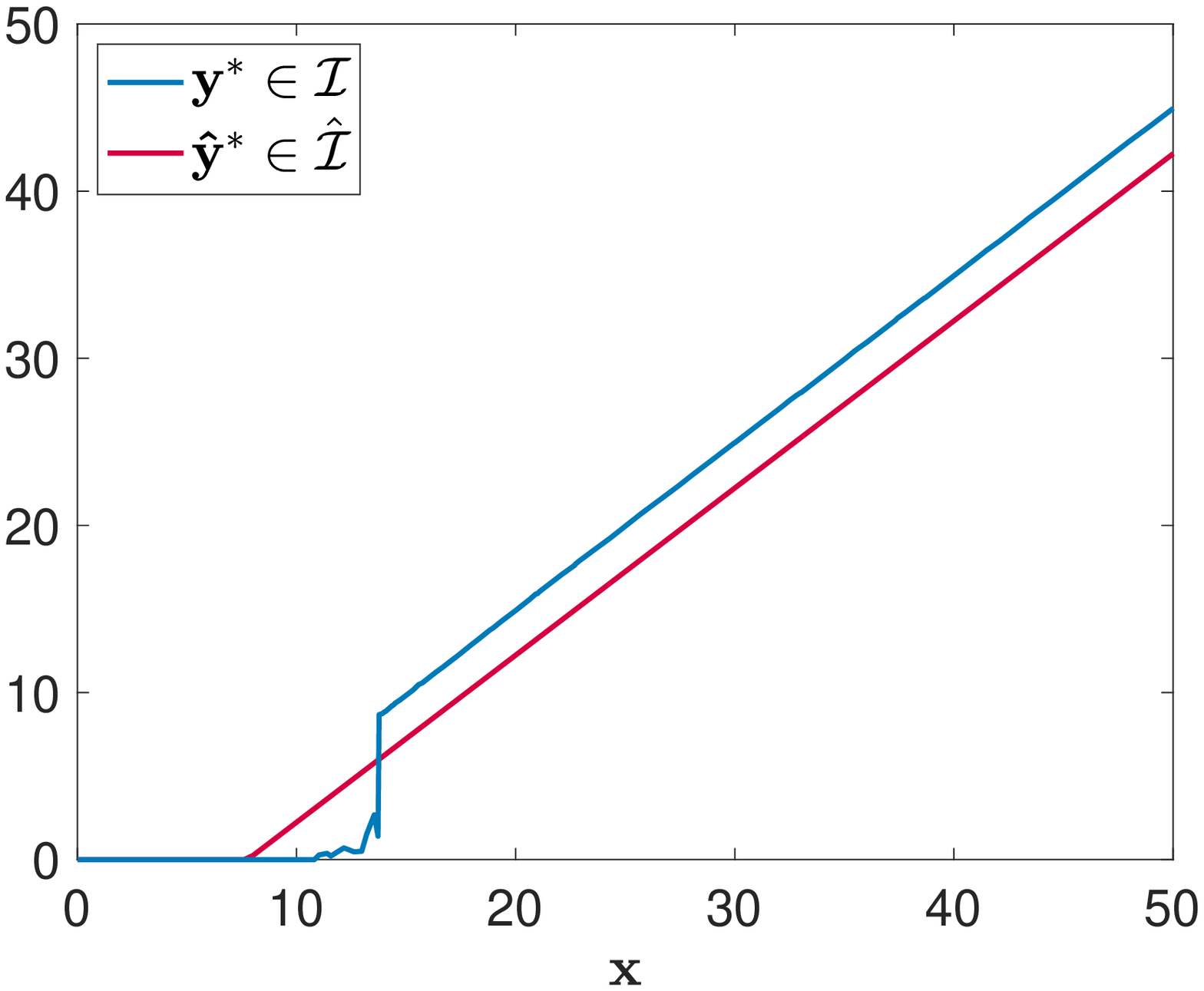}
  }
  \caption{Left: optimal retention functions $\mathbf{r}^*$ and $\mathbf{\hat{r}}^*$ for Problem~\eqref{optimWasserstein} when the feasibility sets are $\mathcal{I}$ and $\hat{\mathcal{I}}$, respectively. Right: the corresponding indemnities $\mathbf{y}^*= \mathbf{x}-\mathbf{r}^*$ and $\mathbf{\hat{y}}^*=\mathbf{x}-\mathbf{\hat{r}}^*$ with a zoomed-in perspective.}\label{fig:WDCase12}
\end{figure}
\end{example}

%======================================================================================

\vspace{0.4cm}

\begin{example}[R{\'e}nyi ambiguity set]\label{exRenyi}
For this example we focus on Problem~\eqref{optimV}, when the admissible set of indemnities is $\mathcal{F}=\hat{\mathcal{I}}$ as defined in eq.~\eqref{FeasibWithNSC}. Let DM's ambiguity set $\mathcal{C}_\delta$ be given by
\[
\mathcal{C}_{\delta}^{D_\alpha}:=\lbrace P\ll Q :\, D_\alpha(P\Vert Q)\leq\delta\rbrace,
\]
where $D_\alpha$ is the R{\'e}nyi divergence of order $\alpha$ between $P$ and $Q$, i.e.,
\begin{equation*}
D_\alpha(P\Vert Q):=\dfrac{1}{\alpha-1}\log\mathbb{E}_Q\bigg[\bigg(\dfrac{dP}{dQ}\bigg)^\alpha\bigg].
\end{equation*}
We observe that for every $\alpha\geq 1$, $D_\alpha(P\Vert Q)=0$ if and only if $P=Q$. When $\alpha\rightarrow 1$, $D_\alpha$ is the well-known Kullback-Leibler divergence. Moreover, since $S$ is a Polish space, for any ambiguity radius $\delta\in [0,\infty)$ and degree $\alpha\geq 1$ the set $\mathcal{C}_{\delta}^{D_\alpha}$ is a convex  and compact in the topology of weak convergence (e.g., \cite[Theorem 20]{vErvenHarremos2014}). For more on the properties of the divergence $D_\alpha$, we refer to \cite{Renyi1961} and \cite{LieseVajda1987}.

\vspace{0.2cm}

To illustrate our results, we follow the existing literature and consider a discretely distributed loss $X$. For a sample of size $n$, we assume without loss of generality that $x_1\leq\cdots\leq x_n$, and we denote this $n$-sample by $\mathbf{x}=[x_1,\ldots,x_n]^\top$. For our example, a random sample $\mathbf{x}$ of size $n=100$ is drawn from a truncated exponential distribution with mean $\mu=20$, and with an upper bound $M=W_0-\Pi_0$. Moreover, the insurer's belief $\widehat{Q}$ is the empirical distribution of the sample $\mathbf{x}$. Let $\widehat{\mathbf{q}}=[\widehat{q}_1,\ldots,\widehat{q}_n]^\top$ be the insurer's probability mass function (pmf), where $\widehat{q}_i:=\widehat{Q}(X=x_i)$, $\widehat{q}_i\geq 0$, $i=1,\ldots, n$, $\mathbf{1}^\top\widehat{\mathbf{q}}=1$.

\vspace{0.2cm}

Let $\mathbf{\hat{y}}=[\hat{y}_1,\ldots,\hat{y}_n]^\top\in\mathbb{R}_+^n$ be the indemnification function corresponding to the loss $\mathbf{x}$. Following the approach in~\cite{Asimitetal2017}, the feasibility constraints $0\leq \hat{y}_i\leq x_i$ and $0\leq \hat{y}_i-\hat{y}_{i-1}\leq x_i-x_{i-1}$, for $i=1, \ldots, n$, are represented by $\mathbf{0}\leq \mathbf{\hat{y}}\leq\mathbf{x}$ and $\mathbf{0}\leq \mathbf{A}_{n-1}\mathbf{\hat{y}}\leq \mathbf{A}_{n-1}\mathbf{x}$, where $\mathbf{A}_{n-1}\in\mathbb{R}^{(n-1)\times n}$ is defined in eq.~\eqref{defMatrixAi}.

\vspace{0.2cm}
Moreover, $\mathbf{p}=[p_1,\ldots,p_n]^\top\in [0,1]^n$ belongs to $\mathcal{C}_{\delta}^{D_\alpha}$ if it satisfies the following conditions:
\begin{enumerate}[label=(\roman*)]
  \item is a pmf: $\mathbf{1}^\top\mathbf{p}=1$;
  \item is absolutely continuous with respect to $\widehat{\mathbf{q}}$: if $\exists\, i \in \{1,\ldots, n\}$ such that $\widehat{q}_i=0$, then $p_i=0$.
  \item lies in a R{\'e}nyi ambiguity set around $\widehat{\mathbf{q}}$:
  $$ \mathbf{p}^\alpha\cdot\widehat{\mathbf{q}}^{\, 1-\alpha}= \sum_{i=1}^n p_i^{\alpha}\,\widehat{q}_i^{\, 1-\alpha}\leq\overline{\delta}, $$
  where $\overline{\delta}:=\exp(\delta(\alpha-1))$.
\end{enumerate}

To simplify the notation, let $D:=\lbrace \mathbf{p}\in [0,1]^n : \,  p_i=0  \text{ if } \widehat{q}_i=0, \, i=1,\ldots,n\rbrace$.

\vspace{0.2cm}

With these representations for the variables $\mathbf{\hat{y}}$ and $\mathbf{p}$, Problem~\eqref{optimV} is given by:
\begin{equation}\label{optimRenyi}
\left\{
\begin{aligned}[3]
    & \underset{\ \mathbf{\hat{y}}\in\mathbb{R}^n_+ \quad \mathbf{p}\in D{}}{\max\quad \min}   &  & \displaystyle{\sum_{i=1}^n} u(W_0-x_i+\hat{y}_i-\Pi_0)p_i  {} \\
    & \makebox(60,0){s.t.}		 	&       & \mathbf{0}\leq \mathbf{A}_{n-1}\mathbf{\hat{y}}\leq  \mathbf{A}_{n-1}\mathbf{x} {}, \\
    &                                             &       & \mathbf{0}\leq \mathbf{\hat{y}} \leq \mathbf{x} {},\\
    &                                             &       & \sum_{i=1}^n -v(W_0^{\text{Ins}}-(1+\theta)\hat{y}_i+\Pi_0)\widehat{q}_i \leq -v(W_0^{\text{Ins}}) {}, \\
    &												&       & \mathbf{p}^\alpha\cdot\widehat{\mathbf{q}}^{\, 1-\alpha}  \leq  \overline{\delta} {}, \\
    &                                             &       & \mathbf{1}^\top\mathbf{p}  =  1 {}.
\end{aligned}\tag{P$^{n}$}
\right.
\end{equation}

\vspace{0.2cm}

Observe that Problem~\eqref{optimRenyi} is a convex optimization problem, as the objective function is concave in $\hat{y}_i$ and linear in $p_i$, for  $i=1,\ldots, n$, while the constraints are convex in $\hat{y}_i$ and $p_i$, for any $\alpha>1$. Similarly to Example~\ref{exWD},  Problem~\eqref{optimRenyi} is solved in a step-wise manner, by splitting the initial problem into an inner Problem~\eqref{optimRenyiInner}:
\begin{equation}\label{optimRenyiInner}
\left\{
\begin{aligned}[2]
    & \underset{\mathbf{p}\in D{}}{\min}   &  & \displaystyle{\sum_{i=1}^n} u(W_0-x_i+\hat{y}_i-\Pi_0)p_i  {} \\
    & \makebox(20,0){s.t.}		   &       & \mathbf{p}^\alpha\cdot\widehat{\mathbf{q}}^{\, 1-\alpha}  \leq  \overline{\delta} {}, \\
    &                                             &       & \mathbf{1}^\top\mathbf{p}  =  1 {},
\end{aligned}\tag{P$^{n}_{\text{inner}}$}
\right.
\end{equation}

\noindent and outer Problem~\eqref{optimRenyiOuter}:
\begin{equation}\label{optimRenyiOuter}
\left\{
\begin{aligned}[2]
    & \underset{\ \mathbf{\hat{y}}\in\mathbb{R}^n_+ \quad \mathbf{p}\in \mathcal{P}^{(m)} {}}{\max\quad \min}  &  & \displaystyle{\sum_{i=1}^n} u(W_0-x_i+\hat{y}_i-\Pi_0)p_i  {} \\
    & \makebox(60,0){s.t.}		 	&       & \mathbf{0}\leq \mathbf{A}_{n-1}\mathbf{\hat{y}}\leq  \mathbf{A}_{n-1}\mathbf{x} {}, \\
    &                                             &       & \mathbf{0}\leq \mathbf{\hat{y}} \leq \mathbf{x} {},\\
    &                                             &       & \sum_{i=1}^n -v(W_0^{\text{Ins}}-(1+\theta)\hat{y}_i+\Pi_0)\widehat{q}_i \leq -v(W_0^{\text{Ins}}) {},
\end{aligned}\tag{P$^{n}_{\text{outer}}$}
\right.
\end{equation}
where $\mathcal{P}^{(m)}=\lbrace \widehat{\mathbf{q}},\mathbf{p}^{(1)},\ldots,\mathbf{p}^{(m)}\rbrace$ is a finite set of alternative pmfs, obtained as worst-case solutions in previous iterations of Problem~\eqref{optimRenyiInner}. More precisely, the algorithm starts with the initialization $\mathcal{P}^{(0)}=\lbrace \widehat{\mathbf{q}}\rbrace$. The outer Problem~\eqref{optimRenyiOuter} receives as input the singleton $\mathcal{P}^{(0)}$ and yields an optimal solution $\mathbf{\hat{y}}^{(1)}$. The indemnity $\mathbf{\hat{y}}^{(1)}$ is then used as input for~\eqref{optimRenyiInner}, which gives a new pmf $\mathbf{p}^{(1)}$. The new $\mathbf{p}^{(1)}$ is added to a finite set of feasible models $\mathcal{P}^{(1)}=\mathcal{P}^{(0)}\cup \lbrace \mathbf{p}^{(1)}\rbrace$, which is the input for~\eqref{optimRenyiOuter}. The algorithm continues until no new model is added to $\mathcal{P}^{(m)}$. A sketch of the proof of convergence is provided in Appendix~\ref{appendixAlgorithm}.

\vspace{0.2cm}

To obtain an explicit solution, let the DM's initial wealth be $W_0=200$, the budget be $\Pi_0=10$, the safety loading be $\rho=0.2$, and the DM's utility be $u(x)=x^{1/3}$. The insurers' initial wealth is $W_0^{\text{Ins}}=400$, and the utility is $v(x)=x^{1/5}$. For the ambiguity set $\mathcal{C}_{\delta}^{D_\alpha}$, we choose the ambiguity radius $\delta=0.7$ and the order of R{\'e}nyi divergence $\alpha=3$. Figure~\ref{fig:RenyiYP} shows the optimal indemnity $\mathbf{\hat{y}}^*$ (left) and the worst-case distribution $F_{X,P^*}$, corresponding to $\mathbf{p}^*$ (right).

\begin{figure}[ht]
\center
  \subfloat{
    \includegraphics[width=.5\linewidth]{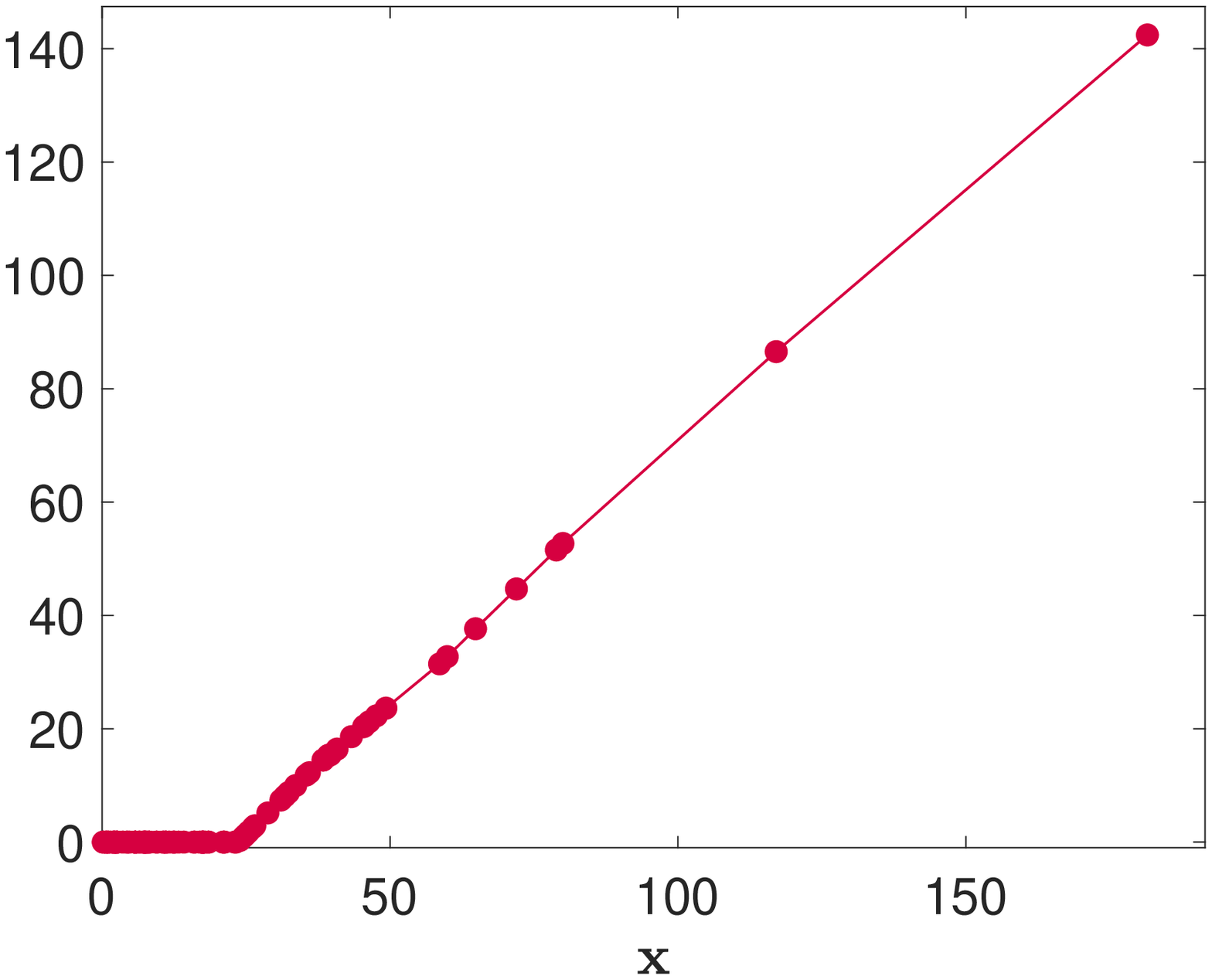}
  }
  \subfloat{
    \includegraphics[width=.5\linewidth]{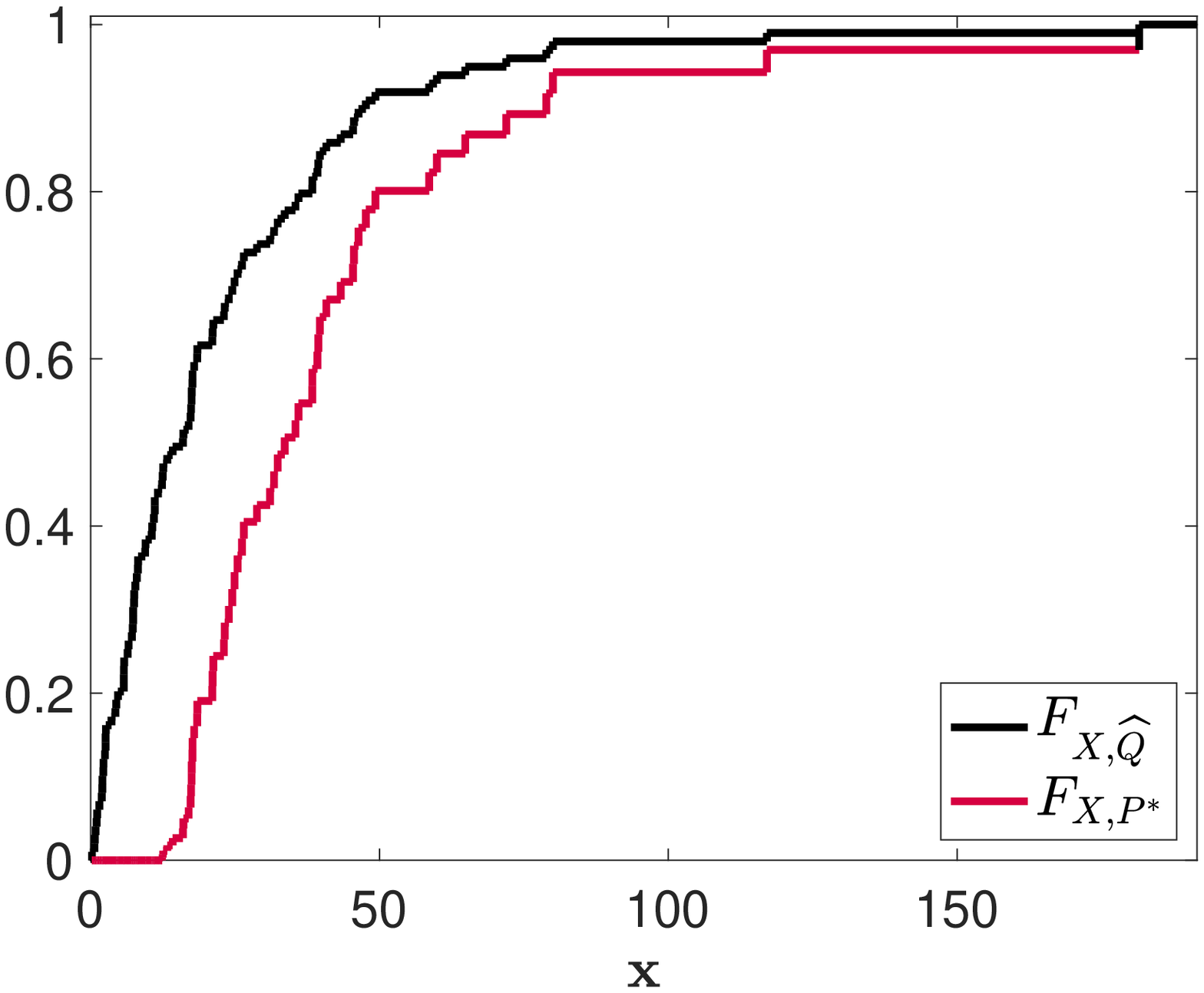}
  }
  \caption{Left: the optimal indemnity $\mathbf{\hat{y}}^*$ as function of $\mathbf{x}$. Right: the DM's optimal distribution $F_{X,P^*}$ (red) compared to insurers' belief $F_{X,\widehat{Q}}$ (black).}\label{fig:RenyiYP}
\end{figure}

\vspace{0.2cm}

We also solve Problem~\eqref{optimRenyi} for the same ambiguity set $\mathcal{C}_{\delta}^{D_\alpha}$, when the feasibility set is $\mathcal{F}=\mathcal{I}$ as defined in eq.~\eqref{FeasibWithoutNSC}. This implies that the constraint $\mathbf{0}\leq \mathbf{A}_{n-1}\mathbf{\hat{y}}\leq \mathbf{A}_{n-1}\mathbf{x}$ is removed from the optimization Problem~\eqref{optimRenyi}. Figure~\ref{fig:RenyiCase12} illustrates the difference between the optimal indemnities corresponding to $\mathcal{I}$ and $\hat{\mathcal{I}}$.

\begin{figure}[h]
\center
\includegraphics[scale=0.5]{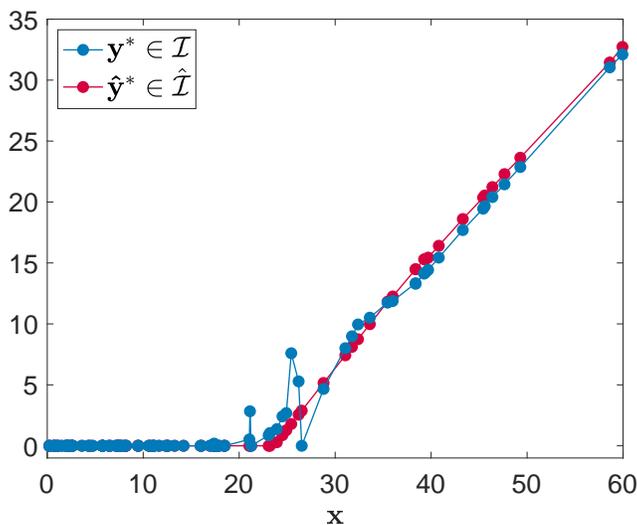}
\caption{Optimal indemnities $\mathbf{y}^*$ and $\mathbf{\hat{y}}^*$ for Problem~\eqref{optimRenyi} when the feasibility sets are $\mathcal{I}$ and $\hat{\mathcal{I}}$, respectively.}\label{fig:RenyiCase12}
\end{figure}

\vspace{0.2cm}

We next investigate the decrease in optimal expected utility, when the ambiguity set increases. Certainty equivalence is used to quantify the impact of ambiguity radii on the optimal value of Problem~\eqref{optimRenyi}.

\vspace{0.2cm}

For each $\delta$, let $(\mathbf{\hat{y}}^*,\mathbf{p}^*)$ be an optimal solution of~\eqref{optimRenyi} and define the certainty equivalents $\ce1$ and $\CE2$ as follows:
\begin{equation*}
\left\{\begin{aligned}
& \inf_{P\in\mathcal{C}_{\delta}^{D_\alpha}}\mathbb{E}_P[u(W_0-\mathbf{x}+\ce1(\delta))]= \sup_{\mathbf{\hat{y}}\in\hat{\mathcal{I}}}\,\inf_{P\in\mathcal{C}_{\delta}^{D_\alpha}}\mathbb{E}_P[u(W_0-\mathbf{x}+\mathbf{\hat{y}}-\Pi_0)],\\
& u(\CE2(\delta)) = \sup_{\mathbf{\hat{y}}\in\hat{\mathcal{I}}}\,\inf_{P\in\mathcal{C}_{\delta}^{D_\alpha}}\mathbb{E}_P[u(W_0-\mathbf{x}+\mathbf{\hat{y}}-\Pi_0)],
\end{aligned}\right.
\end{equation*}
where $P$ is the probability measure corresponding to $\mathbf{p}$.  The constant $\ce1$ quantifies the marginal benefit of the optimal insurance contract, which we interpret as the willingness-to-pay for insurance. Moreover, $\CE2$ measures the certainty equivalent of DM's final position. Figure~\ref{fig:RenyiCE} displays the changes in certainty equivalents for increased values of ambiguity radius. The left figure shows that a larger ambiguity radius yields a higher marginal benefit of the optimal insurance contract. This implies that the DM has a higher willingness-to-pay for the optimal insurance contract if the ambiguity set gets larger. On the other hand, the certainty equivalent of the final wealth position decreases when the ambiguity set gets larger because the DM is more ambiguity-averse (Figure~\ref{fig:RenyiCE} (right)).

\begin{figure}[ht]
\center
  \subfloat{
    \includegraphics[width=.5\linewidth]{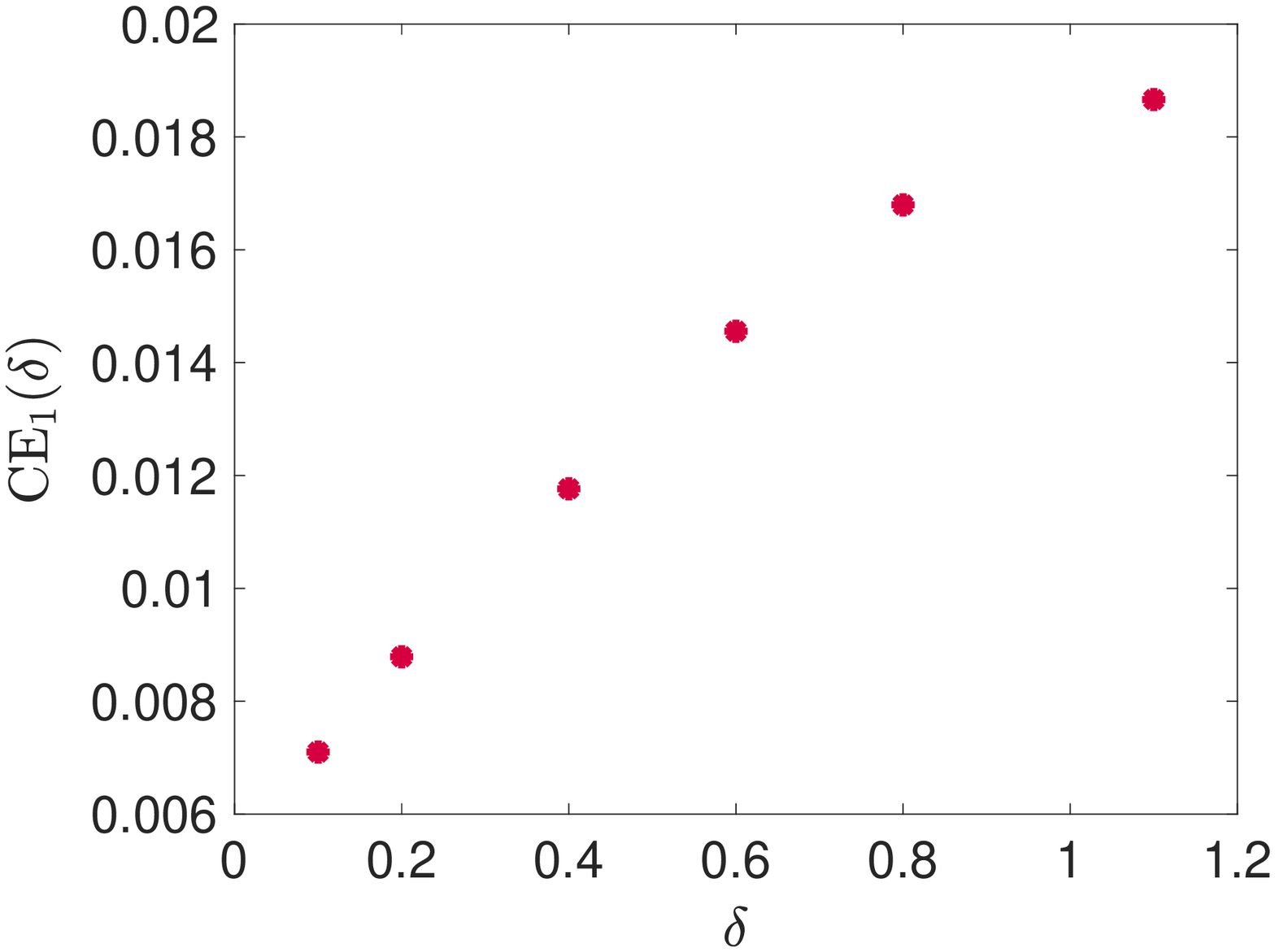}
  }
  \subfloat{
    \includegraphics[width=.5\linewidth]{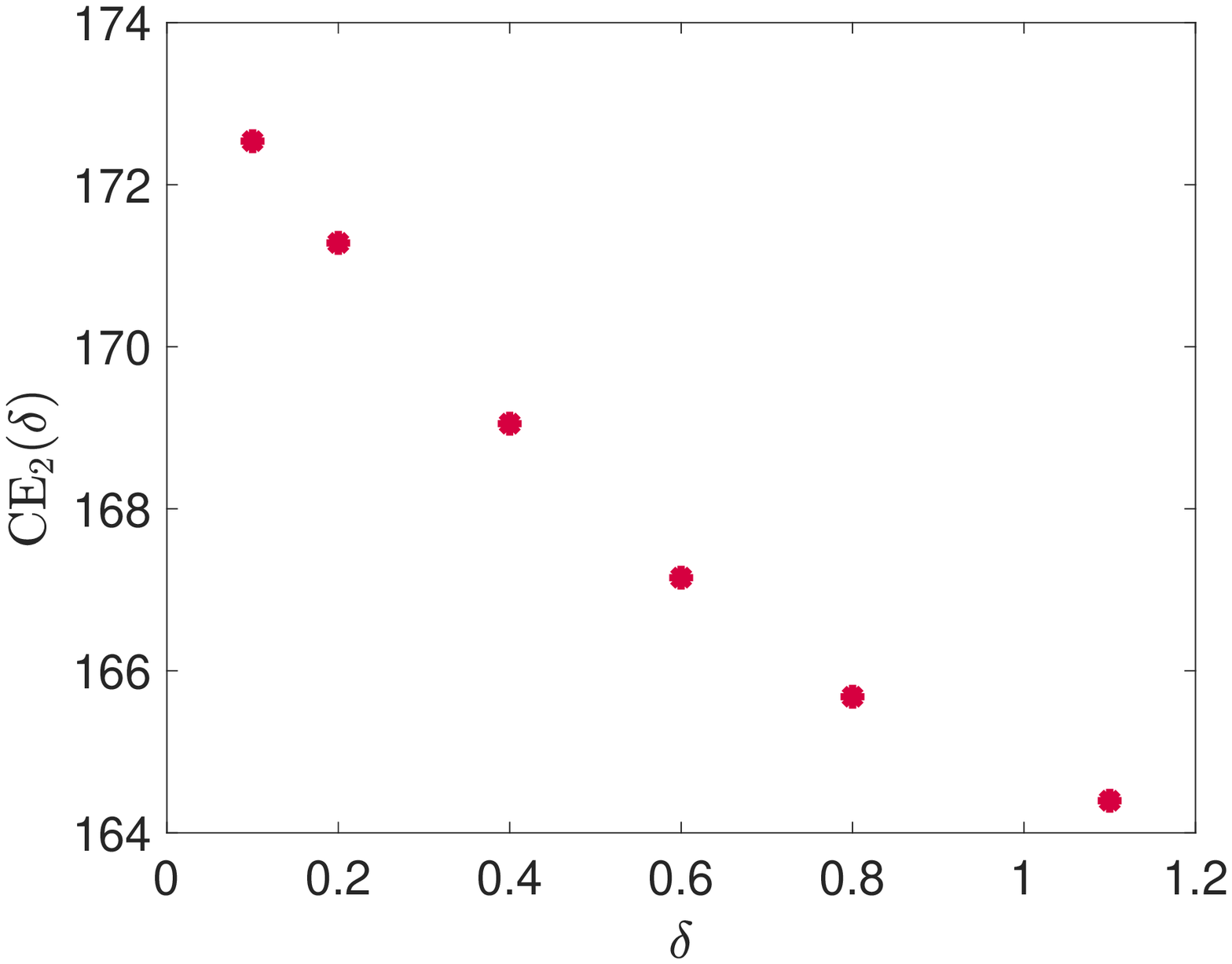}
  }
  \caption{Left: certainty equivalent $\ce1$ as function of the ambiguity radius $\delta$. Right: certainty equivalent $\CE2$ as function of the ambiguity radius $\delta$.}\label{fig:RenyiCE}
\end{figure}

\vspace{0.2cm}

Figure~\ref{fig:RenyiDelta} (left) provides a closer look at the optimal indemnities $\mathbf{\hat{y}}^*$, when the ambiguity set $\mathcal{C}_{\delta}^{D_\alpha}$ becomes wider. Figure~\ref{fig:RenyiDelta} (right) shows the worst-case distribution $F_{X,P^*}$ for several values of $\delta$. For increasing values of the ambiguity radius, it can be observed that each $F_{X,P^*}$ dominates all the previous distributions in the first stochastic order.

\begin{figure}[h]
\center
  \subfloat{
    \includegraphics[width=.5\linewidth]{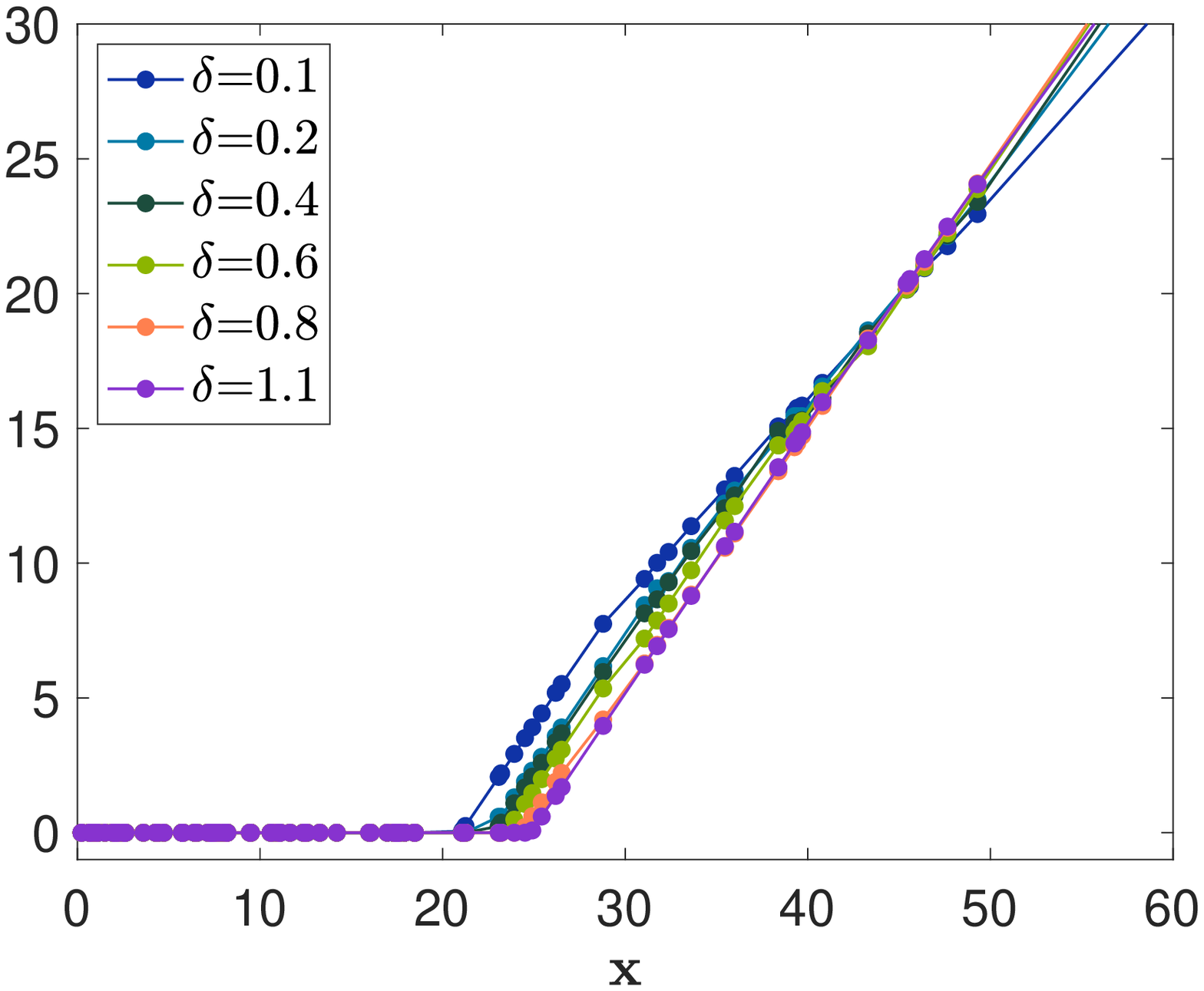}
  }
  \subfloat{
    \includegraphics[width=.5\linewidth]{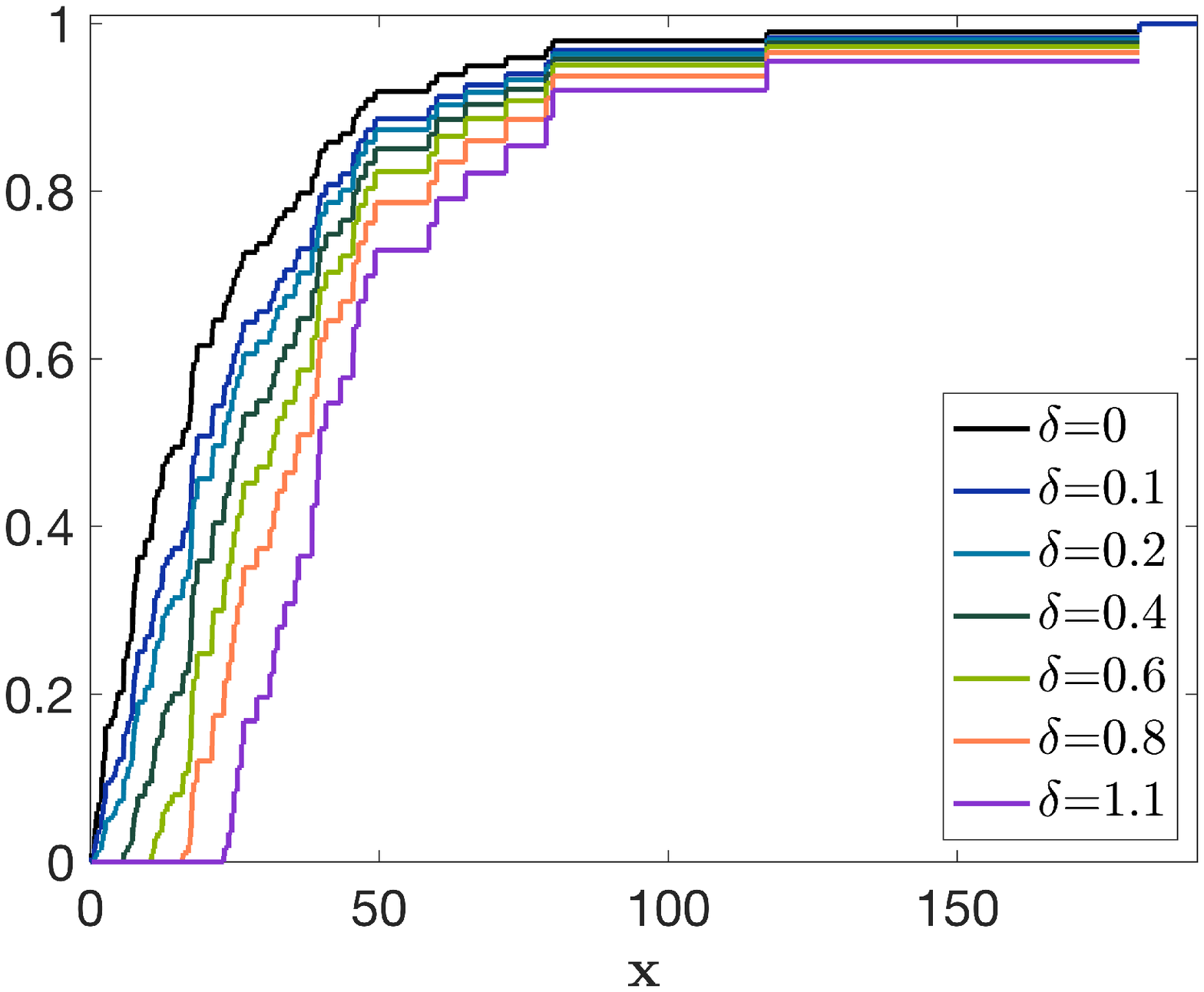}
  }
  \caption{Left: optimal indemnities $\mathbf{\hat{y}}^*\in \mathcal{\hat{I}}$ of Problem~\eqref{optimRenyi} for several values of $\delta$. Right: the corresponding worst-case probability distributions $F_{X,P^*}$ for several values of $\delta$.\vspace{0.5cm}}\label{fig:RenyiDelta}
\end{figure}
\end{example}

%\vspace{0.4cm}

\newpage
%====================================================================================
%====================================================================================
%====================================================================================

\section{Conclusion}\label{section6}
The impact of ambiguity on insurance markets in general, and insurance contracting in particular, is by now well-documented. One of the most popular and intuitive ways to model sensitivity of preferences to ambiguity is the Maxmin-Expected Utility (MEU) model of Gilboa and Schmeidler \cite{Gilboa-Schmeidler1989}. Nonetheless, to the best of our knowledge, none of the theoretical studies of risk sharing in insurance markets in the presence of ambiguity have examined the case in which the insured is a MEU-maximizer. This paper fills this void. Specifically, we extend the classical setup and results in two ways: (i) the decision maker (DM) is endowed with MEU preferences; and (ii) the insurer is an Expected-Utility-maximizer who is not necessarily risk-neutral (that is, the premium principle is not necessarily an expected-value premium principle). The main objective of this paper is then to determine the shape of the optimal insurance indemnity in that case.

\vspace{0.4cm}

We characterize optimal indemnity functions both with and without the customary \textit{ex ante no-sabotage} requirement on feasible indemnities, and for both concave and linear utility functions for the two agents. The no-sabotage condition is shown to play a key role in determining the shape of optimal indemnity functions. An equally important factor in characterizing optimal indemnities is the singularity in beliefs between the two agents.

\vspace{0.4cm}

We subsequently examine several illustrative examples, and we provide numerical studies for the case of a Wasserstein and a R{\'e}nyi ambiguity set. Specifically, we provide a successive convex programming algorithm to compute optimal insurance indemnities in a discretized framework. The Wasserstein and R{\'e}nyi distances are two popular metrics to construct probability ambiguity sets. We show in numerical examples that a larger ambiguity set yields a lower certainty equivalent of final wealth, but increases the willingness-to-pay for insurance.

%\vspace{2cm}
\newpage

%====================================================================================
%====================================================================================
%====================================================================================
 %\newpage
\setlength{\parskip}{0.5ex}
\hypertarget{LinkToAppendix}{\ }
\appendix

\vspace{-1.2cm}

%====================================================================================
%====================================================================================
%====================================================================================

\section{Proof of Proposition~\ref{propCase2}}\label{appendixProofpropCase2}

For $\hat{R}:=X-\hat{Y}$ the retention random variable, Problem~\eqref{optim} becomes
\begin{equation}\label{optimc2}
\begin{aligned}
& \underset{\hat{R}\in\hat{\mathcal{I}}}{\sup}\,\underset{P\in\mathcal{C}}{\inf}
& & \mathbb{E}_P [u(W_0-\hat{R}-\Pi_0)]\\
& \text{s.t.}
& & \mathbb{E}_Q [\hat{R}]\geq \tilde{\Pi}_0,
\end{aligned}
\end{equation}
where $\tilde{\Pi}_0=\mathbb{E}_Q[X]-(1+\rho)^{-1}\Pi_0$. The Sion's Minimax Theorem yields that there exists a saddle point $(\hat{R}^*,P^*)\in\hat{\mathcal{I}}_0\times\mathcal{C}$ such that $\sup_{\hat{R}\in\hat{\mathcal{I}}_0}\,\inf_{P\in\mathcal{C}}\mathbb{E}_P [u(W_0-\hat{R}-\Pi_0)] = \min_{P\in\mathcal{C}}\max_{\hat{R}\in\hat{\mathcal{I}}_0}\mathbb{E}_{P^*} [u(W_0-\hat{R}-\Pi_0)]$. For $P^*\in\mathcal{C}$, the inner optimization problem becomes:
\begin{equation*}
\begin{aligned}
& \underset{\hat{R}\in\hat{\mathcal{I}}}{\max}
& & \int_A u(W_0-\hat{R}-\Pi_0)h^*\,dQ+ \int_{S\setminus A} u(W_0-\hat{R}-\Pi_0)\,dP^*_s \\
& \text{s.t.}
& & \int_A \hat{R}\,dQ(s)\geq \tilde{\Pi}_0,
\end{aligned}
\end{equation*}
where $A:=A_{P^*}$ and $h^*:=h_{P^*}$ are as in Remark~\ref{remDecomposition}, and $h^*=\xi^*(X)$, for some Borel measurable function $\xi^*:\mathbb{R}_+\rightarrow\mathbb{R}_+$. The splitting technique in \cite[Lemma C.6]{Ghossoub2019c} states that the optimal retention function $\hat{R}^*$ can be obtained as $\hat{R}^*=\hat{R}_1^*\mathrm{1}_A$, where $\hat{R}_1^*$ solves Problem~\eqref{optim1c2} below.

\begin{equation}\label{optim1c2}
\sup_{\hat{R}_1\in\hat{\mathcal{I}}} \bigg\lbrace \int u(W_0-\hat{R}_1-\Pi_0)h^*\,dQ\, : \, \int \hat{R}_1 dQ \geq \tilde{\Pi}_0\bigg\rbrace.
\end{equation}

Therefore, the optimal solution $\hat{R}^*=\hat{R}_1^*\mathrm{1}_A$ is 1-Lipschitz $Q$-a.s., as long as $\hat{R}_1^*$ is 1-Lipschitz on $A$. Similar to Theorem~\ref{propCase2V}, a necessary and sufficient condition for $\hat{R}_1^*$ to be the optimal solution of Problem~\eqref{optim1c2} is
\begin{equation}\label{eqRoptimc2}
\begin{aligned}
& \int_0^M \big[u'(W_0-\hat{r}_1^*(t)-\Pi_0)\xi^*(t)-\lambda\big](\hat{r}_1^*(t)-\hat{r}_1(t))dQ\circ X^{-1}(t)\leq 0,\, \forall \hat{R}_1=\hat{r}_1(X)\in \hat{\mathcal{I}},
\end{aligned}
\end{equation}
where $\lambda\in\mathbb{R}_+$ is the Lagrange multiplier associated to Problem~\eqref{optim1c2}. As any $\hat{R}_1=\hat{r}_1(X)\in\hat{\mathcal{I}}$ is absolutely continuous, it is almost everywhere differentiable on $[0,M]$, and hence~\eqref{eqRoptimc2} becomes
\begin{equation*}
\begin{split}
0 & \geq \int_0^M\int_0^t ((\hat{r}_1^*)'(x)-\hat{r}_1'(x))(u'(W_0-\hat{r}_1^*(t)-\Pi_0)\xi^*(t)-\lambda)\,dx\,dQ\circ X^{-1}(t) \\
& = \int_0^M\int_x^M (u'(W_0-\hat{r}_1^*(t)-\Pi_0)\xi^*(t)-\lambda)\,dQ\circ X^{-1}(t)((\hat{r}_1^*)'(x)-\hat{r}_1'(x))\,dx,
\end{split}
\end{equation*}
for all $\hat{R}_1=\hat{r}_1(X) \in\hat{\mathcal{I}}$; hence $\hat{R}^*_1=\hat{r}_1^*(X)$ is of the form
\begin{equation}\label{eqOptimMIF}
(\hat{r}^*_1)'(x)=\begin{cases}
0, &\mbox{if } \int_{[x,M]\cap X(A)} (u'(W_0-\hat{r}^*(t)-\Pi_0)\xi^*(t)-\lambda^*)\,dQ\circ X^{-1}(t)>0,\\
\kappa(t),  &\mbox{if }  \int_{[x,M]\cap X(A)} (u'(W_0-\hat{r}^*(t)-\Pi_0)\xi^*(t)-\lambda^*)\,dQ\circ X^{-1}(t)=0,\\
1,  &\mbox{if } \int_{[x,M]\cap X(A)} (u'(W_0-\hat{r}^*(t)-\Pi_0)\xi^*(t)-\lambda^*)\,dQ\circ X^{-1}(t)<0,
\end{cases}
\end{equation}
for some Lebesgue measurable  and $[0,1]$-valued function $\kappa$. The existence of $\lambda^*\in\mathbb{R}_+$ such that $\mathbb{E}_Q[\hat{R}^*_1]=\tilde{\Pi}_0$ follows similar to Theorem~\ref{propCase2V}.\hfill {\small $\boxvoid$}

\vspace{0.4cm}

%====================================================================================
%====================================================================================
%====================================================================================

\section{Convexity and Compactness of $\mathcal{C}_{\mathcal{W}}$ in Examples~\ref{exGeneralU} and \ref{exLinearUalternative1}}\label{appendixCompactConvexCw}

\begin{lemma}\label{lemmaSetC}
For a fixed $Q \in ca^+_1(\Sigma)$, let $\mathcal{C}_{\mathcal{W}}$ be the set defined as follows:
\begin{equation}\label{SetC2}
\mathcal{C}_{\mathcal{W}}:=\bigg\lbrace P\in ca^{+}_{1}(\Sigma):\, \dfrac{dP}{dQ}=\dfrac{w(X)}{\int w(X)dQ},\, w\in \mathcal{W}\bigg\rbrace,
\end{equation}

\vspace{0.2cm}

\noindent where $\mathcal{W} \subset L^1\left(\mathbb{R}, \mathcal{B}(\mathbb{R}), Q \circ X^{-1}\right)$ is a collection of nonnegative increasing weight functions, such that $\int w(X)dQ > 0$, for all $w \in \mathcal{W}$. Then the following hold:
\vspace{0.2cm}
\begin{enumerate}
\item[(i)] If $\mathcal{W}$ is a convex cone, then $\mathcal{C}_{\mathcal{W}}$ is convex.
\vspace{0.2cm}
\item[(ii)] If $\mathcal{C}_{\mathcal{W}}$ is uniformly absolutely continuous with respect to some $\mu \in ca^+(\Sigma)$, then $\mathcal{C}_{\mathcal{W}}$ is weak$^*$-compact.
\end{enumerate}
\end{lemma}

%\vspace{0.2cm}

\begin{proof}
(i) is easy to verify. To show (ii), first note that $\mathcal{C}_{\mathcal{W}}$ is norm-bounded. Since $\mathcal{C}_{\mathcal{W}}$ is also uniformly absolutely continuous with respect to $\mu \in ca^+(\Sigma)$, it follows from \cite[Theorem IV.9.2]{Dunford} that $\mathcal{C}_{\mathcal{W}}$ is weakly sequentially compact, and hence weak$^*$-compact, by \cite[Theorem 1]{MaccheroniMarinacci2001}.
\end{proof}

\vspace{0.2cm}

\begin{remark}
In Lemma \ref{lemmaSetC}, if $\mathcal{C}_{\mathcal{W}}$ is countable, that is, is of the form
$$\bigg\lbrace P_n\in ca^{+}_{1}(\Sigma):\, n \in \mathbb{N}, \, \dfrac{dP_n}{dQ}=\dfrac{w(X)}{\int w(X)dQ},\, w\in \mathcal{W}\bigg\rbrace,$$
and if $\underset{n \rightarrow +\infty}\lim P_{n}(A)$ exists for each $A \in \Sigma$, then the requirement of uniform absolute continuity of $\mathcal{C}_{\mathcal{W}}$ is superfluous by the Vitali-Hahn-Saks Theorem (\cite[Theorem III.7.2]{Dunford}).
\end{remark}

\vspace{0.2cm}

\begin{proposition}
If $\mathcal{W}$ is order bounded in the Banach lattice $L^1\left(\mathbb{R}, \mathcal{B}(\mathbb{R}), Q \circ X^{-1}\right)$, with a constant upper bound and a nonnegative lower bound having nonzero $L^1$-norm, then $\mathcal{C}_{\mathcal{W}}$ is uniformly absolutely continuous with respect to $Q$.
\end{proposition}

\vspace{0.2cm}

\begin{proof}
Suppose that $\mathcal{W}$ is order bounded in $L^1\left(\mathbb{R}, \mathcal{B}(\mathbb{R}), Q \circ X^{-1}\right)$, with a constant upper bound and a nonnegative lower bound having nonzero $L^1$-norm. Then there exists $f \in L^1_+\left(\mathbb{R}, \mathcal{B}(\mathbb{R}), Q \circ X^{-1}\right)$ and $M \in \mathbb{R}^+$,  such that $M < +\infty$, $\|f\|_1 = \int f\,dQ \circ X^{-1} > 0$, and $f \leq w \leq M$, for each $w \in \mathcal{W}$. Consequently, for each $P \in \mathcal{C}_{\mathcal{W}}$,
$$\dfrac{dP}{dQ} \leq \frac{M}{\|f\|_1} < +\infty.$$
Hence, for each $P \in \mathcal{C}_{\mathcal{W}}$ and each $A \in \Sigma$,
$$P(A) = \int_A dP  \leq \frac{M}{\|f\|_1} \,Q(A).$$
Consequently, for each $\varepsilon > 0$, letting $\delta := \frac{\|f\|_1}{M} \, \varepsilon > 0$, it follows that for each $A \in \Sigma$ and each $P \in \mathcal{C}_{\mathcal{W}}$,
$$Q(A) < \delta \Longrightarrow P(A) < \frac{M}{\|f\|_1} \, \delta = \varepsilon.$$
Hence, $\mathcal{C}_{\mathcal{W}}$ is uniformly absolutely continuous with respect to $Q$.
\end{proof}

\vspace{0.4cm}

%====================================================================================
%====================================================================================
%====================================================================================

\section{$\Delta$-approximation of a Continuous Function}\label{appendixDeltaApprox}

Let $f:[a,b]\rightarrow\mathbb{R}$ be a function defined on a compact interval $[a,b]\subseteq\mathbb{R}$. For the definition of piecewise linear approximation of $f$, we follow the notation in \cite{Ngueveu2019}.

\vspace{0.2cm}

\begin{definition}\normalfont \
\begin{itemize}
\item A function $g:[a,b]\rightarrow\mathbb{R}$ is a \textit{piecewise linear function (pwl)} with $m\in\mathbb{N}$ line-segments if $\exists\, \mathbf{a}=[a_1,\ldots,a_m]^\top\in\mathbb{R}^m$, $\mathbf{b}=[b_1,\ldots,b_m]^\top\in\mathbb{R}^m$, $\mathbf{x}^{\min}=[x^{\min}_1,\ldots,x^{\min}_m]^\top\in [a, b]^m$ and $\forall i=1,\ldots, m$, $\exists\,x^{\max}_i\in [x^{\min}_i,b]$ such that the following holds:
\begin{align*}
g(x) & =a_i x+b_i, \quad \forall i=1,\ldots,m,\, \forall x\in [x^{\min}_i,x^{\max}_i]; \\
x^{\max}_i & =x^{\min}_{i+1}, \quad \forall i=1,\ldots,m-1; \\
x^{\min}_1 & =a; \\
x^{\max}_m & =b.
\end{align*}
\item A pwl function $g:[a,b]\rightarrow\mathbb{R}$ is called \textit{$\Delta$-approximation} of $f$ with $\Delta>0$ if:
\begin{equation}\label{defApprox}
\max_{x\in [a,b]} \vert g(x)-f(x)\vert \leq \Delta.
\end{equation}
\item We call function $g$ a \textit{$\Delta$-underestimator} of function $f$ if condition~\eqref{defApprox} is satisfied along with
\[
g(x)\leq f(x), \forall x\in [a,b].
\]
We call function $g$ a \textit{$\Delta$-overestimator} of $f$ if $-g$ is a $\Delta$-underestimator of $-f$.
\end{itemize}
\end{definition}

\vspace{0.2cm}

To simplify the notation, we characterize a pwl function associated to $f$ via the system $G(f):=\bigcup _{i=1}^m([a_i,b_i], [x^{\min}_i,x^{\max}_i])$, where $x^{\min}_i$ and $x^{\max}_i$ are the breakpoint of the $i$-th segment.

 \vspace{0.2cm}

According to \cite{DuistermaatKolk2004}, if the function $f$ defined on the compact set $[a,b]$ is continuous, then for any scalar $\Delta>0$,  there exists a continuous $\Delta$-approximation. Moreover, \cite[Corollary 2.1]{RebennackKallrath2015} states that if a $\Delta$-approximation $g$ of $f$ exists, then for $\epsilon=2\Delta$, the functions $g_-(x):=g(x)-\Delta$ and $g_+(x):=g(x)+\Delta$ define an $\epsilon$-underestimator and an $\epsilon$-overestimator, respectively, of $f$.

%\vspace{0.4cm}
\newpage
%====================================================================================
%====================================================================================
%====================================================================================

\section{Convergence of the Algorithm in Examples~\ref{exWD} and~\ref{exRenyi}}\label{appendixAlgorithm}

The convergence of the SCP algorithm in Examples~\ref{exWD} and~\ref{exRenyi} is proven in Pflug and Pichler \cite[Proposition B.6]{PflugPicher2014}. Below we present a sketch of the proof for the general setting of Problem~\eqref{optimV}. Recall that the sets $\hat{\mathcal{I}}_0$ in \eqref{FeasibHatI} and $\mathcal{C}_\delta$ in \eqref{defCparametric} are compact and the mapping $( \hat{Y}, P )\mapsto \mathbb{E}_P [ u( W_0 - X + \hat{Y} - \Pi_0 )=:u( \hat{Y}, P )$ is jointly continuous.

\vspace{0.2cm}

\begin{proposition}
Every cluster point $\hat{Y}^*$  of the iteration:
\begin{align}
\hat{Y}^{(m+1)} &\in  \argmax_{ \hat{Y} \in \hat{\mathcal{I}}_0 } \min_{ P^{(i)} \in \mathcal{P}^{(m)} } u( \hat{Y}, P^{(i)} ), \label{eqIteration1} \\
P^{(m+1)} &\in \argmin_{P\in\mathcal{C}_\delta} \, u( \hat{Y}^{(m+1)}, P ) \label{eqIteration2}
\end{align}
is a solution of Problem~\eqref{optimV}.
\end{proposition}

\vspace{0.2cm}

\begin{proof}
Let $u^m:= \min_{ P^{(i)} \in \mathcal{P}^{(m)} } u( \hat{Y}^{(m+1)}, P^{(i)} )$;  by Remark~\ref{remC}, it follows that
\begin{equation}\label{eqIteration}
u^m = \max_{ \hat{Y}\in \hat{\mathcal{I}}_0 } \min_{ P^{(i)} \in \mathcal{P}^{(m)} } u( \hat{Y}, P^{(i)} ) \geq \max_{ \hat{Y}\in \hat{\mathcal{I}}_0 } \min_{ P^{(i)} \in \mathcal{P}^{(m+1)} } u( \hat{Y}, P^{(i)} ) = u^{m+1}.
\end{equation}

\vspace{0.2cm}

Since $u( \hat{Y}, P )$ is bounded, the iterations~\eqref{eqIteration1} -\eqref{eqIteration2}  form a decreasing sequence $u^m$ that converges to $u^*:=\inf u^m>-\infty$. By compactness of $\hat{\mathcal{I}}_0$, the sequence $Y^{(m)}$ has at least one cluster point. Let $\hat{Y}^*$ be one of the cluster points.

\vspace{0.2cm}

We would like to show first that $u^*=\min_{P \in \mathcal{C}_\delta } u(\hat{Y}^*, P)$. Assume by contradiction that $u^*>\min_{P \in \mathcal{C}_\delta } u(\hat{Y}^*, P)$; then there exists some $\tilde{P} \in \mathcal{C}_\delta$ such that $u^*> u( \hat{Y}^*, \tilde{P} )$. Since $u( \hat{Y}, P )$ is continuous in $\hat{Y}$, there exists some $\hat{Y}^{(m)}$ in a neighborhood of $\hat{Y}^*$  such that $ u^{(m-1)} \leq u( \hat{Y}^{(m)}, \tilde{P} )< u^*\leq u^{(m)}$, which contradicts eq.~\eqref{eqIteration}.

\vspace{0.2cm}

Next, we show that $\hat{Y}^* \in \argmax_{\hat{Y}\in \hat{\mathcal{I}}_0 } \min_{P\in\mathcal{C}_\delta} u( \hat{Y}^*, P )$. Again assume that there exists some $\tilde{Y}\in \hat{\mathcal{I}}_0$ such that $ \min_{ P \in \mathcal{C}_\delta } u( \tilde{Y}, P ) > \min_{ P\in \mathcal{C}_\delta } u( \hat{Y}^*, P ) $. By construction, we have $\min_{ P^{(i)} \in \mathcal{P}^{(m)} } u( \tilde{Y}, P^{(i)}) \leq \min_{ P^{(i)} \in \mathcal{P}^{(m)} } u( Y^{(m+1)}, P^{(i)}) = u^m$. Taking $m\rightarrow +\infty$ yields $\min_{ P \in \mathcal{C}_\delta } u( \tilde{Y}, P)\leq u^*$, a contradiction.
\end{proof}

%====================================================================================
%====================================================================================
%====================================================================================

%\vspace{0.2cm}
\newpage

\bibliographystyle{plain}
\bibliography{BiblioMario10}

\vspace{0.6cm}

\end{document}